\newif\ifsubmit     
\newif\ifllncs      
\newif\ifexabs      
\newif\ifblind      

\submittrue

\ifllncs
  \documentclass[runningheads,a4paper]{llncs}

\else
  \documentclass[letterpaper,11pt]{article}
  \usepackage[in]{fullpage}
  \usepackage{setspace}
  \usepackage[margin=1in]{geometry}
\fi

\usepackage{iftex}
\ifPDFTeX
  \usepackage[utf8]{inputenc}
  \usepackage[noTeX]{mmap}
  \usepackage[T1]{fontenc}
\fi
\ifLuaTeX
  \usepackage{luatex85}
  \usepackage[noTeX]{mmap}
\fi

\usepackage{mdframed}
\usepackage{amsmath}
\usepackage{amsfonts}
\usepackage{amssymb}
\usepackage{amsthm}
\usepackage{color}

\allowdisplaybreaks[3]

\usepackage{appendix}
\usepackage{algorithm}
\usepackage{algpseudocode}
\usepackage{braket}
\usepackage{comment}
\usepackage{url}
\usepackage{bbm}
\usepackage{multicol}
\usepackage{mdframed}
\usepackage{multirow}
\usepackage{booktabs}

\usepackage[bookmarks]{hyperref}
\usepackage[nameinlink]{cleveref}
\hypersetup{
  colorlinks,
  allcolors=blue
}

\ifllncs

  \spnewtheorem{claim}{Claim}{\bfseries}{\rmfamily}
  \crefname{claim}{claim}{claims}
  \Crefname{claim}{Claim}{Claims}
\else
  \newtheorem{theorem}{Theorem}[section]
  \newtheorem{definition}[theorem]{Definition}
  \newtheorem{remark}[theorem]{Remark}
  \newtheorem{lemma}[theorem]{Lemma}
  \newtheorem{corollary}[theorem]{Corollary}
  \newtheorem{proposition}[theorem]{Proposition}
  
  \newtheorem{claim}[theorem]{Claim}
  \newtheorem*{remark*}{Remark}
  
  \newtheorem{conjecture}[theorem]{Conjecture}

  \newtheorem*{theorem*}{Theorem}
  \newtheorem*{lemma*}{Lemma}

\usepackage[style=alphabetic,minalphanames=3,maxalphanames=4,maxnames=99,backref=true]{biblatex}

  \DeclareFieldFormat{eprint:iacr}{Cryptology ePrint Archive: \href{https://ia.cr/#1}{\texttt{#1}}}
  \DeclareFieldFormat{eprint:iacrarchive}{Cryptology ePrint Archive: \href{https://eprint.iacr.org/archive/#1}{\texttt{#1}}}
  \AtEveryBibitem{
    \clearlist{address}
    \clearfield{date}
    \clearfield{isbn}
    \clearfield{issn}
    \clearlist{location}
    \clearfield{month}
    \clearfield{series}
 
    \ifentrytype{book}{}{
      \clearlist{publisher}
      \clearname{editor}
    }
  }
\fi

\usepackage{appendix}
\usepackage{algorithm}
\usepackage{algpseudocode}
\usepackage{braket}
\usepackage{comment}
\usepackage{url}
\usepackage{multicol}
\usepackage{tikz}
\usetikzlibrary{arrows,automata,positioning}
\usepackage{caption}
\usepackage[caption=false]{subfig}
\usepackage[font=small,labelfont=bf]{caption}
\usepackage{comment}
\usepackage{pifont}
\usetikzlibrary{patterns}
\usetikzlibrary{arrows.meta}
\usepackage[T1]{fontenc}
\usepackage{textcomp}
\usepackage{dsfont}

\usepackage{enumitem}
\setlist[description]{noitemsep}
\setlist[enumerate]{noitemsep}
\setlist[itemize]{noitemsep}

\usepackage{soul, xcolor, xparse}
\makeatletter
  \ExplSyntaxOn
    \cs_new:Npn \white_text:n #1
    {
      \fp_set:Nn \l_tmpa_fp {#1 * .01}
      \llap{\textcolor{white}{\the\SOUL@syllable}\hspace{\fp_to_decimal:N \l_tmpa_fp em}}
      \llap{\textcolor{white}{\the\SOUL@syllable}\hspace{-\fp_to_decimal:N \l_tmpa_fp em}}
    }
    \NewDocumentCommand{\whiten}{ m }
    {
      \int_step_function:nnnN {1}{1}{#1} \white_text:n
    }
  \ExplSyntaxOff
  
  \NewDocumentCommand{ \varul }{ D<>{5} O{0.2ex} O{0.1ex} +m } {%
    \begingroup
    \setul{#2}{#3}%
    \def\SOUL@uleverysyllable{%
      \setbox0=\hbox{\the\SOUL@syllable}%
      \ifdim\dp0>\z@
      \SOUL@ulunderline{\phantom{\the\SOUL@syllable}}%
      \whiten{#1}%
      \llap{%
        \the\SOUL@syllable
        \SOUL@setkern\SOUL@charkern
      }%
      \else
      \SOUL@ulunderline{%
        \the\SOUL@syllable
        \SOUL@setkern\SOUL@charkern
      }%
      \fi}%
    \ul{#4}%
    \endgroup
  }
\makeatother

\ifsubmit
    \newcommand{\qipeng}[1]{}
    \newcommand{\zikuan}[1]{}
    \newcommand{\zihan}[1]{}
\else
    \newcommand{\qipeng}[1]{{\color{red} Qipeng: #1}}
    \newcommand{\zikuan}[1]{{\color{blue} Zikuan: #1}}
    \newcommand{\zihan}[1]{{\color{orange} Zihan: #1}}
\fi

\newcommand{\E}{\mathop{\mathbb{E}}}

\newcommand{\abs}[1]{\left|#1\right|}
\newcommand{\norm}[1]{\left\lVert#1\right\rVert}

\newcommand{\As}{\mathcal{A}}
\newcommand{\Bs}{\mathcal{B}}

\newcommand{\cA}{\mathcal{A}}

\newcommand{\Rs}{\mathcal{R}}

\newcommand{\gen}{{\sf Gen}}

\newcommand{\sto}{{\sf StO}}
\newcommand{\pho}{{\sf PhO}}
\newcommand{\csto}{{\sf CStO}}
\newcommand{\cpho}{{\sf CPhO}}

\newcommand{\stddecomp}{{\sf StdDecomp}}

\newcommand{\cD}{{\mathcal{D}}}

\newcommand{\eps}{\varepsilon}

\renewcommand{\kappa}{\ell}

\newcommand{\poly}{{\sf poly}}

\DeclareMathOperator{\Tr}{Tr}

\newcommand{\ketbra}[2]{\ket{#1}\!\bra{#2}}
\renewcommand{\braket}[2]{\langle #1\vert #2\rangle}
\renewcommand{\cal}[1]{\mathcal{#1}}

\newcommand{\brackets}[1]{\left(#1 \right)}

\newcommand{\bit}{{\{0, 1\}}}

\DeclareMathAlphabet\mathbfcal{OMS}{cmsy}{b}{n}

\DeclareFontFamily{U}{skulls}{}
\DeclareFontShape{U}{skulls}{m}{n}{ <-> skull }{}

\DeclareRobustCommand{\substack}[1]{\subarray{c}#1\endsubarray}

\renewcommand{\set}[1]{\left\{#1\right\}}

\newenvironment{boxfig}[2]{\begin{figure}[#1]\fbox{\begin{minipage}{0.97\linewidth}
                        \vspace{0.2em}
                        \makebox[0.025\linewidth]{}
                        \begin{minipage}{0.95\linewidth}
            {{
                        #2 }}
                        \end{minipage}
                        \vspace{0.2em}
                        \end{minipage}}
                        }
                        {\end{figure}}
\newcommand{\protocol}[4]{
\begin{boxfig}{ht}{\footnotesize 
\centering{\textbf{#1}}
    #4
\vspace{0.2em} } \caption{\label{#3} #2}
\end{boxfig}
}

\newcommand{\areg}{\mathbf{A}}
\newcommand{\breg}{\mathbf{B}}

\newcommand{\dreg}{\mathbf{D}}

\newcommand{\trace}[1]{\Tr\brackets{#1}}

\title{
    On the Need for (Quantum) Memory with Short Outputs
}

\ifllncs
\author{
}
\institute{
}
\else
\author{Zihan Hao  \\ \small{UC San Diego} \and Zikuan Huang  \\ \small{University of Cambridge} \and Qipeng Liu \\ \small{UC San Diego}}
\fi

\date{}

\begin{document}

\maketitle

\ifllncs 
\begin{abstract}
\end{abstract}
\else
\begin{abstract}


In this work, we establish the first separation between computation with bounded and unbounded space, for problems with short outputs (i.e., working memory can be exponentially larger than output size), both in the classical and the quantum setting. Towards that, we introduce a problem called \emph{nested collision finding}, and show that optimal query complexity cannot be achieved without exponential memory. 

Our result is based on a novel ``two-oracle recording'' technique, where one oracle ``records'' the computation's long outputs under the other oracle, effectively reducing the time-space trade-off for short-output problems to that of long-output problems. We believe this technique will be of independent interest for establishing time-space tradeoffs in other short-output settings.
\end{abstract}
\fi

\newpage

\tableofcontents

\newpage

\section{Introduction}

Does quantum speedup need large memory? This question has been extensively studied in different contexts, including advice-aided computation~\cite{nayebi2014quantum,chung2020tight} (preprocessing model), online learning~\cite{liu2023memory,chen2024optimal} (streaming model), quantum pebbling~\cite{blocki2022parallel} (pebbling model) and space-bounded computation with massive outputs~\cite{Hamoudi_2023,beame2024quantum} (general model). These works collectively demonstrate that quantum computation faces fundamental limitations when restricted to a small amount of (classical/quantum) memory, whether in the form of advice or working space.

One of the most important questions towards understanding the role of memory/space in quantum computation is the quantum collision finding problem. In this problem, one is given a function $f: [M] \to [N]$ (modeled as a random function) and the goal is to find a pair of distinct inputs $x \ne x'$ such that $f(x) = f(x')$.
Classically, an algorithm needs $O(\sqrt{N})$ time to achieve the goal, and it is tight due to the Birthday Paradox. The famous Pollard-Rho algorithm~\cite{pollard1978monte} can achieve this with $O(\sqrt{N})$ time and $\widetilde{O}(1)$ space. 
Quantumly, the BHT algorithm~\cite{BHT98} can solve this problem with both time and space at most $O(N^{1/3})$. Unlike its classical counterpart, we do not know if $N^{1/3}$ space is needed to achieve the optimal running time. More surprisingly, we do not even know whether quantum algorithms require any nontrivial amount of memory~\footnote{At least $\log N + \log M$ to allow classical/coherent evaluations of the function $f$.} to achieve an advantage over classical algorithms. This question of understanding the time-space tradeoffs for collision finding was also listed as one of the ten open challenges by Aaronson~\cite{aaronson2021open}, yet little progress has been made.

One key difficulty in answering the above questions, stems from the fact that proving time-space tradeoffs for problems with short outputs in the ``generic model'' is challenging even in the \emph{classical} setting. Here, the ``generic model'' refers to algorithms that perform computation on bounded memory, can discard some used memory, and introduce new memory on demand.
This was already pointed out by Dinur~\cite{dinur2020tight}. Since the introduction of time-space tradeoffs by Cobham~\cite{cobham1966recognition}, establishing strong time-space tradeoffs for short-output problems is open;~\cite{fortnow2005time} demonstrated some tradeoffs for short-output problems, but both time and space are only logarithmic in the input size, making the statement not strong. Other models (preprocessing, streaming and pebbling models) do not have a similar barrier and thus most results directly apply to problems with short outputs. On the other hand, existing trade-off results in the ``generic model'' only apply to large-output problems, where the required memory is at least as large as the output size, like \cite{dinur2020tight} in the classical setting and \cite{Hamoudi_2023} in the quantum setting for multi-collision finding. Given these challenges, we pose the following fundamental questions in this paper: 

\begin{center}
    {\it Can we establish strong classical and quantum time-space trade-offs for problems with \textbf{short outputs}?}
\end{center}  

More precisely we ask: 
\begin{center}
    {\it If a problem has a short output ($\poly(n)$ bits), does an optimal (quantum) algorithm need $\exp(n)$ space?}
\end{center}

Short outputs are not an edge case: decision problems and most search problems in the (Quantum) Random Oracle Model all fall into this regime. Technically, all previous methods in the generic model (for long outputs) seem to require a certain progress measure, which is no longer well defined for short-output problems. Even partial results --- say, ruling out $T\cdot S=o(N^{\alpha})$ for some constant $\alpha>0$ on natural short-output tasks --- would represent a qualitative advance over the current logarithmic barrier and bring us materially closer to Aaronson’s question.


\subsection{Our results}

In this work, we introduce the \emph{$\ell$-Nested Collision Finding} problem for a constant \( \ell > 1 \).
Given two random oracles \( H: [M] \to [N] \) and \( G: [M]^\ell \to [N_0] \), the goal is to find two distinct ordered tuples \( (x_1, x_2, \ldots, x_\ell) \neq (x'_1, x'_2, \ldots, x'_\ell) \in [M]^\ell \) such that:  
\begin{itemize}
    \item \( \sum_{i=1}^\ell H(x_i) = \sum_{i=1}^\ell H(x'_i) = 0\), \quad (here $\sum$ denotes summation mod $N$), and,   
    \item \( G(x_1, \ldots, x_\ell) = G(x'_1, \ldots, x'_\ell) \).  
\end{itemize} 
We will sometimes call the problem ``finding a nested collision pair''.
We can strengthen the problem so that instead of asking the two $\ell$-sums equal to $0$, we require the two $\ell$-sums to be any fixed $y$ or uniformly random $y$ that is presented in the input. This strengthening will not change any lemma below. For simplicity, we will just define the problem with the $\ell$-sums being $0$. 

\medskip

Note that the problem has a short output; the output is only of size poly-logarithmic in the size of the input (if we treat as input the truth table of the oracles $H, G$).
At a high level, the goal is to find two \(\ell\)-tuples whose images under \( H \) have the same sum $0$ and whose images under \( G \) are identical.  
Although this formulation involves two random oracles, a single random oracle with a sufficiently large domain can serve as both \( H \) and \( G \). However, for clarity, we will continue to work with two separate oracles throughout this paper.  


\medskip

Our main contribution is to show that the $\ell$-Nested Collision Finding problem has time-space tradeoffs in both the classical and the quantum setting. 

\subsubsection*{Quantum algorithms for short-output problems require space.}
\begin{theorem}[Main quantum theorem] \label{thm:quantum_informal}
    For $\ell\geq 2$, any quantum algorithm that solves the $\ell$-Nested Collision Finding problem with constant probability with space $S$ and $T$ oracle queries must satisfy $S^{\frac{\ell+1}{2}}T=\Omega\brackets{N^{\frac{1}{2(\ell+1)}}N_0^{\frac{1}{4}}}$ if $S=\Omega(\log N)$ and $S=O\brackets{N^{\frac{1}{(\ell + 1)(2 \ell + 1)}}}$. 
\end{theorem}
The main theorem has both parameters $\ell$ and $N_0$ that are yet to be determined. By comparing against the best algorithm we constructed (see \Cref{thm:quantum-algorithm} for the algorithm), the $\ell$-Nested Collision Finding problem establishes a separation between space-bounded and unbounded quantum computation, provided that $\ell \geq 4$ .
For clarity, we select the most straightforward parameter choices and state a specific instance of our result \Cref{thm:quantum_informal} when $\ell = 4$ and $N_0 = N$.
\begin{theorem}[Separation between space-bounded/unbounded quantum computation]\label{thm:quantum_separation_informal}
    Any quantum algorithm that solves $4$-Nested Collision Finding problem with oracles $H:[M] \to [N]$ and $G:[M^4] \to [N]$, with space $S$ and $T$ oracle queries, must satisfy $S^{5/2} T = \Omega(N^{0.35})$ if $S=\Omega(\log N)$ and $S=O(N^{1/45})$. 

    As a special case, when $S = \widetilde{\Theta}(1)$, the quantum algorithm requires at least $\widetilde{\Omega}(N^{0.35})$ quantum queries to solve the $4$-Nested Collision Finding problem with a constant probability. Also, as long as $S = o(N^{1/150})$, the optimal query complexity $O(N^{1/3})$ can not be achieved.
    On the other hand, the optimal quantum algorithm only requires $O(N^{1/3})$ quantum queries (as well as $O(N^{1/3})$ memory).
\end{theorem}

\subsubsection*{Classical algorithms for short-output problems require space.}
We have a similar theorem in the classical setting. 

\begin{theorem}[Main classical theorem]\label{thm:classical_informal}
    For $\ell\geq 2$, any classical algorithm that solves the $\ell$-Nested Collision Finding problem with constant probability with space $S$ and $T$ oracle queries must satisfy $S^{\frac{\ell^2-1}{2\ell}}T=\Omega\brackets{N^{\frac{1}{2\ell}}N_0^{\frac{1}{2}}}$ if $S=\Omega(\log N)$ and $S=O\brackets{N^{\frac{1}{\ell^2-1}}}$.
\end{theorem}

On the other hand, we construct the following classical algorithm. For any $\ell\geq 1$, there exists a classical algorithm that uses $O\brackets{N^\frac{1}{\ell}N_0^\frac{1}{2\ell}}$ queries that solve the $\ell$-Nested Collision Finding Problem when $N_0=O\brackets{N^\frac{2}{\ell-1}}$ with constant probability. By setting $\ell = 2$ and $N_0 = N^2$, we have the following corollary:
\begin{theorem}[Separation between space-bounded/unbounded classical computation]\label{thm:classical_separation_informal}
    Any classical algorithm that solves $2$-Nested Collision Finding problem with oracles $H:[M] \to [N]$ and $G:[M^2] \to [N^2]$, with space $S$ and $T$ oracle queries, must satisfy $S^{3/4} T = \Omega(N^{5/4})$ if $S = \Omega(\log N) $ and $S = O(N^{1/3})$. 

    As a special case, when $S = \widetilde{\Theta}(1)$, the classical algorithm requires at least $\widetilde{\Omega}(N^{5/4})$ classical queries to solve the same problem with a constant probability. On the other hand, the optimal classical algorithm only requires $O(N)$ classical queries.
\end{theorem}
Again, this is the first example in the classical setting, where a short-output problem requires exponential space to reach an optimal time.


\subsection{Discussions and open questions}

We discuss applications and directions related to our new framework.

\paragraph{Proof of space.}
By scaling down the parameter $N$, the classical/quantum tradeoffs yield a non-interactive proof-of-space in the random oracle model. The prover computes and returns a collision pair for the Nested Collision Finding problem as the proof; any prover lacking sufficient space incurs a longer running time to find such a collision. The main drawback is that the time gap between space-bounded and unbounded provers is only polynomial. A potential remedy could be to define a task that further nests the Nested Collision Finding problem, achieving an exponential separation between space-bounded and unbounded algorithms.

\paragraph{Distinguishing between quantum and classical space?}
All our results treat space as a whole, i.e. we do not distinguish between quantum and classical space. This technical limitation also appears in most prior works. Currently, there is no known general method for distinguishing quantum and classical memory in this context. The only known approaches apply to the streaming setting, such as the techniques developed in~\cite{liu2023memory}, are based on unproven conjectures~\cite{liu2025qma}, or are tailored to specific constructions~\cite{bostanci2025separating,bostanci2026separating}.


\subsection{Related work}
A recent work \cite{GaoJiXie2026} constructs a short-output \textbf{decision} problem whose classical query complexity (with sufficient space) is strictly smaller than its bounded-space quantum query complexity. This can be read as: the quantum advantage requires space. The construction in that paper is based on our Nested Collision Finding problem but the outer problem is redefined and the analyzed in another way. If we can further extend it to a hash function satisfies this property, then in virtually any hash-based cryptographic protocol, one could neutralize quantum advantages in attacks against the hash, as having access to exponential space is unrealistic for practical purposes.

Our classical construction is also related to Dinur's~\cite{dinur2020tight} earlier post-filtering approach for short-output lower bounds. Conceptually, the classical part of our argument can be viewed as lifting that post-filtering viewpoint into a standard random-oracle-style formulation by introducing a separate helper oracle $G$ that records the relevant tuples. 

\paragraph{Acknowledgment.} We thank Itai Dinur for pointing out that an earlier version did not sufficiently formalize the lazy-sampling extraction step in the classical proof, and for drawing our attention to the connection with his earlier post-filtering construction for classical short-output lower bounds.

\section{Technical overview}

We focus exclusively on lower bounds in this overview. The algorithm is relatively straightforward and is presented in~\Cref{thm:quantum-algorithm}.

\subsection*{Time-space tradeoffs in massive-output problems.}

We start by reviewing the ideas behind time-space tradeoffs in massive-output problems in the general computation model and the difficulty of obtaining a tradeoff in the short-output setting. We consider the problem of finding multiple two-collisions in~\cite{Hamoudi_2023}. The problem is as follows: given a random oracle $H: [M] \to [N]$, to find $K$ pairs of disjoint inputs $(x_1, x_2), \ldots, (x_{2K-1}, x_{2K})$ such that $H(x_{2 i - 1}) = H(x_{2i})$ for all $i = 1, 2, \ldots, K$, which was first studied by~\cite{dinur2020tight} in the classical setting and by~\cite{Hamoudi_2023} in the quantum setting. 

To achieve time-space tradeoffs in the general computation model, most approaches (if not all) rely on the time-segmentation method~\cite{borodin1993time,KlauckSW07}. Roughly speaking (see \Cref{fig:time-segment}), instead of enforcing a strict space bound of \( S \) at every step of an algorithm’s execution, this method divides time into \( L \) segments. The algorithm is then relaxed so that the space usage is only constrained to \( S \) when transitioning between segments, while no space bound is imposed within each segment. In this model, the algorithm is allowed to output the answers on the fly; i.e., for the multi-collision problem, at the end of each segment, it can output all the pairs of collisions it finds within that segment.

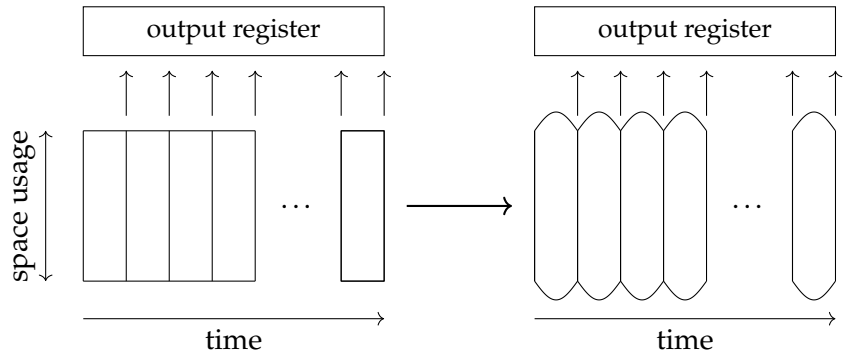
\begin{figure}[ht]
    \centering
\begin{tikzpicture}
    \def\width{4}
    \def\height{2}
    \def\n{7} 
    \def\shiftX{6} 
    \def\arcHeight{0.25} 
    \def\boxHeight{0.65} 
    \def\boxWidth{4} 
    \def\boxShiftY{1} 

    \foreach \i in {0,1,2,3,4,6,7} { 
        \draw (\i*\width/\n, 0) -- (\i*\width/\n, \height);
    }
    \draw (0,0) -- (4*\width/\n,0); 
    \draw (0,\height) -- (4*\width/\n,\height); 
    \draw (6*\width/\n,0) -- (\width,0); 
    \draw (6*\width/\n,\height) -- (\width,\height); 
    \node at (5*\width/\n, \height/2) {\dots};
    \draw (6*\width/\n, 0) rectangle (\width, \height);
    \draw[<->] (-0.5, 0) -- (-0.5, \height) node[midway, rotate=90, above] {space usage};
    \draw[->] (0, -0.5) -- (\width, -0.5) node[midway, below] {time};

    \foreach \i in {0,1,2,3,4,6,7} { 
        \draw (\i*\width/\n+\shiftX, 0) -- (\i*\width/\n+\shiftX, \height);
    }
    
    \foreach \i in {0,1,2,3,6} { 
        \draw (\i*\width/\n+\shiftX, 0) -- (\i*\width/\n+\shiftX, \height);
        \draw[domain=0:180, smooth, variable=\t] plot ({\shiftX + \i*\width/\n + \width/\n * \t / 180}, {\height + \arcHeight * sin(\t)});   
        \draw[domain=0:180, smooth, variable=\t] plot ({\shiftX + \i*\width/\n + \width/\n * \t / 180}, {-\arcHeight * sin(\t)});
    }
    \node at (\shiftX+5*\width/\n, \height/2) {\dots};
    
    \draw[->, thick] (\width+0.3, \height/2) -- (\shiftX-0.3, \height/2);
    \draw[->] (\shiftX, -0.5) -- (\shiftX+\width, -0.5) node[midway, below] {time};

    \draw (\width/2-\boxWidth/2, \height+\boxShiftY) rectangle (\width/2+\boxWidth/2, \height+\boxShiftY+\boxHeight);
    \draw (\shiftX+\width/2-\boxWidth/2, \height+\boxShiftY) rectangle (\shiftX+\width/2+\boxWidth/2, \height+\boxShiftY+\boxHeight);
    \node at (\width/2, \height+\boxShiftY+\boxHeight/2) {\small output register};
    \node at (\shiftX+\width/2, \height+\boxShiftY+\boxHeight/2) {\small output register};

    \foreach \i in {1,2,3,4,6,7} { 
        \draw[->] (\i*\width/\n, \height+0.2) -- (\i*\width/\n, \height+\boxShiftY-0.2);
        \draw[->] (\i*\width/\n+\shiftX, \height+0.2) -- (\i*\width/\n+\shiftX, \height+\boxShiftY-0.2);
    }

\end{tikzpicture}
    \caption{Conversion from a space-bounded computation to a time-segmented computation with space constraints applied only during transitions.}
    \label{fig:time-segment}
\end{figure}

Therefore, if an algorithm with space constraint \( S \) in the time-segmentation framework can find more than \( K \) pairs of disjoint collisions, it must be able to find at least \( K/L \) pairs of disjoint collisions within a single segment. The key idea in time-segmentation is to track progress toward finding more than \( K \) disjoint collisions --- i.e., how many disjoint collisions can be found in each segment. By appropriately choosing parameters, \cite{Hamoudi_2023} proved that any quantum algorithm making \( T/L \) queries cannot find more than \( K/L \) disjoint collisions, except with probability exponentially small in \( S \).  

They further extended this argument to non-uniform algorithms with arbitrary \( S \)-qubit memory, modeling the memory passed from one segment to the next. By setting \( K/L \approx \Theta(S) \), the probability of a uniform algorithm finding \( K/L \) disjoint collisions is on the order of \( 2^{-\Theta(S)} \). Thus, an $S$-qubit advice state can be removed by guessing it at random, introducing a multiplicative factor of \( 2^S \). This factor is completely absorbed into \( 2^{-\Theta(S)} \), preserving the asymptotic order.

This is the general approach used to establish time-space tradeoffs in multi-collision finding and other time-segmentation methods. However, several inherent barriers arise when dealing with short-output problems.  
First, the notion of a progress measure becomes unclear. Since the algorithm is only required to produce a short output, it does not necessarily follow a structured computational pattern imposed by a predefined progress measure.  
Second, as the \( S \)-qubit memory is passed between segments, the only general approach to handling this memory is to guess and remove it. This introduces a multiplicative factor of \( 2^S \) in the overall success probability for each segment. To ensure that this factor does not render the bound meaningless or trivial, the success probability within each segment must be at least exponentially small in \( S \).  
It appears that only massive-output problems can safely absorb this factor without changing the asymptotic order of the probability.

\subsection*{Nested Collision Finding problem.}

From the above discussion, finding an appropriate progress measure for a short-output problem is crucial for establishing a time-space tradeoff. Our starting point is the problem of finding a \emph{3-collision}, where the goal is to search for three distinct inputs that all map to the same output. 
This problem appears to be a promising candidate. Intuitively, an algorithm that finds a 3-collision must first identify multiple 2-collisions, which, as shown in the previous section, already admits a time-space tradeoff. Therefore, if we can transform an algorithm that finds a 3-collision into one that efficiently finds multiple 2-collisions, the 3-collision problem will inherently exhibit a non-trivial time-space tradeoff.

The above intuition serves as the key insight leading to our final construction and theorems. However, reducing the problem of finding multiple 2-collisions to that of finding a 3-collision has some technical difficulties, due to the following challenges:  
\begin{enumerate}
    \item \textbf{Preserving time and space:} The reduction must be both time- and space-efficient, as we aim to maintain the time-space tradeoff properties throughout the reduction.  
    \item \textbf{Applicability to general algorithms:} The reduction may not hold for arbitrary (non-optimal) algorithms. In particular, certain contrived algorithms may find a 3-collision without necessarily discovering many 2-collisions along the way. For instance, consider an algorithm that searches specifically for three distinct inputs mapping to zero and keeps only the preimages of zero during its execution. While this algorithm eventually finds a 3-collision with polylog space, extracting multiple 2-collisions from its execution could be challenging.  
\end{enumerate}

Although we are unable to apply this idea directly to the 3-collision finding problem, we implement the previous intuition through a different construction in the random oracle model.
Recall the \emph{$\ell$-Nested Collision Finding} problem (see \Cref{sec:nested_col}) for a constant \( \ell > 1 \). Given two random oracles \( H: [M] \to [N] \) and \( G: [M]^\ell \to [N_0] \), the goal is to find two distinct ordered tuples \( (x_1, x_2, \ldots, x_\ell) \neq (x'_1, x'_2, \ldots, x'_\ell) \in [M]^\ell \) such that:  
\begin{itemize}
    \item \( \sum_{i=1}^\ell H(x_i) = \sum_{i=1}^\ell H(x'_i) = 0\), and  \( G(x_1, \ldots, x_\ell) = G(x'_1, \ldots, x'_\ell) \).  
\end{itemize} 

Following the intuition from the $3$-collision finding problem, we aim to establish the following two parts:  
\begin{itemize}
    \item \textbf{Part 1.} Any algorithm that finds a collision pair for the Nested Collision Finding problem can be transformed (while preserving both time and space complexity) into an algorithm that outputs multiple $\ell$-tuples $(x_{1,1}, \ldots, x_{1, \ell})$, $\ldots$, $(x_{K, 1}, \ldots, x_{K, \ell})$ satisfying the same sum condition:  
        \begin{align*}
            \sum_{j=1}^\ell H(x_{1, j}) = \sum_{j=1}^\ell H(x_{2, j}) = \cdots = \sum_{j=1}^\ell H(x_{K, j}) = 0.
        \end{align*}
    We refer to this problem as ``finding $K$ $\ell$-tuples with the same sum''.
    
    \item \textbf{Part 2.} Establishing time-space tradeoffs for ``finding $K$ $\ell$-tuples with the same sum'' in both classical and quantum settings.
\end{itemize}

We elaborate on both parts in the following sections. The roadmap is given as in \Cref{fig:outline}.

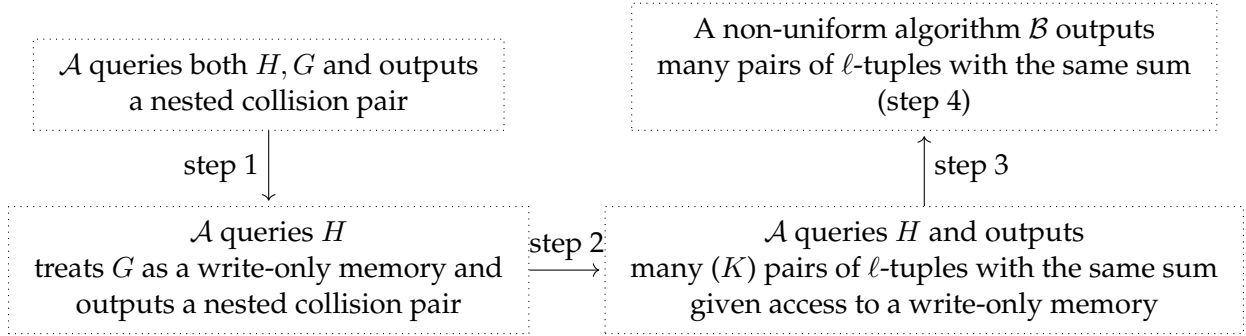
\begin{figure}[!hbt]
\centering
			\begin{tikzpicture}[scale=0.9,->,shorten >=2pt]
			    \node (first) [rectangle, dotted, draw=black]  {\begin{tabular}{c} $\As$ queries both $H, G$ and outputs \\  a nested collision pair \end{tabular}};
			    \node (second) [rectangle, dotted, draw=black, below=of first] {\begin{tabular}{c} $\As$ queries $H$ \\ treats $G$ as a write-only memory and \\ outputs a nested collision pair \end{tabular}};
			    \node (third) [rectangle, dotted, draw=black, right=of second]{\begin{tabular}{c} $\As$ queries $H$ and outputs \\ many ($K$) pairs of $\ell$-tuples with the same sum \\ given access to a write-only memory \end{tabular}};
			    \draw[->] (first) -- (second) node [midway, left] {step 1};
			    \draw[->] (second) -- (third) node [midway, above] {step 2};
			    \node (fourth) [rectangle, dotted, draw=black, above=of third] {\begin{tabular}{c} A non-uniform algorithm $\Bs$ outputs \\  many pairs of $\ell$-tuples with the same sum  \\ (step 4) \end{tabular}};
			    \draw[->] (third) -- (fourth) node [midway, right] {step 3};
			\end{tikzpicture}
        \caption{A roadmap in establishing our main results.}
        \label{fig:outline}
\end{figure}

\subsection*{``Two-oracle recording'' technique that preserve time and space (step 1).}

We begin by considering the classical setting. Let \(\mathcal{A}\) be a classical query algorithm with oracle access to both \(H\) and \(G\). Typically, \(G\) is viewed as a pre-sampled random function. However, here we adopt an alternative perspective by regarding \(G\) as a lazily sampled truth table.
In the standard formulation, a random oracle \(G: [M]^\ell \to [N_0]\) is chosen uniformly at random from the set of all functions. In contrast, the lazy sampling approach (often referred to as the ``principle of deferred decisions'') begins with \(G\) initialized to \(\{\bot\}^{\otimes M^\ell}\), meaning that every output is initially undefined. When a query \(x\) is made, if \(G(x) = \bot\) (i.e., if its value has not yet been determined), we sample a uniformly random \(y \gets [N_0]\) and then set \(G(x) := y\).
Since these two views are statistically identical, we can freely switch between them in our analysis.

Our key observation is that, in the lazy sampling view, $G$ can be interpreted as an exponentially large memory that records the computation of $\mathcal{A}$. This allows us to view the computation of $\mathcal{A}$ using the oracles $H$ and $G$ in an alternative way: $\mathcal{A}$ interacts with the oracle $H$ as usual while treating $G$ as a memory that supports a specific form of write-only operation:  
\begin{itemize}
    \item $\mathcal{A}$ can query $H$ as a standard random oracle.
    \item $\mathcal{A}$ can write to the memory $G$ by querying on $x$. Initially, all entries of $G$ are set to $\bot$, indicating that all memory cells are empty. When $\mathcal{A}$ queries $G(x)$, if $G(x) = \bot$, a random $y \gets [N_0]$ is sampled, stored in the $x$-th memory block, and returned to $\mathcal{A}$. Otherwise, the stored value is simply returned.
\end{itemize}  

Although $G$ returns $y$ to the algorithm, $y$ is chosen randomly and does not depend on prior computation, so $G$ can still be viewed as write-only memory. It is straightforward to see that this write-only memory model is equivalent to the lazy sampling view of a random oracle.  

By modeling a random oracle as a write-only memory, the computation with oracle access to both $H$ and $G$ can now be interpreted as a computation with oracle access to $H$ and a write-only memory $G$. This conversion is illustrated in \Cref{fig:G_as_writeonly}. We refer to this perspective as the ``two-oracle recording'' technique, which will serve as the foundation for our proofs.  

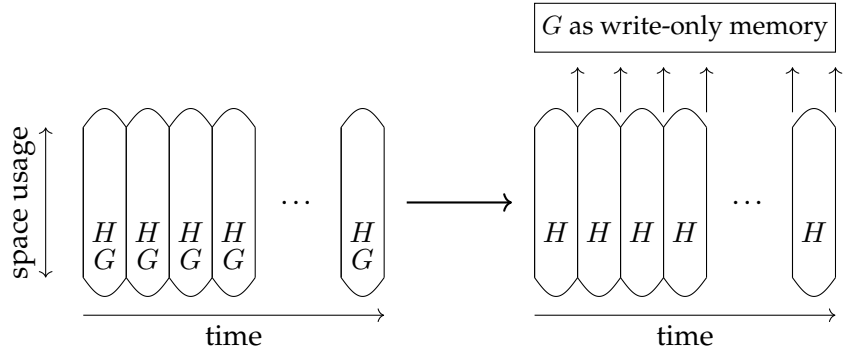
\begin{figure}[ht]
    \centering
\begin{tikzpicture}
    \def\width{4}
    \def\height{2}
    \def\n{7} 
    \def\shiftX{6} 
    \def\arcHeight{0.25} 
    \def\boxHeight{0.65} 
    \def\boxWidth{4} 
    \def\boxShiftY{1} 

    \foreach \i in {0,1,2,3,4,6,7} { 
        \draw (\i*\width/\n, 0) -- (\i*\width/\n, \height);
    }
    \foreach \i in {0,1,2,3,6} {
        \node at (\i*\width/\n+0.5*\width/\n, 0.25) {$G$};
        \node at (\i*\width/\n+0.5*\width/\n, 0.6) {$H$};
    }
    
    \foreach \i in {0,1,2,3,6} { 
        \draw (\i*\width/\n, 0) -- (\i*\width/\n, \height);
        \draw[domain=0:180, smooth, variable=\t] plot ({\i*\width/\n + \width/\n * \t / 180}, {\height + \arcHeight * sin(\t)});   
        \draw[domain=0:180, smooth, variable=\t] plot ({\i*\width/\n + \width/\n * \t / 180}, {-\arcHeight * sin(\t)});
    }
    \node at (5*\width/\n, \height/2) {\dots};
    \draw[<->] (-0.5, 0) -- (-0.5, \height) node[midway, rotate=90, above] {space usage};
    \draw[->] (0, -0.5) -- (\width, -0.5) node[midway, below] {time};

    \foreach \i in {0,1,2,3,4,6,7} { 
        \draw (\i*\width/\n+\shiftX, 0) -- (\i*\width/\n+\shiftX, \height);
    }
    \foreach \i in {0,1,2,3,6} {
        \node at (\i*\width/\n+0.5*\width/\n+\shiftX, 0.6) {$H$};
    }
    
    \foreach \i in {0,1,2,3,6} { 
        \draw (\i*\width/\n+\shiftX, 0) -- (\i*\width/\n+\shiftX, \height);
        \draw[domain=0:180, smooth, variable=\t] plot ({\shiftX + \i*\width/\n + \width/\n * \t / 180}, {\height + \arcHeight * sin(\t)});   
        \draw[domain=0:180, smooth, variable=\t] plot ({\shiftX + \i*\width/\n + \width/\n * \t / 180}, {-\arcHeight * sin(\t)});
    }
    \node at (\shiftX+5*\width/\n, \height/2) {\dots};
    
    \draw[->, thick] (\width+0.3, \height/2) -- (\shiftX-0.3, \height/2);
    \draw[->] (\shiftX, -0.5) -- (\shiftX+\width, -0.5) node[midway, below] {time};

    \draw (\shiftX+\width/2-\boxWidth/2, \height+\boxShiftY) rectangle (\shiftX+\width/2+\boxWidth/2, \height+\boxShiftY+\boxHeight);
    \node at (\shiftX+\width/2, \height+\boxShiftY+\boxHeight/2) {\small $G$ as write-only memory};

    \foreach \i in {1,2,3,4,6,7} {
        \draw[->] (\i*\width/\n+\shiftX, \height+0.2) -- (\i*\width/\n+\shiftX, \height+\boxShiftY-0.2);
    }

\end{tikzpicture}
    \caption{The LHS represents computation with oracle access to both $H, G$. The RHS represents computation with only oracle access to $H$, whereas having access to an (exponentially large) write-only memory $G$.}
    \label{fig:G_as_writeonly}
\end{figure}

\subsection*{The quantum ``two-oracle recording'' technique (step 1, quantum).} 

The above idea extends to the quantum case with additional effort. \cite{zhandry2019record} shows a formulation of lazy sampling for quantum accessible random oracle, called ``compressed oracle''. Since quantum queries can be made in superposition, a classical database/memory is no longer feasible to track all the information. Zhandry showed that, the notion of a ``database'' is still meaningful if the ``database'' itself is also stored in superposition. A quantum-query algorithm interacting with a quantum-accessible random oracle can be equivalently simulated (informally) as follows:
\begin{itemize}
    \item The database register is initialized as  $\ket \bot^{\otimes M^\ell}$: the algorithm has not yet queried anything.
    \item When the algorithm makes a quantum query, the database register gets updated in superposition: for a query $\ket x$ and a database $\ket D$, the simulator looks up $x$ in $D$.
    If the $x$-th entry is $\ket \bot$, it changes this entry to an equal superposition $\sum_y \ket y$ (up to normalization). 

    It then returns the $x$-th entry in the updated $D$ also in superposition.
\end{itemize}
Any quantum computation of $\cA$ using oracles $H, G$ can be alternatively viewed as $\cA$ using oracle $H$ and has access to a memory (the random oracle $G$) that supports coherent write-only operation, as described above. The above description provides a high-level idea, albeit slightly inaccurate. Since a quantum algorithm can forget  previously acquired information, a formal ``compressed oracle'' necessitates an additional step to perform these forgetting operations. We leave the details in \Cref{sec:compress_oracle}.


\subsection*{Relating to massive-output problems (step 2).}

Assume that an algorithm $\mathcal{A}$ finds a pair of nested collisions $(x_1, \ldots, x_\ell) \neq (x'_1, \ldots, x'_\ell)$. By the definition, we have  
\[
\sum_{i} H(x_i) = \sum_{i} H(x'_i) = 0
\]
and  
\[
G(x_1, \ldots, x_\ell) = G(x'_1, \ldots, x'_\ell).
\]

\medskip

\emph{The classical case.} Since $\mathcal{A}$ finds a collision under $G$, a standard lazy sampling argument implies that there exists some value $z$ such that at least $\Omega(N_0/T)$ distinct $\ell$-tuples queried to $G$ sum to $z$, where $T$ is the number of queries. This follows from the observation that if at any moment of the computation, the number of such tuples is bounded by $K$, then the probability of finding a collision in a single step is at most $K/N_0$. Given $T$ queries, the total success probability is thus at most $(T K) / N_0$\footnote{In the actual proof, we give an even tighter bound $K^2 / N_0$, since $K \leq T$.}.

Therefore, if $\mathcal{A}$ successfully finds a nested collision with a constant probability, then by interpreting $G$ as write-only memory (via the ``two-oracle recording'' technique), $G$ must store at least $K = \Omega(N_0 / T)$ pairs of $\ell$-tuples of the same sum. 

In the actual proof, we further strengthen this argument by showing that the claim holds even when $K$ represents the expected number of such $\ell$-tuples of the same sum stored in the lazily sampled database, where the expectation is taken uniformly at random over all possible $G$ across all time steps. Thus, to extract $K$ pairs of $\ell$-tuples of the same sum, one can simply run $\mathcal{A}$ and halt at a uniformly random time, outputting all pairs stored in $G$ at that moment.

\medskip

\emph{The quantum case.} The key difference in the quantum setting is that the random oracle $G$ must be simulated using the compressed oracle formulation. We establish the following quantum analog: 
\begin{lemma}[Informal]
    If, at any moment during the computation, the compressed oracle database has a bound of at most $K$ on the number of $\ell$-tuples with the same sum, then given $T$ queries, the total success probability is at most $(T^2 K) / N_0$. By setting the probability to be a constant, we should have $K = \Omega(N_0/T^2)$. 
\end{lemma}
The proof follows from a more refined analysis using the compressed oracle technique~\cite{zhandry2019record} and is presented in \Cref{sec:many_entries_database}. At a high level, each query increases the amplitude (rather than the probability) of the final outcome by at most $O(\sqrt{K/N_0})$. Consequently, after $T$ queries, the total amplitude is at most $O(T \sqrt{K/N_0})$, leading to a success probability bound of $O(T^2 K / N_0)$.  
We further strengthen this argument by showing that the claim remains true even when $K$ represents the expected number of such $\ell$-tuples, measured at a uniformly random time step. Thus, to extract $K = \Omega(N_0/T^2)$ pairs of $\ell$-tuples, one can simply run $\mathcal{A}$, halt at a uniformly random time, measure the database, and output all pairs obtained from the measurement outcomes.

\subsection*{Bounding the success probability in the massive-output case, with write-only memory (step 3.1).} 

In steps 1 and 2, we reduce an algorithm with space $S$ and query complexity $T$ to another algorithm that, with the same time and space complexity, outputs $K = \Omega(N_0 / T^2)$ $\ell$-tuples with the same sum. The next objective is to establish a time-space tradeoff for this large-output problem, which will then yield a corresponding tradeoff for the nested collision-finding problem. Similar bounds have been derived in~\cite{Hamoudi_2023} for multi-collision finding. However, while the setting is similar, our case presents additional challenges: in~\cite{Hamoudi_2023}, the quantum algorithm can coherently write to a quantum output register, with each write operation stored in a new register\footnote{More formally, each write operation coherently appends the output to the output register.}. In contrast, our setting requires each output $x$ to be stored in a designated register (the $x$-th register), introducing potential interference/entanglement between the algorithm and the output register.

The approach in~\cite{Hamoudi_2023} (building on earlier works such as~\cite{borodin1993time, KlauckSW07, dinur2020tight}) relies on the time-segmentation method. This technique partitions an algorithm into $L$ segments, where each segment is a non-uniform algorithm with $S$ (qu)bits and at most $T/L$ queries, with the goal of outputting $K/L$ $\ell$-tuples with the same sum. If one can show that such an algorithm, except with negligible probability, fails to output more than $K/L$ tuples, then by a standard union bound, the original algorithm cannot produce more than $K$ tuples (except with small probability).

This approach fails when an algorithm interacts with a memory where two distinct write-only operations can be applied to the same memory cell. Since we model the random oracle $G$ as a write-only quantum memory using the ``two-oracle recording'' technique, we must consider this type of memory rather than the one analyzed in~\cite{Hamoudi_2023}. 

To explain the challenges we face, we consider the following simplified question (we call it ``Sparse-Write Tail-Bound Composition Problem''): 
\begin{itemize}
    \item Consider two algorithms $\As, \Bs$, one working on registers $\areg\dreg$ and the other working on registers $\breg\dreg$. All registers are initiated to zeros; specifically, $D$ consists of $|0^{M^\ell}\rangle$. 
    \item $\As$ can apply local unitary on $\areg$ as well as controlled flip operation on $\areg\dreg$ jointly: $\ket x\ket y \to \ket x \ket {y \oplus e_x}$ where $e_x$ has only one $1$ at its $x$-th location.
    At the end of $\As$, the probability that measuring $\dreg$ yields a string $y$ with Hamming weight more than $C$ is at most $\epsilon$.
    \item Similar for $\Bs$.
    \item Then the question is, can we show that if both $\As, \Bs$ run with the same $\dreg$, the probability that measuring $\dreg$ yields a string $y$ with hamming weight more than $2C$ is bounded by $\poly(\epsilon)$?
\end{itemize}
This question is straightforward to prove in the classical setting by incurring a union bound. However, in the quantum case, the quantum union bound~\cite{gao2015quantum,khabbazi2019union,o2022quantum} does not apply due to the potential interference between $\As$ and $\Bs$'s writing behaviors. Moreover, if the above question is true, it will imply another seemingly irrelevant conjecture in the context of boolean functions:
\begin{conjecture}[Fourier Tail-Bound Composition Conjecture]
    There exists a polynomial $p(\cdot)$ such that, for all $\eps$ and all functions $f, g: \{-1, 1\}^n \to \{-1,1\}$ satisfying 
    \[
        W^{> C}[f] \leq \epsilon\text{  and  } W^{> C}[g] \leq \epsilon
    \]
    i.e., the Fourier weights of both functions at degrees above $C$ are small. We have
    \[
        W^{> 2C}[f\cdot g] \leq p(\epsilon).
    \]
\end{conjecture}
This implication arises from considering all $\As$ of the following form defined by $f$, where we take a Hadamard on $\dreg$ to make it the standard oracle register:
\[
    \ket{0}_{\areg}\ket{H}_{\dreg}\rightarrow (-1)^{f\brackets{(-1)^{H(0^n)},\cdots,(-1)^{H(1^n)}}}\ket{0}_{\breg}\ket{H}_{\dreg}.
\]
And similar for all $\Bs$ defined by $g$,
\[
    \ket{0}_{\breg}\ket{H}_{\dreg}\rightarrow (-1)^{g\brackets{(-1)^{H(0^n)},\cdots,(-1)^{H(1^n)}}}\ket{0}_{\breg}\ket{H}_{\dreg}.
\]
As far as we know, this simple question in boolean function analysis has not been studied yet; all standard tools seem not to work for this problem. Therefore, we do not know if the claim is true; let alone for the original `Sparse-Write Tail-Bound Composition Problem'. To show that this problem may not be trivial we provide an example in \Cref{sec:tail_bound_example}. For which we show that it is false to claim $W^{> 2C}[f\cdot g] \leq 4\eps$. Here this $4\eps$ is in fact the bound from adding amplitudes $(\sqrt{\eps}+\sqrt{\eps})^2$ which many may think that this is the correct bound at first glance. \textbf{This example and its proof are generated by GPT 5.6 Sol. The human authors take full responsibility for this example and its proof.}

\subsection*{Bounding the expectation in the massive-output case, with write-only memory (step 3.2).} 

To address this, instead of looking at the composition of tail bounds, we look at the composition of expectations. Namely, we will examine for every $\vec{x} = (x_1, \ldots, x_\ell)$ of interest, what is the probability/expectation $w_{\vec{x}}$ that the write-only memory $G$ has recorded $\vec{x}$. By the linearity of expectation, the final expected number is the summation of the expectation of all inputs of interest. 
We define the following game:
\begin{itemize}
    \item Let $H, G$ be two random oracles.
    \item Let $\vec{x}$ be a uniformly random input of interest (i.e., form an $\ell$-tuple with the sum $0$ under $H$). 
    \item Run a (space-bounded) algorithm with oracle access to $H, G$ except $G(x)$ is replaced with a uniform superposition $\ket +^{\otimes \log N_0}$; i.e., only $G(x)$ is treated as a compressed oracle\footnote{In the actual analysis, we will treat $G$ as an oracle with binary outputs; i.e., split each output in $[N_0]$ into $\log {N_0}$ binary outputs and keep tracks the probability that each bit flips. This is a minor difference and only introduces a multiplicative $\log {N_0}$ for the expectation. We do not complicate the intuition here and only provide the game without splitting outputs.}. 
    \item Finally, measure the qubit in the Hadamard basis. The algorithm wins if and only if the measurement outcome is not all ``$+$'', indicating the $x$-th entry is written. 
\end{itemize}

We show that we can apply the time-segmentation on the above game to remove space constraints, resulting in $L$ segments. Within each segment, there is a $T/L$ query algorithm with $S$ qubits of advice and tries to maximize the expectation of the above game. Let $w^{(i)}_{\vec{x}}$ be the probability/expectation that the $x$-th entry is written. By a standard triangle inequality, we show that:
\begin{align*}
    \sum_{\vec{x}} w_{\vec{x}} \leq  L \cdot \sum_{i=1}^L \sum_{\vec{x}} w^{(i)}_{\vec{x}}.
\end{align*}
Therefore, if we can bound the maximum expectation in the above game with a non-uniform algorithm of $T/L$ queries and $S$ qubits by $W$ (i.e., $\sum_{\vec{x}} w_{\vec{x}}^{(i)} \leq W$ for every $i$), we can show that a $T$ query and $S$ space algorithm can have a maximum expectation $L^2 \cdot W$.

To prove an upper bound of the expectation by a non-uniform algorithm, we use the quantum presampling tool introduced in~\cite{guo2021unifying} and improved in~\cite{liu2022nonuniformityquantumadvicequantum}, which is a powerful technique to remove the quantum advice by introducing a multi-instance version of the original game; for which, we leave all the details to the main body. Eventually, we are able to reduce the problem to a standard problem: how many distinct $\ell$-tuples with the same sum that a uniform query-bounded algorithm can find.

\subsection*{The time bound for finding many $\ell$-tuples with the same sum (step 4).}

Previous works have established similar bounds by tracking the number of tuples discovered in each step using the compressed oracle technique~\cite{zhandry2019record}. For instance, \cite{liu2019finding,Hamoudi_2023} analyzed the case of finding multiple $2$-collisions. However, most of these works considered only cases where adding a new entry to the database introduces at most one new instance (e.g., a preimage or a collision pair), focusing solely on disjoint instances. For example, if four points $x_1, x_2, x_3, x_4$ map to the same value, prior work would treat them as two disjoint collisions, whereas in our setting, they correspond to six distinct collision pairs (i.e., any two distinct $x_i, x_j$ can form a pair of collisions). 

The key challenge in our proof is that in our case, adding a single new entry can introduce multiple new instances, since our target is to bound the number of $\ell$-tuples. Our proof proceeds inductively based on the tuple size $\ell$. 
\begin{itemize}
    \item \textbf{Base Case ($\ell = 1$).} We compute the probability that the algorithm finds a set of distinct points that all map to the same value under the random oracle using the standard compressed oracle method. 

    \item \textbf{Inductive Step ($\ell > 1$).} For larger tuples, we extend the base case by analyzing how smaller tuples can be combined to form larger ones. The key observation is that an algorithm attempting to find multiple $\ell$-tuples with the same sum can proceed in one of two ways:
    \begin{itemize}
        \item It already stores at least $\eta$ distinct $(\ell-1)$-tuples with the same sum in its database at some point during execution.
        \item Alternatively, the following event occurs at least ${\kappa}/{\eta}$ times: a newly added entry combines with $(\ell-1)$ existing entries in the database to form at least one new $\ell$-tuple with the target sum. 
    \end{itemize}
    Here, $\eta$ is a threshold that we will set later. The probability of the first case can be bounded using the induction hypothesis. In the second case, by focusing on this specific event, we circumvent the challenge that arises when a single query introduces multiple new instances at once.
\end{itemize}
The details can be found in~\Cref{sec:time_bound_l_tuples}.
Finally, by combining all the steps, we are able to show a time-space tradeoff for finding a single nested collision.

\section{Preliminaries}
We use $[N]$ to denote the set $\set{0,1,2,\cdots,N-1}$. A random oracle with domain $[M]$ and image $[N]$ is a function $H$ sampled uniformly from all functions mapping $[M]$ to $[N]$. For two $n$-bit strings $a,b$ we use $\langle a,b\rangle$ to denote its inner product mod 2.
\begin{lemma}[Chebyshev Bound]\label{lem:Chebyshev}
    Let $X$ be a random variable with expectation $\mu$ and variance $\sigma$, then
    \begin{equation*}
        \Pr\left[\abs{X-\mu}\geq k\right]\leq\frac{\sigma}{k^2}.
    \end{equation*}
\end{lemma}
\begin{lemma}[Grover Search,~\cite{grover1996fast}]\label{lem:grover_search}
    Let $f:S\rightarrow\set{0,1}$ be a function. There exists an efficient algorithm that finds an $x\in S$ such that $f(x)=1$ with $O\brackets{\sqrt{\frac{|S|}{|f|}}}$ calls to a $\Pi_f$ gate:
    \begin{equation*}
        \Pi_f\ket{x}:=(-1)^{f(x)}\ket{x}
    \end{equation*}
    if a uniform superposition over $S$ can be efficiently generated from $\ket{0}$.
\end{lemma}
\begin{lemma}[Jensen's Inequality,~\cite{JensenInequality}]\label{lem:jensen_inequal}
    Let $D,g$ be positive integers. Let $c_0,c_1,\cdots,c_{D-1}$ be a distribution over $[D]$. Let $p_0,p_1,\cdots,p_{D-1}\in\mathbb{R}^D$ satisfies that $\sum_{i\in[D]}c_ip_i>0$. Then we have
    \begin{equation*}
        \sum_{i\in[D]}c_ip_i\leq\brackets{\sum_{i\in[D]}c_ip_i^g}^{\frac{1}{g}}.
    \end{equation*}
\end{lemma}
\section{Compressed oracle and our model of computation}
\label{sec:compress_oracle}
\newcommand{\qryreg}{\mathbf{X}}
\newcommand{\ansreg}{\mathbf{U}}
\newcommand{\auxreg}{\mathbf{W}}
\newcommand{\roreg}{\mathbf{H}}
\newcommand{\dbreg}{\mathbf{D}}
\newcommand{\outreg}{\mathbf{R}}
\newcommand{\Hadamardo}{{\sf HaO}}

\newcommand{\lab}{\mathsf{Lbl}}
\newcommand{\maxcap}[1]{|#1|^{(\sf{max})}}

In this subsection, we recall the technique introduced by Zhandry \cite{zhandry2019record}. This part is adapted from~\cite{chung2020tight}. We present five equivalent oracle formulations: the standard oracle, the phase oracle, the compressed standard oracle, the compressed phase oracle, and the Hadamard oracle.

\renewcommand{\output}{{\sf Output}}
In most cases in this paper, the computation is about a state on five registers.
\begin{itemize}
    \item $\qryreg$ is the register that stores either an oracle query or an answer waiting to be written.
    \item $\ansreg$ is the register that stores the oracle's response or is used to store phase for the output process (as will be explained later).
    \item $\auxreg$ is the register that stores as ancilla qubits in the computation.
    \item $\outreg$ is the output register that is used to store answers.
    \item $\dbreg$ is the register that stores the random oracle $H$ or its database $D$ in the compressed oracle view.
\end{itemize}
We will explain how $\outreg$ works in \Cref{sec:time_bound_l_tuples}. Intuitively, it can always be viewed as the tensor product of $M'$ sub-registers with $n'$-qubits each for some $M'$ and $n'$.  The algorithm is only allowed to apply the $\output$ gate to a state of the form $\ket{x}_{\qryreg}\ket{u}_{\ansreg}$. This gate will add $u$ to the sub-register corresponding to $x$ on $\outreg$. We will write this process as:
\begin{equation*}
    \output\ket{x}_{\qryreg} \ket {u}_{\ansreg} \ket {w}_{\auxreg}\ket{r}_{\outreg} \otimes \ket {H}_{\dbreg} = \ket{x}_{\qryreg} \ket {u}_{\ansreg} \ket {w}_{\auxreg}\ket{r\oplus (x,u)}_{\outreg} \otimes \ket {H}_{\dbreg}.
\end{equation*}

\paragraph{Standard oracle.} Now we switch to our normal oracle setting. Let $H$ be a random oracle $[M] \to [N]$ where $N=2^n$. We can view an algorithm that runs in the random oracle model with respect to $H$ as the algorithm itself tensored with a random oracle register $\dbreg$ that is initialized to $\sum_{H}\ket{H}_{\dbreg}$ (ignoring the normalizing factor). The register $\dbreg$ stores the random function $\ket{H}_{\dbreg} = \ket {H(1)} \ket {H(2)} \cdots \ket {H(N)}$. The oracle unitary $\sto$ can be written as follows:
\begin{align*}
    \sto \ket{x}_{\qryreg} \ket {u}_{\ansreg} \ket {w}_{\auxreg}\ket{r}_{\outreg} \otimes \ket {H}_{\dbreg} = \ket{x}_{\qryreg} \ket {u \oplus H(x)}_{\ansreg} \ket {w}_{\auxreg}\ket{r}_{\outreg} \otimes \ket {H}_{\dbreg},  
\end{align*}
Note that there are some format mismatch in $\qryreg$ and $\ansreg$ when running $\sto$ and $\output$. We can assume that the algorithm itself can disambiguate this on its own since it knows which type of query it is using. In general, an algorithm consists of interleaving $\sto$, $\output$, local quantum unitaries $U_i$ that operate only on registers $\qryreg\otimes\ansreg\otimes\auxreg$ and a final computational measurement on the output register $\outreg$. The following proposition tells that the output distribution using a standard oracle is exactly the same as using a random oracle. 

\begin{lemma}[{\cite[Lemma 2]{zhandry2019record}}] \label{prop:purifyoracle}
    Let $\As$ be an (unbounded) quantum algorithm making oracle queries. The output of $\As$ given a random function $H$ is exactly identical to the output of $\As$ given access to a standard oracle. 
    Therefore, a random oracle with quantum query access can be perfectly simulated as a standard oracle. 
\end{lemma}

\paragraph{Phase oracle.} 
Define the unitary $V$ as $(I_{\qryreg} \otimes H^{\otimes n} \otimes I_{\auxreg\otimes\outreg\otimes\dbreg})$ which applies $H^{\otimes n}$ on the answer register $\ansreg$.  Define the phase oracle operator $\pho := V^\dagger \cdot \sto \cdot V$. 
\begin{align*}
    &\pho \ket{x}_{\qryreg} \ket{u}_{\ansreg}\ket{w}_{\auxreg}\ket{r}_{\outreg}\otimes \ket{H}_{\dbreg} \\
    =&  V^\dagger \cdot  \sto \cdot  \frac{1}{\sqrt{N}}  \sum_{y\in[N]}  (-1)^{\langle u,y\rangle} \ket{x}_{\qryreg} \ket{y}_{\ansreg} \ket{w}_{\auxreg}\ket{r}_{\outreg}\otimes \ket{H}_{\dbreg}  \\ 
    =& V^\dagger  \cdot   \frac{1}{\sqrt{N}}  \sum_{y\in[N]} (-1)^{\langle u,y\rangle} \ket{x}_{\qryreg} \ket {y + H(x)}_{\ansreg} \ket{w}_{\auxreg}\ket{r}_{\outreg}\otimes \ket{H}_{\dbreg}  \\ 
     =&   \frac{1}{N} \sum_{y,y'\in[N]}  (-1)^{\langle u,y\rangle + \langle y + H(x), y'\rangle}  \ket{x}_{\qryreg} \ket {y'}_{\ansreg} \ket{w}_{\auxreg}\ket{r}_{\outreg}\otimes \ket{H}_{\dbreg}  \\ 
     =&   \frac{1}{N} (-1)^{\langle u,H(x)\rangle} \sum_{y,y'\in[N]}  (-1)^{\langle y + H(x),y' +u\rangle}  \ket{x}_{\qryreg} \ket {y'}_{\ansreg} \ket{w}_{\auxreg}\ket{r}_{\outreg}\otimes \ket{H}_{\dbreg}  \\ 
    =&  \ket{x}_{\qryreg} \ket{u}_{\ansreg} \ket{w}_{\auxreg}\ket{r}_{\outreg}\otimes (-1)^{\langle u,H(x)\rangle} \ket {H}_{\dbreg}. 
\end{align*}
The following lemma states that a phase oracle is equivalent to a standard oracle.
\begin{lemma}[{\cite[Lemma 3]{zhandry2019record}}] 
    \label{lem:pho-simulate}
    Let $\As$ be an (unbounded) quantum algorithm making queries to a standard oracle. 
    Let $\Bs$ be the algorithm that is identical to $\As$, except it performs $V$ and $V^\dagger$ before and after each query.
    Then the output distributions of $\As$ (given access to a standard oracle) and $\Bs$ (given access to a phase oracle) are identical.
    Therefore, a quantum random oracle can be perfectly simulated as a phase oracle.
\end{lemma}


\paragraph*{Compressed standard oracle.} 
The compressed standard oracle can be viewed as a type of lazy sampling technique. Instead of initializing $H$ at the very beginning, the compressed oracle creates a database $\ket{D}_{\dbreg}=\ket{D(1)}_{\dbreg_1}\ket{D(2)}_{\dbreg_2}\cdots\ket{D(N)}_{\dbreg_N}$ where $D(x)\in\mathbb{F}_N\cup\set{\bot}$ and $\ket{D}$ is initialized to $\ket{\emptyset}_{\dbreg}=\ket{\bot, \bot, \cdots, \bot}$ where $\bot$ is a symbol that indicates the lack of information of the algorithm on certain function values. Let $|D|$ denote the number of entries in $D$ that are not $\bot$. 

The database is initialized as an empty list $D_0$ of length $N$, in other words, it is initialized as the pure state $\ket{\emptyset} := \ket{\bot, \bot, \cdots, \bot}$. Let $|D|$ denote the number of entries in $D$ that are not $\bot$. 

For any $D$ and $x$ such that $D(x) = \bot$, we define $D \cup (x, u)$ to be the database $D'$, such that for every $x' \ne x$, $D'(x') = D(x)$ and at the input $x$, $D'(x) = u$. 

The compressed standard oracle is the unitary $\csto := \stddecomp \cdot \csto' \cdot \stddecomp$, where
\begin{itemize}
    
    \item $\csto'$ writes $D(x)$ to the answer register $\ansreg$ by writing $u\oplus D(x)$ into it when $D(x)\neq \bot$ as usual but does nothing when $D(x)=\bot$. Or to say that we can define addition for $\bot$: $u\oplus\bot=u$, $\forall u\in[N]$. 
    Formally,
    \begin{equation*}
        \csto'\ket{x}_{\qryreg}\ket{u}_{\ansreg}\ket{w}_{\auxreg}\ket{r}_{\outreg}\otimes\ket{D}_{\dbreg} = \ket{x}_{\qryreg}\ket{u\oplus D(x)}_{\ansreg}\ket{w}_{\auxreg}\ket{r}_{\outreg}\otimes\ket{D}_{\dbreg}.
    \end{equation*}
    
    \item When the algorithm queries, the database calls $\stddecomp$ which unfolds the database and samples a value $y$ for positions that the algorithm does not know what the value is. More specifically, $\stddecomp \ket{x}_{\qryreg}\ket{u}_{\ansreg}\ket{w}_{\auxreg}\ket{r}_{\outreg}\otimes\ket{D}_{\dbreg} := \ket{x}_{\qryreg}\ket{u}_{\ansreg}\ket{w}_{\auxreg}\ket{r}_{\outreg}\otimes\stddecomp_x\ket{D}_{\dbreg}$, where $\stddecomp_x$ works on $\dbreg_x$. 
    \begin{itemize}
        \item If $D(x) = \bot$, $\stddecomp_x$ maps $\ket \bot$ to 
        \[\frac{1}{\sqrt{N}} \sum_{y\in[N]} \ket y.\]
        
        \item If $D(x) \ne \bot$, $\stddecomp_x$ works on the $x$-th register, and it is an identity on 
        \[\frac{1}{\sqrt{N}} \sum_{y\in[N]} (-1)^{\langle u,y\rangle} \ket y\]
        for all $u \ne 0$; it maps the uniform superposition $\frac{1}{\sqrt{N}} \sum_{y\in[N]} \ket y$ to $\ket \bot$. 
        
        More formally, for $D'$ such that $D'(x) = \bot$, 
        \begin{align*}
            \stddecomp_x  \frac{1}{\sqrt{N}} \sum_{y\in[N]} (-1)^{\langle u,y\rangle} \ket{D' \cup (x, y)}_{\dbreg}
                        = \frac{1}{\sqrt{N}} \sum_{y\in[N]} (-1)^{\langle u,y\rangle}  \ket{D' \cup (x, y)}_{\dbreg}
        \end{align*}
        for any $u\neq 0$ and, 
        \begin{align*}   
            \stddecomp_x  \frac{1}{\sqrt{N}} \sum_{y\in[N]} \ket{D' \cup (x, y)}_{\dbreg}  =\ket{D'}_{\dbreg}.
        \end{align*}
    \end{itemize}
    Intuitively, it swaps a uniform superposition $\frac 1{\sqrt{N}} \sum_{y\in[N]} \ket y$ with $\ket \bot$ on $\dbreg_x$ and does nothing on other orthogonal basis. So it is a well defined unitary.
\end{itemize}
Zhandry proves, that ${\sf StO}$ and \csto{} are perfectly indistinguishable by any \textit{unbounded} quantum algorithm. 
\begin{lemma}[{\cite[Lemma 4]{zhandry2019record}}] 
    Let $\As$ be an (unbounded) quantum algorithm making oracle queries. The output of $\As$ given access to the standard oracle is exactly identical to the output of $\As$ given access to a compressed standard oracle. 
\end{lemma}

In this work, we only consider query complexity, and thus simulation efficiency is irrelevant to us. 
\paragraph{Compressed phase oracle} The compressed phase oracle is the unitary $\cpho:=\stddecomp\cdot\cpho'\cdot\stddecomp$, where
\begin{align*}
    &\cpho'\ket{x}_{\qryreg}\ket{u}_{\ansreg}\ket{w}_{\auxreg}\ket{r}_{\outreg}\otimes\ket{D}_{\dbreg}\\
    =&\ket{x}_{\qryreg}(-1)^{\langle u,D(x)\rangle}\ket{u}_{\ansreg}\ket{w}_{\auxreg}\ket{r}_{\outreg}\otimes\ket{D}_{\dbreg}\\
    =&\ket{x}_{\qryreg}\ket{u}_{\ansreg}\ket{w}_{\auxreg}\ket{r}_{\outreg}\otimes(-1)^{\langle u,D(x)\rangle}\ket{D}_{\dbreg}.
\end{align*}
Redefine the unitary $V$ as $(I_{\qryreg} \otimes H^{\otimes n} \otimes I_{\auxreg\otimes\outreg\otimes\dbreg})$. Note that $\cpho'=V^\dagger\cdot\csto'\cdot V$ and $V$ commutes with $\stddecomp$. Thus,
\begin{align*}
    \cpho=&\stddecomp\cdot\cpho'\cdot\stddecomp\\
    =&\stddecomp\cdot V^\dagger\cdot\csto' \cdot V\cdot\stddecomp\\
    =&V^\dagger\cdot \stddecomp\cdot\csto' \cdot\stddecomp\cdot V\\
    =&V^\dagger\cdot\csto\cdot V.
\end{align*}
By this equivalence, we obtain the following corollary:
\begin{corollary}
    \label{lem:cpho-simulate}
    Any quantum algorithm $\cal{A}$ equipped with a quantum random oracle can be perfectly simulated by another algorithm $\cal{B}$ equipped with a compressed phase oracle.\zikuan{Need to be precise? For example mention that the number of query is the same.}
\end{corollary}

The following lemma states that the size of the database is at most the number of queries.
\begin{lemma}[\cite{chung2020tight}, Lemma 2.6]
    \label{lem:bounded-database-phase}
    Let $\As$ be a quantum algorithm making at most $T$ queries to a compressed phase oracle. The overall state of $\As$ and the oracle database can be written as 
    \begin{equation*}
        \sum_{x,u,w,r, D: |D| \leq T} \alpha_{x,u,w,r, D}\ket{x,u,w,r}\otimes\ket{D}.
    \end{equation*}
    
    Moreover, this holds even if the state is conditioned on arbitrary outcomes (with non-zero probability) of $\As$'s intermediate measurements. 
\end{lemma}
Here we further formalize the way an algorithm interacts with a compressed phase oracle. For a quantum algorithm $\As$ interacting with a compressed phase oracle $\cpho$, suppose it makes $T$ queries in total. Then the algorithm is equivalent to applying a sequence of unitaries on the joint system $\qryreg\otimes\ansreg\otimes\auxreg\otimes\outreg\otimes \dbreg$ to a certain initialized state and then measuring the output register $\outreg$ to obtain the output. Without loss of generality, we initialize the system to be
\begin{equation*}
    \ket{\Phi_{\sf{init}}} = \ket{0}_{\qryreg} \ket{0}_{\ansreg}\ket{0}_{\auxreg}\ket{0}_{\outreg}\otimes \ket{\emptyset}_{\dbreg}.
\end{equation*}
The algorithm $\As$ could be characterized as a sequence of local unitaries and compressed phase oracle queries, 
\begin{equation*} 
    (U_{T}\otimes I_{\dbreg}) \cdot \cpho\cdot (U_{T-1}\otimes I_{\dbreg})\cdot\cpho\cdots \cpho \cdot(U_{0}\otimes I_{\dbreg}),
\end{equation*}
where $U_0,\cdots,U_T$ are unitaries on the algorithm space $\qryreg\otimes\ansreg\otimes\auxreg\otimes\outreg$. Then, we denote $\ket{\Phi_i}$ to be the intermediate system state after applying $i$ oracle queries and some local unitaries,
\begin{equation*}
\ket{\Phi_i} = (U_{i}\otimes I_{\dbreg})\cdot \cpho\cdot (U_{i-1}\otimes I_{\dbreg})\cdot\cpho\cdots \cpho \cdot(U_{0}\otimes I_{\dbreg}) \ket{\Phi}_0,
\end{equation*}
where $i\in [T]$. 

We now relate the winning probability of the algorithm to the final state of the system, $\ket{\Phi_T}$. For the scenario of an algorithm interacting with a compressed oracle, we determine whether it wins or not by measuring the database $\dbreg$ and the output register $\outreg$ first, and then check whether the resulting database $D$ and output $r$ satisfy the goal of the algorithm. More formally, suppose that we formalize the goal of the algorithm as a relation $\Rs$ over all pairs of possible output and database, and the algorithm wins if and only if it outputs a pair of $(D,r)$ such that $(D,r)\in \Rs$. For example, in the case of the collision finding problem, $\Rs = \{(r,D)| r=(x,x'),x\neq x', D(x)=D(x')\neq \bot \}$. We then define $\Pi_{\sf{win}}$ as to be a projector onto all basis states $\ket{x}_{\qryreg}\ket{u}_{\ansreg}\ket{w}_{\auxreg}\ket{r}_{\outreg}\otimes\ket{D}_{\dbreg}$ such that $(r,D)\in \Rs$, in other words projecting on all ``winning'' outputs. Therefore, the winning probability of $\As$ equals $\norm{\Pi_{\sf{win}}{\ket{\Phi_T}}}^2$.

The following lemma establishes a connection between the winning probability when interacting with a random oracle and when interacting with a compressed oracle.

\begin{lemma}[{\cite[Lemma 5]{zhandry2019record}}] \label{lem:win-prob-db-algo}
Consider a quantum algorithm $\As$ making phase queries to a random oracle $H$ and outputting tuples $(x_1,\cdots,x_k,y_1,\cdots,y_k,z)$. Let $\Rs$ be a collection of such tuples. Let $p$ be the probability that, $\As$ outputs a tuple such that (1) the tuple is in $\Rs$, and (2) $H(x_i)=y_i$ for all $i$. Now consider running $\As$ with the oracle $\cpho$, and suppose the database register $\dbreg$ is measured after $\As$ produces its output. Let $p'$ be the probability that, $\As$ outputs a tuple such that (1) the tuple is in $\Rs$, and (2) $D(x_i)=y_i\neq \bot$ for all $i$. Then we have $\sqrt{p}\leq \sqrt{p'}+\sqrt{k/N}$, where $N$ is the range size.
    
\end{lemma}

\begin{remark}
    The \Cref{lem:win-prob-db-algo} is a little different from the original statement of the lemma 5 in Zhandry's work \cite{zhandry2019record}, but it is not hard to see that they are equivalent. In the original lemma\cite{zhandry2019record}, the algorithm $\As$ is making standard queries in the first scenario, and making queries to oracle $\csto$ in the second scenario. Here in the statement, the algorithm uses phase queries and $\cpho$ oracle instead. However, since we previously showed that $\pho := V^\dagger \cdot \sto \cdot V$, and that $\cpho = V^\dagger\cdot\csto\cdot V$, the analysis could similarly apply to the phase oracle setting in \Cref{lem:win-prob-db-algo}.
\end{remark}

The following lemma shows how $\cpho$ affects the database under the standard basis:

\begin{lemma}[\cite{Hamoudi_2023}, Lemma 4.1]\label{each-query-in-compress-oracle}
    If the operator $\cpho$ is applied to a basis state $\ket{x}_{\qryreg}\ket{u}_{\ansreg}\ket{w}_{\auxreg}\ket{r}_{\outreg}\otimes\ket{D}_{\dbreg}$ where $u\neq 0$ then the register $\ket{D(x)}_{\dbreg_x}$ is mapped to
    \begin{align*}
        \circ&\sum_{y\in[N]}\frac{(-1)^{\langle u,y\rangle}}{\sqrt{N}}\ket{y} & \text{if }D(x)=\bot\\
        \circ&\substack{\frac{(-1)^{\langle u,D(x)\rangle}}{\sqrt{N}}\ket{\bot}+\frac{1+(-1)^{\langle u,D(x)\rangle}(N-2)}{N}\ket{D(x)}\\+\sum_{y\in[N]\backslash\set{D(x)}}\frac{1-(-1)^{\langle u,y\rangle}-(-1)^{\langle u,D(x)\rangle}}{N}\ket{y}} & \text{if }D(x)\in[N]
    \end{align*}
    and other registers are unchanged. If $u=0$ then none of the registers are changed.
\end{lemma}

\paragraph{Hadamard Oracle}
Define $\ket{\widehat{v}}_{\dbreg_x}:=\frac{1}{\sqrt{N}} \sum_{y\in[N]}(-1)^{\langle v,y\rangle}\ket{y}_{\dbreg_x}$ for $v\neq 0$ and $\ket{\widehat{0}}_{\dbreg_x}=\ket{\bot}_{\dbreg_x}$. Define $\dbreg_{-x}:=\otimes_{x'\in[M]-\set{x}}\dbreg_{x'}$. Then we have
\begin{align*}
    &\cpho\ket{x}_{\qryreg}\ket{u}_{\ansreg}\ket{w}_{\auxreg}\ket{r}_{\outreg}\otimes\ket{\widehat{d}}_{\dbreg_x}\otimes\ket{D}_{\dbreg_{-x}}\\
    =&\stddecomp\cdot\cpho'\cdot\stddecomp\ket{x}_{\qryreg}\ket{u}_{\ansreg}\ket{w}_{\auxreg}\ket{r}_{\outreg}\otimes\ket{\widehat{d}}_{\dbreg_x}\otimes\ket{D}_{\dbreg_{-x}}\\
    =&\stddecomp\cdot\cpho'\ket{x}_{\qryreg}\ket{u}_{\ansreg}\ket{w}_{\auxreg}\ket{r}_{\outreg}\otimes\brackets{\sum_{y\in[N]}\frac{(-1)^{\langle d,y\rangle}}{\sqrt{N}}\ket{y}_{\dbreg_x}}\otimes\ket{D}_{\dbreg_{-x}}\\
    =&\stddecomp\ket{x}_{\qryreg}\ket{u}_{\ansreg}\ket{w}_{\auxreg}\ket{r}_{\outreg}\otimes\brackets{\sum_{y\in[N]}\frac{(-1)^{\langle d+u,y\rangle}}{\sqrt{N}}\ket{y}_{\dbreg_x}}\otimes\ket{D}_{\dbreg_{-x}}\\
      =&\ket{x}_{\qryreg}\ket{u}_{\ansreg}\ket{w}_{\auxreg}\ket{r}_{\outreg}\otimes\ket{\widehat{d\oplus u}}_{\dbreg_x}\otimes\ket{D}_{\dbreg_{-x}}.
\end{align*}

The compressed Hadamard oracle $\Hadamardo$ uses the database register under a different basis. Define $\Hadamardo\ket{x}_{\qryreg}\ket{u}_{\ansreg}\ket{w}_{\auxreg}\ket{r}_{\outreg}\otimes\ket{D}_{\dbreg}:=\ket{x}_{\qryreg}\ket{u}_{\ansreg}\ket{w}_{\auxreg}\ket{r}_{\outreg}\otimes\ket{D\oplus (x,u)}_{\dbreg}$ where $D\oplus (x,u)$ is the database resulting from applying $+u$ under $\mathbb{F}_{2}^n$ to the value on $\dbreg_x$ register on $D$. Since $\Hadamardo$ and $\cpho$ are two identical operators defined under two different bases and the basis choice of $\dbreg$ does not affect the behavior of the algorithm, we have the following lemma.

\begin{lemma}
    Let $\As$ be an (unbounded) quantum algorithm making oracle queries. The output of $\As$ given access to the compressed phase oracle is exactly identical to the output of $\As$ given access to a compressed Hadamard oracle. 
\end{lemma}
\begin{proof}
    A Hadamard oracle is essentially a phase oracle under Hadamard basis.
\end{proof}

\section{Finding collisions implies many entries in the compressed oracle}\label{sec:many_entries_database}

Below we present a new lemma showing intuitively that, if an algorithm can find a pair of collisions with high probability, under the compressed oracle view, the expected number of non-$\bot$ entries in the compressed oracle cannot be too small. Our main proof idea is to first bound the increment in winning probability after each of the $T$ queries, then show that the overall winning probability is bounded by a combination of each increment. In \Cref{lem:progress-measure}, we carefully bound the weight increase on the basis where the database contains collision, by analyzing the evolution of the database during oracle query separately according to the number of existing non-$\bot$ entries. Then in \Cref{remember-tuples}, we combine the per-query increments, and connect it to the final winning probability.

We clarify that the lemma works for a slightly generalized version of the collision finding problem, namely the \emph{Labeled Collision Finding} problem, defined as follows:
\begin{definition}[Labeled Collision Finding]
    Given a random oracle $G: [M]\to [N_0]$, as well as a label function $\lab: [M]\to [N]$, for some fixed $y^*$ the goal is to find a pair $(x,x')\in [M]^2$ such that:
    \begin{itemize}
        \item $x\neq x'$.
        \item $\lab(x) = \lab(x')=y^*$.
        \item $G(x) = G(x')$.
    \end{itemize}
    We call the satisfying pair $(x,x')$ a pair of label-$y^*$ collisions.
\end{definition}


Notice that when the label function maps all inputs $x\in [M]$ to the same label, this problem becomes the standard collision finding problem, which requires only $G(x)=G(x'), x\neq x'$. For nontrivial label functions, the problem is essentially finding a collision pair of a random oracle $G$ that is also of the same label. We start by introducing a set of projectors on the joint system of algorithm register and database, we may use them to categorize which condition the database is on after or before each oracle query.
\begin{definition}
    \label{def:projectors}
    Given a label function $\lab:[M]\to [N]$, we define the following projectors by giving the basis states on which they project:
    \begin{itemize}
        \item $\Pi_v$: all basis states $\ket{x}_{\qryreg}\ket{u}_{\ansreg}\ket{w}_{\auxreg}\ket{r}_{\outreg}\otimes\ket{D}_{\dbreg}$ such that there are $v$ non-$\bot$ entries of the label $y^*$, i.e. $\sum_{x\in [M]}\mathds{1}\{D(x)\neq \bot, \lab(x)=y^*\}=v$. Let $\abs{D}_{y^*}$ denote the number of non-$\bot$ entries in $D$ with label $y^*$
        \item $P$: all basis states $\ket{x}_{\qryreg}\ket{u}_{\ansreg}\ket{w}_{\auxreg}\ket{r}_{\outreg}\otimes\ket{D}_{\dbreg}$ such that $\ket{D}_{\dbreg}$ contains a pair of labeled-collisions, i.e., $\exists x\neq x'$, $D(x)=D(x')\neq\bot, \lab(x) = \lab(x')=y^*$.
        \item $O$: all basis states $\ket{x}_{\qryreg}\ket{u}_{\ansreg}\ket{w}_{\auxreg}\ket{r}_{\outreg}\otimes\ket{D}_{\dbreg}$ such that: \\
        (1) there are no label-$y^*$ collisions in $D$; (2)$\lab(x)\neq y^*$.
        \item $Q$: all basis states $\ket{x}_{\qryreg}\ket{u}_{\ansreg}\ket{w}_{\auxreg}\ket{r}_{\outreg}\otimes\ket{D}_{\dbreg}$ such that: \\
        (1) there are no label-$y^*$ collisions in $D$;  (2) $D(x)=\bot$;  (3) $\lab(x)=y^*$; (4) $u\neq 0$.
        \item $R$: all basis states $\ket{x}_{\qryreg}\ket{u}_{\ansreg}\ket{w}_{\auxreg}\ket{r}_{\outreg}\otimes\ket{D}_{\dbreg}$ such that: \\
        (1) there are no label-$y^*$ collisions in $D$;  (2) $D(x)\neq\bot$; (3) $\lab(x)=y^*$; (4) $u\neq 0$.
        \item $S$: all basis states $\ket{x}_{\qryreg}\ket{u}_{\ansreg}\ket{w}_{\auxreg}\ket{r}_{\outreg}\otimes\ket{D}_{\dbreg}$ such that: \\
        (1) there are no label-$y^*$ collisions in $D$;  (2) $u=0$.(3) $\lab(x)=y^*$.
        \item Also define $P_v = \Pi_vP$, $O_v=\Pi_vO$, $Q_v = \Pi_vQ$, $R_v = \Pi_v R$, $S_v = \Pi_v S$.
    \end{itemize}
\end{definition}
Then we have the following properties:
\begin{itemize}
    \item[-] $\sum_{v=0}^M \Pi_v = I$.
    \item[-] $P+O+Q+R+S = I$, and $P,O,Q,R,S$ are orthogonal.
    \item[-] $P_v+O_v+Q_v+R_v+S_v = \Pi_v$, $\forall v\in [0,M]$, and that $P_v,O_v,Q_v,R_v,S_v$ are orthogonal.
    \item[-] $\forall v\neq v'$, their corresponding projectors are orthogonal, it also holds for $P_v,O_v,Q_v,R_v,$ and $S_v$.
    \item[-] Notice that the whether a state lies in the projected space of $P$ and $\Pi_v$ only depends on the $\dbreg$ register information of the state, so these projectors commute with local unitaries on algorithm space, $\Pi_v(U\otimes I_{\dbreg}) = (U\otimes I_{\dbreg})\Pi_v$, and it is the same for $P,O,Q,R,S$.
\end{itemize}

Based on the projectors, we could define the average non-empty size of the database, throughout the process of an algorithm interacting with a compressed oracle.

\begin{definition}
\label{def:non-bot size}
Suppose that we have an algorithm $\cA$ interacting with a compressed phase oracle $\cpho$, with a total of $T$ queries, $\cpho$ corresponds to a database with $M$ entries and each entry has an element in $\mathbb{F}_{N_0} \cup\set{\bot}$. Consider the case where we measure the database register after the $i$-th query, and we define
\begin{equation*}
    p_{i,v} = \norm{\Pi_v\ket{\Phi_i}}^2
\end{equation*} 
to be the probability of collapsing into a database with $v$ non-$\bot$ entries with label $y^*$, $\forall i\in [T], v\in [0,M]$. Then, we define
\begin{equation*}
    V_i = \sum_{v=0}^{M} v\cdot p_{i,v}
\end{equation*}
to be the expectation of the number of non-$\bot$ entries with label $y^*$ after $i$ queries, $\forall i\in [T]$. Furthermore, we define $V$ for the case when we randomly pick $i\gets [T]$ and measure the database register $\dbreg$ after $i$ oracle queries,
\begin{equation*}
    V = \frac{1}{T}\sum_{i=1}^T V_i = \frac{1}{T}\sum_{i=1}^T \sum_{v=0}^M v \norm{\Pi_v\ket{\Phi_i}}^2.
\end{equation*}    
\end{definition}

\begin{remark}
\label{rem:non-bot-entry}
    Notice that after $i$ queries to the oracle, we have $V_i = \sum_{v=0}^{M} v\cdot p_{i,v}$. It seems that this sum appears to have $M$ terms, but from \Cref{lem:bounded-database-phase} we know that at the moment the database register only has support over states that have at most $i$ non-$\bot$ entries, and thus the maximal single-labeled non-$\bot$ capacity is also at most $i$. Therefore, for $v>i$, we have $p_{i,v}=0$, and $V_i = \sum_{v=0}^{i} v\cdot p_{i,v}$.
\end{remark}
Now we move on to prove a proposition, intuitively bounding the increase in success probability after each oracle query.

\begin{proposition}
\label{lem:progress-measure}
    For an arbitrary state $\ket{\psi}$ in joint space $\qryreg\ansreg\auxreg\outreg\dbreg$, and an arbitrary local unitary $U$ on algorithm space $\qryreg\ansreg\auxreg\outreg$, we have
    \begin{equation*}
        \norm{P_v\cdot\cpho\cdot(U\otimes I_{\dbreg})\cdot(I-P)\ket{\psi}} \leq 3\sqrt{\frac{v-1}{N_0}}\norm{\Pi_v\ket{\psi}} + \sqrt{\frac{v-1}{N_0}}\norm{\Pi_{v-1}\ket{\psi}},
    \end{equation*}
    where $N_0$ is the range size of the compressed oracle, $v>0$ and $P,\Pi_v,P_v$ are projectors defined in \Cref{def:projectors}.
\end{proposition}

\begin{proof} \zihan{Put it in appendix?}
Define $\ket{\phi} = (U\otimes I_{\dbreg})\ket{\psi}$. Recall that $U\otimes I_{\dbreg}$ commutes with $I-P$, then
\begin{equation}
\label{eq:phi-psi}
\begin{aligned}
    P_v\cdot\cpho\cdot(U\otimes I_{\dbreg})\cdot(I-P)\ket{\psi} =& P_v\cdot\cpho\cdot(I-P)\ket{\phi}.
\end{aligned}
\end{equation}

We now analyze the resulting state after applying $\cpho$ to each support of $(I-P)\ket{\phi}$, considering $(I-P)\ket{\phi} = O\ket{\phi}+Q\ket{\phi}+R\ket{\phi}+S\ket{\phi}$.

We first consider $\cpho\cdot O\ket{\phi}$. Since $\lab(x)\neq y^*$, adding this point to the database does not change the structure of the database on label-$y^*$ entries. We have
\begin{equation*}
    P_v\cdot\cpho\cdot O\ket{\phi}=0
\end{equation*}
because $P_v\ket{\phi}=0$.\\
We then consider $\cpho\cdot Q\ket{\phi}$, which stands for the case of filling a new entry with label $y^*$ into the database. Since $Q\ket{\phi}=\sum_{b=0}^M Q_b\ket{\phi}$, we analyze them separately. Assume $Q_b\ket{\phi} = \sum_{x,u,w,r,D'}\alpha^{(b)}_{x,u,w,r,D'}\ket{x,u,w,r}\ket{D'}$, where $x,u,w,r$ takes arbitrary value from the register space of $\qryreg,\ansreg,\auxreg,\outreg$, and $D'$ takes the value of all database that $D'(x)=\bot$ and $\abs{D}_{y^*}=b$. Then by \Cref{each-query-in-compress-oracle}, we have
\begin{equation*}
\begin{aligned}
    P_v\cdot\cpho\cdot Q\ket{\phi} =& \sum_{b=0}^M P_v\cdot \cpho\cdot Q_b\ket{\phi} \\
    =& \sum_{b=0}^{M}P_v\cdot\cpho\left( \sum_{\substack{x,u,w,r,D' \\ \abs{D}_{y^*}=b}}\alpha^{(b)}_{x,u,w,r,D'}\ket{x,u,w,r}\ket{D'}\right)\\
    =& \sum_{b=0}^{M} \sum_{\substack{x,u,w,r,D' \\ \abs{D}_{y^*}=b}}\alpha^{(b)}_{x,u,w,r,D'} P\cdot \Pi_v \left(\ket{x,u,w,r}\otimes\sum_{y\in[N_0]}\frac{(-1)^{\langle u,y\rangle}}{\sqrt{N_0}}\ket{D'\cup(x,y)} \right)\\
    =&  \sum_{\substack{x,u,w,r,D' \\ \abs{D}_{y^*}=v-1}} \alpha^{(v-1)}_{x,u,w,r,D'}P \left(\ket{x,u,w,r}\otimes\sum_{y\in[N_0]}\frac{(-1)^{\langle u,y\rangle}}{\sqrt{N_0}}\ket{D'\cup(x,y)} \right)\\
    =&  \sum_{\substack{x,u,w,r,D' \\ \abs{D}_{y^*}=v-1}} \alpha^{(v-1)}_{x,u,w,r,D'}\ket{x,u,w,r}\otimes\sum_{\substack{y\ s.t.\\\exists x',D'(x')=y\text{ and }\lab(x)=y^*}}\frac{(-1)^{\langle u,y\rangle}}{\sqrt{N_0}}\ket{D'\cup(x,y)}.
\end{aligned}
\end{equation*}
The fourth equality holds as we want $\ket{x,u,w,r}\ket{D'\cup(x,y)}$ to be inside the projected subspace of $\Pi_v$, so $|D'\cup(x,y)|=b+1=v$. The fifth equality holds as we want $\ket{x,u,w,r}\ket{D'\cup(x,y)}$ to be inside the projected subspace of $P$, so $y$ needs to collide with one of the $v-1$ existing different value inside $D'$. Then, we have
\begin{equation}
\label{eq:Q-norm}
    \begin{aligned}
        \norm{P_v\cdot\cpho\cdot Q\ket{\phi}}^2 =& \sum_{\substack{x,u,w,r,D' \\ \abs{D}_{y^*}=v-1}} \sum_{\substack{y\ s.t.\\\exists x',D'(x')=y\text{ and }\lab(x)=y^*}} \abs{\alpha^{(v-1)}_{x,u,w,r,D'}\frac{(-1)^{\langle u,y\rangle}}{\sqrt{N_0}}}^2\\
        \leq&\frac{1}{\sqrt{N_0}} \sum_{\substack{x,u,w,r,D' \\ \abs{D}_{y^*}=v-1}} \abs{D}_{y^*}\cdot \abs{\alpha^{(v-1)}_{x,u,w,r,D'}}^2\\
        =& \frac{v-1}{N_0}\norm{Q_{v-1}\ket{\phi}}^2.
    \end{aligned}
\end{equation}

The inequality is because the number of $y$ such that there exists $x'$ such that $D'(x')=y$ and $\lab(x)=y^*$ is at most $|D|_{y^*}$. We now consider $\cpho\cdot R\ket{\phi}$, which stands for rerandomizing an existing entry in the database. Similarly, we assume $R_b\ket{\phi} = \sum_{x,u,w,r,D',y}\beta^{(b)}_{x,u,w,r,D',y}\ket{x,u,w,r}\ket{D'\cup(x,y)}$. Here $x,u,w,r$ takes arbitrary value; $D'$ still takes the value of all database that $D'(x)=\bot$ and $\abs{D'}_{y^*}=b-1$, so that $|D'\cup(x,y)|=b$; and $y$ takes all value in $\set{1,2,\cdots,N_0-1}$ such that $D'\cup(x,y)$ has no label-$y^*$ collision, therefore the state lies in support of $R_b$. Then by \Cref{each-query-in-compress-oracle}, we have
\begin{equation*}
\begin{aligned}
    &P_v\cdot\cpho\cdot R\ket{\phi}\\
    =& \sum_{b=0}^M P_v\cdot \cpho\cdot R_b\ket{\phi} \\
    =& \sum_{b=0}^{M}P_v\cdot\cpho\Bigg( \sum_{\substack{x,u,w,r,D',y \\\abs{D}_{y^*}=b-1,y\neq0}}\beta^{(b)}_{x,u,w,r,D',y}\ket{x,u,w,r}\ket{D'\cup(x,y)}\Bigg)\\ 
    =& \sum_{b=0}^{M}\sum_{\substack{x,u,w,r,D',y \\\abs{D}_{y^*}=b-1,y\neq 0}} \beta^{(b)}_{x,u,w,r,D',y} P\cdot \Pi_v\Bigg[ \ket{x,u,w,r}\otimes \Bigg(\frac{(-1)^{\langle u,y\rangle}}{\sqrt{N_0}}\ket{D'}+\\ &\frac{1+(-1)^{\langle u,y\rangle}(N_0-2)}{N_0}\ket{D'\cup(x,y)}+ \sum_{y'\neq y}\frac{1-(-1)^{\langle u,y'\rangle}-(-1)^{\langle u,y\rangle}}{N_0}\ket{D'\cup(x,y')}\Bigg) \Bigg]\\
    =& \sum_{\substack{x,u,w,r,D',y \\\abs{D}_{y^*}=v,y\neq 0}} \beta^{(v+1)}_{x,u,w,r,D',y} P\Bigg( \ket{x,u,w,r}\otimes \frac{(-1)^{\langle u,y\rangle}}{\sqrt{N_0}}\ket{D'} \Bigg) + \sum_{\substack{x,u,w,r,D',y \\\abs{D}_{y^*}=v-1,y\neq 0}} \beta^{(v)}_{x,u,w,r,D',y} P\\
     &\cdot \Bigg[ \ket{x,u,w,r}\otimes \Bigg(\frac{1+(-1)^{\langle u,y\rangle}(N_0-2)}{N_0}\ket{D'\cup(x,y)}+ \sum_{y'\neq y}\frac{1-(-1)^{\langle u,y'\rangle}-(-1)^{\langle u,y\rangle}}{N_0}\ket{D'\cup(x,y')}\Bigg) \Bigg]\\
    =& \sum_{\substack{x,u,w,r,D',y \\\abs{D}_{y^*}=v-1,y\neq 0}} \sum_{\substack{y'\ s.t.\\ \exists x',D'(x')=y'\text{ and }\lab(x)=y^*}} \beta^{(v)}_{x,u,w,r,D',y}\frac{1-(-1)^{\langle u,y'\rangle}-(-1)^{\langle u,y\rangle}}{N_0} \ket{x,u,w,r}\ket{D'\cup(x,y')}.
\end{aligned}
\end{equation*}
Similarly, the fourth equality holds as we want the database corresponds to each state to have $v$ non-$\bot$ entries. The fifth equality holds as given that there is no collisions in $D\cup(x,y)$, the only term that could fall in the projected subspace of $P$ is $\ket{x,u,w,r}\ket{D'\cup(x,y')}$, conditioning on that $y'$ needs to collide with one of the $v-1$ existing different value inside $D'$. Then, we have
\begin{equation}
\label{eq:R-norm}
    \begin{aligned}
        \norm{P_v\cdot\cpho\cdot R\ket{\phi}}^2 =& \sum_{\substack{x,u,w,r,D',y \\\abs{D}_{y^*}=v-1,y\neq 0}} \sum_{\substack{y'\ s.t.\\ \exists x',D'(x')=y'\text{ and }\lab(x)=y^*}} \abs{\beta^{(v)}_{x,u,w,r,D',y}\frac{1-(-1)^{\langle u,y'\rangle}-(-1)^{\langle u,y\rangle}}{N_0}}^2\\
        \leq& \frac{1}{N_0}\sum_{\substack{x,u,w,r,D',y \\\abs{D}_{y^*}=v-1,y\neq 0}} \abs{D}_{y^*}\cdot \abs{\beta^{(v)}_{x,u,w,r,D',y}}\cdot \abs{1-(-1)^{\langle u,y'\rangle}-(-1)^{\langle u,y\rangle}}^2\\
        \leq & \frac{9(v-1)}{N_0}\sum_{\substack{x,u,w,r,D',y \\\abs{D}_{y^*}=v-1,y\neq 0}}\abs{\beta^{(v)}_{x,u,w,r,D',y}}^2\\
        =& \frac{9(v-1)}{N_0} \norm{R_{v}\ket{\phi}}^2.
    \end{aligned}
\end{equation}

We finally consider $\cpho\cdot S\ket{\phi}$, in which case the oracle query does not change the database. $\cpho\ket{x,0,w,r}\ket{D} = (-1)^{\langle 0, D(x)\rangle}\ket{x,0,w,r}\ket{D} = \ket{x,0,w,r}\ket{D}$ for all $x,w,r,D$ , so that 
\begin{equation}
\label{eq:S-norm}
    \cpho\cdot S\ket{\phi} = S\ket{\phi},\ \norm{P_v\cdot\cpho\cdot S\ket{\phi}} = 0.
\end{equation}

Combining \Cref{eq:phi-psi,eq:Q-norm,eq:R-norm,eq:S-norm}, we have
\begin{equation}
\label{eq:phi-progress}
    \begin{aligned}
        &\norm{P_v\cdot\cpho\cdot(U\otimes I_{\dbreg})\cdot(I-P)\ket{\psi}}\\
        =& \norm{P_v\cdot\cpho\cdot(I-P)\ket{\phi}}\\
        \leq& \norm{P_v\cdot\cpho\cdot O \ket{\phi}} + \norm{P_v\cdot\cpho\cdot Q \ket{\phi}}+\norm{P_v\cdot\cpho\cdot R \ket{\phi}} + \norm{P_v\cdot\cpho\cdot S \ket{\phi}}\\
        \leq & \sqrt{\frac{v-1}{N_0}}\norm{Q_{v-1}\ket{\phi}} + 3\sqrt{\frac{v-1}{N_0}} \norm{R_{v}\ket{\phi}}.
    \end{aligned}
\end{equation}
Substituting $\ket{\phi} = (U\otimes I_{\dbreg})\ket{\psi}$, since $(U\otimes I_{\dbreg})$ commutes with $Q_{v-1},R_v$, we have $\norm{Q_{v-1}\ket{\phi}} = \norm{(U\otimes I_{\dbreg})Q_{v-1}\ket{\psi}} = \norm{Q_{v-1}\ket{\psi}}$. Similarly, $\norm{R_v\ket{\phi}} = \norm{R_v\ket{\psi}}$. Finally, substituting these into \Cref{eq:phi-progress}, we have
\begin{equation*}
    \begin{aligned}
        \norm{P_v\cdot\cpho\cdot(U\otimes I_{\dbreg})\cdot(I-P)\ket{\psi}} \leq& \sqrt{\frac{v-1}{N_0}}\norm{Q_{v-1}\ket{\psi}} + 3\sqrt{\frac{v-1}{N_0}} \norm{R_{v}\ket{\psi}}\\
        \leq& \sqrt{\frac{v-1}{N_0}}\norm{\Pi_{v-1}\ket{\psi}} + 3\sqrt{\frac{v-1}{N_0}} \norm{\Pi_{v}\ket{\psi}}.
    \end{aligned}
\end{equation*}

\end{proof}

We now proceed to the main helper lemma of this section.


\begin{lemma}[Helper lemma on bounded capacity collision finding]
\label{remember-tuples} 
For any fixed label function $\lab:[M]\to [N]$, let $\As$ be a quantum algorithm making $T$ queries to a compressed phase oracle, its goal is to find a pair of label-$y^*$ collision. Define $V$ according to \Cref{def:non-bot size}. Then its success probability is at most $O(T^2 V / N_0)$, where $N_0$ is the range size. Also, the constant in $O$ does not depend on $\lab$.
\end{lemma}

The lemma basically shows that if we want a quantum algorithm to find collisions efficiently, the average non-$\bot$ entry number has to be sufficiently large. If an algorithm achieves the optimal $T=O\left(N_0^{1/3}\right)$ running time, then during its interaction with the compressed oracle, the database has an average of $O(T)$ non-$\bot$ entry. In other words, it is impossible to find an algorithm that can achieve optimal time efficiency while only touching few entries in the database.

\begin{proof}
    Define $\Pi_{\sf{col}}$ as a projector onto all basis states $\ket{x}_{\qryreg}\ket{u}_{\ansreg}\ket{w}_{\auxreg}\ket{r}_{\outreg}\otimes\ket{D}_{\dbreg}$ such that $r=(r_1,r_2)$ satisfies $D(r_1)=D(r_2)\neq \bot,\lab(r_1)=\lab(r_2)=y^*, r_1\neq r_2$. The projector describes the case when the algorithm $\As$ outputs a valid pair of collisions that could be verified by checking the database, and the probability of such case equals to $\norm{\Pi_{\sf{col}}\ket{\Phi_T}}^2$, while recall that $\ket{\Phi_T}$ is the final state after $T$-th query.
    By \Cref{lem:win-prob-db-algo}, we have 
    \begin{equation}
    \label{eq:win-prob}
        \sqrt{\Pr[\As\ wins]} \leq \norm{\Pi_{\sf{col}}\ket{\Phi_T}} + \sqrt{2/N_0}.
    \end{equation}

    Notice that the subspace that $\Pi_{\sf{col}}$ projects onto strictly contains the subspace that $P$ projects onto, as $P$ only requires the existence of a collision inside $D$, while $\Pi_{\sf{col}}$ also requires the output $r$ to be that collision. In other words, we have $\Pi_{\sf{col}}P = P\Pi_{\sf{col}} = \Pi_{\sf{col}}$. Since projectors do not increase the norm of a vector, we have
    \begin{equation}
    \label{eq:db-prob}
        \norm{\Pi_{\sf{col}}\ket{\Phi_T}} = \norm{\Pi_{\sf{col}}P\ket{\Phi_T}} \leq \norm{P\ket{\Phi_T}}.
    \end{equation}
    
    We then focus on the projected state $P\ket{\Phi_T}$,
    \begin{align*}
        &P\ket{\Phi_T}\\
        =& P\cdot\cpho\cdot (U_{T-1}\otimes I_{\dbreg})\cdot\cpho\cdots \cpho\cdot (U_{0}\otimes I_{\dbreg}) \ket{\Phi_0}\\
        =& P\cdot\cpho\cdot (U_{T-1}\otimes I_{\dbreg})\cdot(P+I-P)\cdot\cpho\cdots (P+I-P)\cdot\cpho\cdot (U_{0}\otimes I_{\dbreg})\cdot(P+I-P) \ket{\Phi_0}\\
        =& \sum_{i=1}^T P\cdot\cpho\cdot (U_{T-1}\otimes I_{\dbreg})\cdot P\cdot\cpho\cdot (U_{T-2}\otimes I_{\dbreg})\cdots P\cdot\cpho\cdot (U_{i-1}\otimes I_{\dbreg})\cdot(I-P)\\
        &\cdot\cpho\cdot (U_{i-2}\otimes I_{\dbreg})\cdot\cpho\cdots \cpho\cdot (U_{0}\otimes I_{\dbreg})\ket{\Phi_0}.
    \end{align*}
    The summation in the last equation is over $i$, the last oracle query that the state is in $I-P$ after the query. Strictly speaking, there should also be a term $P\ket{\Phi_{0}}$ added to the right hand side, but as $\ket{\Phi_{0}}$ stands for an empty database, we have $\ket{\Phi_{0}}=0$, so we omit this term. For $i\in [T]$, we define $\ket{\phi_i}$ to be the state projected onto the case that: (1) before the $i$-th query there are no collisions in the database; (2) the $i$-th query forms a collision. Specifically,
    \begin{equation*}
        \ket{\phi_i} = P\cdot\cpho\cdot (U_{i-1}\otimes I_{\dbreg})\cdot(I-P)\cdot\cpho\cdot (U_{i-2}\otimes I_{\dbreg})\cdot\cpho\cdots \cpho\cdot (U_{0}\otimes I_{\dbreg}) \ket{\Phi_0}.
    \end{equation*}

    Then, we have
    \begin{equation*}
        P\ket{\Phi_T} = \sum_{i=1}^T P\cdot\cpho\cdot (U_{T-1}\otimes I_{\dbreg})\cdots P\cdot\cpho\cdot (U_{i}\otimes I_{\dbreg})\ket{\phi_{i}}.
    \end{equation*}

    Take the norm of the state, and apply the Cauchy-Schwartz inequality, we have
    \begin{align*}
        \norm{P\ket{\Phi_T}}^2 \leq& T\sum_{i=1}^T \norm{ P\cdot\cpho\cdot (U_{T-1}\otimes I_{\dbreg})\cdots P\cdot\cpho\cdot (U_{i}\otimes I_{\dbreg})\ket{\phi_{i}}}^2\\
        \leq& T\sum_{i=1}^T \norm{\ket{\phi_{i}}}^2\\
        =& T\sum_{i=1}^T \norm{\sum_{v=0}^{M}\Pi_v\ket{\phi_{i}}}^2.
    \end{align*}

    For $i\in [T]$, $v\in [0,M]$ we define $\ket{\phi_{i,v}}=\Pi_v\ket{{\phi_i}}$ to represent the case when there are $v$ non-$\bot$ entries in the database, by the orthogonality of $\Pi_v$, we have
    \begin{equation}
    \label{eq:sum-phi-iv}
        \norm{P\ket{\Phi_T}}^2 \leq T\sum_{i=1}^T \norm{\sum_{v=0}^{M}\ket{\phi_{i,v}}}^2 = T\sum_{i=1}^T\sum_{v=0}^{M} \norm{\ket{\phi_{i,v}}}^2.
    \end{equation}

    Now we focus on how to bound $\norm{\ket{\phi_{i,v}}}$. For arbitrary $i\in [T]$, $v\in [0,M]$, we have 
    \begin{align*}
        \ket{\phi_{i,v}} =& P_v\cdot\cpho\cdot (U_{i-1}\otimes I_{\dbreg})\cdot(I-P)\cdot\cpho\cdot (U_{i-2}\otimes I_{\dbreg})\cdot\cpho\cdots \cpho\cdot (U_{0}\otimes I_{\dbreg}) \ket{\Phi_0}\\
        =& P_v\cdot\cpho\cdot (U_{i-1}\otimes I_{\dbreg})\cdot(I-P) \ket{\Phi_{i-1}}.
    \end{align*}

    Also notice that $P_v=0$ when $v=0$, as a database without non-$\bot$ entry cannot contain collisions. Then, $\ket{\phi_{i,0}}=0$. For other terms when $v>0$, by applying \Cref{lem:progress-measure}, we have 
    \begin{align*}
        \norm{\ket{\phi_{i,v}}} =& \norm{P_v\cdot\cpho\cdot (U_{i-1}\otimes I_{\dbreg})\cdot(I-P) \ket{\Phi_{i-1}}}\\
        \leq& 3\sqrt{\frac{v-1}{N_0}}\norm{\Pi_v\ket{\Phi_{i-1}}} + \sqrt{\frac{v-1}{N_0}}\norm{\Pi_{v-1}\ket{\Phi_{i-1}}}\\
        =& 3\sqrt{\frac{v-1}{N_0}}\sqrt{p_{i-1,v}} + \sqrt{\frac{v-1}{N_0}}\sqrt{p_{i-1,v-1}},
    \end{align*}
    where $p_{i,v}$ is defined in \Cref{def:non-bot size}. Applying another Cauchy-Schwartz inequality, we have
    \begin{equation*}
        \norm{\ket{\phi_{i,v}}}^2 \leq 2\left(\frac{9(v-1)}{N_0}p_{i-1,v} + \frac{v-1}{N_0} p_{i-1,v-1} \right).
    \end{equation*}
    Substituting that into \Cref{eq:sum-phi-iv}, we have
    \begin{equation*}
        \begin{aligned}
            \norm{P\ket{\Phi_T}}^2 \leq& T\sum_{i=1}^T\sum_{v=0}^{M} \norm{\ket{\phi_{i,v}}}^2 = T\sum_{i=1}^T\sum_{v=1}^{M} \norm{\ket{\phi_{i,v}}}^2\\
            \leq& T\sum_{i=1}^T\sum_{v=1}^{M}\frac{18(v-1)}{N_0}p_{i-1,v} + T\sum_{i=1}^T\sum_{v=1}^{M}\frac{2(v-1)}{N_0}p_{i-1,v-1}\\
            \leq& \frac{18T}{N_0}\sum_{i=1}^T\sum_{v=1}^{M}v\cdot p_{i-1,v} + \frac{2T}{N_0}\sum_{i=1}^T\sum_{v=0}^{M-1} v\cdot p_{i-1,v}\\
            \leq& \frac{20T}{N_0}\sum_{i=1}^T\sum_{v=0}^{M}v\cdot p_{i-1,v}\\
            =& \frac{20T}{N_0} (TV - V_T + V_0).
        \end{aligned}
    \end{equation*}
    Notice that according to \Cref{rem:non-bot-entry}, $V_i = \sum_{v=0}^{i} v\cdot p_{i,v}$, then $V_0 = 0$. Therefore, we have
    \begin{equation}
    \label{eq:T2V}
        \norm{P\ket{\Phi_T}}^2 \leq \frac{20T}{N_0} (TV - V_T) \leq 20\frac{T^2V}{N_0},
    \end{equation}
    as $V_T$ being a sum of probabilities has to be non-negative.
    Combining \Cref{eq:T2V,eq:win-prob,eq:db-prob}, we have
    \begin{equation*}
        \Pr[\As\ wins] \leq \left(\norm{P\ket{\Phi_T}} + \sqrt{\frac{2}{N_0}} \right)^2 \leq \left(\sqrt{\frac{20T^2V}{N_0}}+\sqrt{\frac{2}{N_0}} \right)^2 = O\left(\frac{T^2V}{N_0}\right).
    \end{equation*}

\end{proof}

\section{\texorpdfstring{$\ell$-Nested collision finding and upper bounds}{l-Nested collision finding and upper bounds}}
\label{sec:nested_col}
\begin{definition}[$\ell$-Nested Collision Finding]
\label{def:nested_col}
Let $N$ be the size of the range, let $N_0$ and $M$ be polynomials of $N$. For any constant integer $\ell>0$ and any target distribution $\cal{D}$, the following problem is called an $\ell$-Nested Collision Finding Problem: 
\begin{itemize}
    \item Input: two random oracles $H: [M] \to [N], G: [M]^\ell \to [N_0]$ and the target $y\gets\cal{D}$.
    \item The goal is to find two different $\ell$-tuples $(x_1,x_2,\cdots,x_\ell),(x'_1,x'_2,\cdots,x'_\ell)$ such that 
    \begin{itemize}
        \item $0\leq x_1<x_2<\cdots<x_\ell<M$, and $0\leq x'_1<x'_2<\cdots<x'_\ell<M$. We will call such tuples valid tuples.
        \item $\sum_{i=1}^{\ell} H(x_i)\equiv\sum_{i=1}^{\ell} H(x'_i)\equiv y\mod N$.
        \item $G(x_1, x_2, \cdots, x_\ell)=G(x'_1, x'_2, \cdots, x'_\ell)$.
    \end{itemize}
\end{itemize}
\end{definition}
\begin{remark}
    Two special cases of the problem would be:
    \begin{itemize}
        \item $\cal{D}$ is the uniform distribution. That is, the inputs are two random functions $H,G$ and a random target $y$.
        \item $\cal{D}$ is the distribution with unique support $y$. For example, when $y=0$ the problem becomes:
        \begin{itemize}
            \item Input: two random oracles $H: [M] \to [N], G: [M]^\ell \to [N_0]$.
            \item The goal is to find two different valid $\ell$-tuples $(x_1,x_2,\cdots,x_\ell),(x'_1,x'_2,\cdots,x'_\ell)$ with sum $0$ on $H$ such that $G(x_1, x_2, \cdots, x_\ell)=G(x_1, x_2, \cdots, x_\ell)$.
        \end{itemize}
    \end{itemize}
\end{remark}
Now we present a classical algorithm and a quantum algorithm that solves this problem. In later sections, we show that without sufficient classical/quantum memory, no algorithm can perform as good as them, in terms of query complexity. Before that, we need the following lemma:
\begin{lemma}\label{lem:tuple_of_sum_y_even}
    For any $y$ and a random function $H:[M]\rightarrow [N]$. If we query $T$ points, the probability that we get $\frac{\binom{T}{\ell}}{2N}$ $\ell$-tuples (formed by queried points) of sum $y$ is constant.
\end{lemma}
\begin{proof}
    Let $S$ be the set of queried points. For $|I|=\ell$ being a subset of $S$ define $K_I$ be event that the sum of the $\ell$-tuple on $H$ formed by points in $I$ is $y$. We have $\mathbb{E}[K_I]=\frac{1}{N}$ and $\text{Var}(K_I)=\frac{N-1}{N^2}$. Notice that $\set{K_I}_{I}$ is pairwise independent. Thus
    \begin{align*}
        \mathbb{E}[K]:=&\mathbb{E}\brackets{\sum_{I\subseteq S,|I|=\ell}K_I}=\sum_{I\subseteq S,|I|=\ell}\mathbb{E}(K_I)=\binom{T}{\ell}\frac{1}{N}\\
        \text{Var}(K):=&\text{Var}\brackets{\sum_{I\subseteq S,|I|=\ell}K_I}=\sum_{I\subseteq S,|I|=\ell}\text{Var}(K_I)=\binom{T}{\ell}\frac{N-1}{N^2}.
    \end{align*}
    Define $K=\sum_{I\subseteq S,|I|=\ell}K_I$ be the total number of $\ell$-tuples of sum $y$, by the Chebyshev bound \Cref{lem:Chebyshev}, 
    \begin{equation*}
        \Pr\left[K\leq \frac{1}{2}\cdot\frac{\binom{T}{\ell}}{N}\right]\leq\frac{\text{Var}(K)}{\frac{1}{4}\cdot\brackets{\mathbb{E}[K]}^2}<\frac{1}{N}.
    \end{equation*}
\end{proof}
\begin{theorem}\label{thm:classical-algorithm}
    For any $\ell\geq 1$, there exists a classical algorithm that uses $\Theta\brackets{N^\frac{1}{\ell}N_0^\frac{1}{2\ell}}$ queries to solve the $\ell$-Nested Collision Finding Problem with constant probability if $N_0=O\brackets{N^\frac{2}{\ell-1}}$.
\end{theorem} 
\begin{proof}
    The algorithm is as follows:
    \begin{enumerate}
        \item In the first step, the algorithm produces $K=\Theta\brackets{N_0^\frac{1}{2}}$ $\ell$-tuples with the sum $y$ by random querying $T_1=\Theta\brackets{K^{\frac{1}{\ell}}N^{\frac{1}{\ell}}}$ points on $H$. This is done by setting an appropriate constant and applying \Cref{lem:tuple_of_sum_y_even}.
        \item Query all the $G$ values of these $\ell$ -tuples, this step costs $T_2=K$ queries. Output a collision if there is one.
    \end{enumerate}
    The success probability of this algorithm is constant since $K=\Theta\brackets{N_0^\frac{1}{2}}$ points of $G$ are sufficient for finding a collision with constant probability. When $N_0=O\brackets{N^\frac{2}{\ell-1}}$, the query complexity of the first step is larger than the one in the second step, thus the overall query complexity when $N_0=O\brackets{N^\frac{2}{\ell-1}}$ is $\Theta\brackets{N^\frac{1}{\ell}N_0^\frac{1}{2\ell}}$.
\end{proof}
\begin{theorem}\label{thm:quantum-algorithm}
    For any $\ell\geq 2$, there exists a quantum algorithm that uses $\Theta\brackets{N^{\frac{2}{2\ell+1}}N_0^{\frac{1}{2\ell+1}}}$ queries that solves the $\ell$-Nested Collision Finding problem with constant probability if $N_0=\Omega\brackets{N^{\frac{1}{\ell}}}$ and $N_0=O\brackets{N^{\frac{3}{\ell-1}}}$.
\end{theorem}
\begin{proof}
    The algorithm is as follows:
    \begin{enumerate}
        \item In the first step, the algorithm produces $K$ $\ell$-tuples with sum $y$ by randomly querying $T_1=\Theta\brackets{K^{\frac{1}{\ell}}N^{\frac{1}{\ell}}}$ points on $H$. This is done by setting an appropriate constant and applying \Cref{lem:tuple_of_sum_y_even}.
        \item Query all $G$ values of these $\ell$-tuples, this step costs $T_2=K$ queries.
        \item Query $T_3$ individual $H$ values different from the points queried in the first step. This will produce $\Theta\brackets{T_3^{\ell-1}}$ $\ell-1$ tuples. 
        \item Now we run a Grover search to find an $\ell$-tuple with:
        \begin{itemize}
            \item Sum $y$ on $H$.
            \item Its $\ell-1$ sub-tuple (the $\ell-1$-tuple formed by the first $\ell-1$ elements) being one of the $\ell-1$-tuples found in step 3. This means that we only need to run the Grover search on the last element.
            \item Its $G$ value collides with one of the $\ell$-tuples in step 2.
        \end{itemize}
        The query complexity is $\Theta\brackets{\sqrt{\frac{N_0}{K}}\cdot\sqrt{\frac{N}{T_3^{\ell-1}}}}$ when $K=O(N_0)$ and $T_3^{\ell-1}=O(N)$ by \Cref{lem:grover_search}. This is because we can run a Grover search on the last element of the $\ell$-tuple. For each such element, the probability that there exists a $\ell-1$-tuple in step 3, with this element, adds up to a sum of $y$ on $H$, is $\frac{T_3^{\ell-1}}{N}$. And the probability that this $\ell$-tuple collides with an $\ell$-tuple in step 2 on $G$ is $\frac{K}{N_0}$.
    \end{enumerate}
    Set $T_3=O\brackets{\brackets{\frac{N_0N}{K}}^\frac{1}{\ell+1}}$ and $K=\Theta\brackets{\frac{N_0^{\frac{\ell}{2\ell+1}}}{N^{\frac{1}{2\ell+1}}}}=\Omega(1)$ because $N_0=\Omega\brackets{N^{\frac{1}{\ell}}}$. Now we have
    \qipeng{we should also set $N < N_0^\ell$?}
    \begin{itemize}
        \item Step 1 costs $\Theta\brackets{N^{\frac{2}{2\ell+1}}N_0^{\frac{1}{2\ell+1}}}$ queries.
        \item Step 2 costs $\Theta\brackets{\frac{N_0^{\frac{\ell}{2\ell+1}}}{N^{\frac{1}{2\ell+1}}}}$ queries. To make this step not a bottleneck, we need $N_0=O\brackets{N^{\frac{3}{\ell-1}}}$.
        \item Step 3 costs $\Theta\brackets{N^{\frac{2}{2\ell+1}}N_0^{\frac{1}{2\ell+1}}}$ queries. We also need $T_3^{\ell-1}=O(N_0)$. Adding the constraint $N_0=O\brackets{N^{\frac{3}{\ell-1}}}$ again.
        \item Step 4 costs $\Theta\brackets{N^{\frac{2}{2\ell+1}}N_0^{\frac{1}{2\ell+1}}}$ queries.
    \end{itemize}
\end{proof}
\section{\texorpdfstring{The time bound for finding $K$ $\ell$-tuples with the same sum}{The time bound for finding K l-tuples with the same sum}}
In this section, we first prove bounds on the probability that an algorithm on $T$ oracle queries, produces $K$ $\ell$-tuples with the same sum. Here, the sum of an $\ell$-tuple $(x_1,x_2,\cdots,x_\ell)$ is defined by $\sum_{i=1}^{\ell}H(x_i)\mod N$ where $H$ is the random oracle which will be clear from the content. We first prove the case where the algorithm is classical to demonstrate our idea. Then we shift to the quantum case and use the compressed oracle technique to derive a similar bound. 
\subsection{The classical bound}
\begin{theorem}\label{thm:classical-time-bound-for-finding-tuples-database}
    Let $H:[M]\rightarrow[N]$ be a random oracle with sufficiently large polynomial $M$ in $N$ and let $K=\Omega(\log N)^\ell$. For any classical algorithm with $T=\Omega(K)$ oracle queries, the probability that there are $K$ distinct $\ell$-tuples consist only of queried points with the same sum is at most $O\brackets{\frac{T^{\ell}}{K^\frac{1}{\ell}N}}^{K^\frac{1}{\ell}}$.
\end{theorem}
\begin{remark}
    At first glance, the bound may seem strange due to the superscript $\frac{1}{\ell}$ on $K$. This is in fact natural, because if an algorithm queries $O\brackets{K^\frac{1}{\ell}}$ points, with probability $\brackets{\frac{1}{N}}^{O\brackets{K^\frac{1}{\ell}}}$ it gets value $0$ on all queried points and the task is done automatically. So one would expect the term $K^{\frac{1}{\ell}}$ to exists in reasonable bounds.
\end{remark}
\begin{proof}
    In the classical case, query strategy is meaningless since each point of the random oracle is independent. Thus we assume the algorithm queries points sequentially one-by-one. We prove by induction on $\ell$. For $\ell=1$, let $y\in[N]$ and $\Delta^y_{t,k}$ be the probability that after $t$ queries one can find out at least $k$ distinct queried points $x_1,x_2,\cdots,x_k$ such that $H(x_1)=H(x_2)=\cdots=H(x_k)=y$. We have
    \begin{equation*}
        \Delta^y_{t,k}\leq\Delta^y_{t-1,k}+\frac{1}{N}\cdot\Delta^y_{t-1,k-1}.
    \end{equation*}
    Thus
    \begin{equation*}
        \Delta^y_{t,k}\leq \binom{t}{k}\brackets{\frac{1}{N}}^k.
    \end{equation*}
    The probability that there are $K$ distinct points $x_1,x_2,\cdots,x_k$ of the same value is at most
    \begin{align*}
        &\sum_{y\in[N]}\Delta^y_{T,K}\\
        \leq& N\binom{T}{K}\brackets{\frac{1}{N}}^K\\
        =& O\brackets{\frac{T}{KN}}^K.
    \end{align*}
    Suppose that the statement holds for $\ell=\ell^*-1$, now we prove that the statement holds for $\ell=\ell^*$. After $T$ queries there are two cases: for the first case, one can find out at least $K^{\frac{\ell-1}{\ell}}$ distinct $\ell-1$ tuples consist of only queried points with the same sum. The probability of this case, by induction, is $O\brackets{\frac{t^{\ell-1}}{K^\frac{1}{\ell}N}}^{K^\frac{1}{\ell}}$.\zihan{what is $k_0$?} The rest form the second case. Let $y\in[N]$. Now we want to calculate how many fractions of the second case satisfy that at some point one can find out at least $K$ distinct $\ell$ tuples with sum $y$. Note that number of distinct $\ell-1$ tuples with the same sum is at most $K^{\frac{\ell-1}{\ell}}$ since we are in the second case. Thus we can relax the requirement by assuming that we can find $K^{\frac{\ell-1}{\ell}}$ $\ell$-tuples of sum $y$ whenever a new point is added for which there exists an $\ell-1$ tuple in the database such that its sum plus the value of the new entry is $y$. Let $\Delta^y_{t,k}$ denote the probability that after $t$ queries there exists $k$ queries out of these $t$ queries such that queried point with all the previous queried points can form an $\ell$-tuple of sum $y$ that contains that queried point. Since there are at most $(t+k_0)^{\ell-1}$ possible sum for $\ell-1$ tuples consists of only queried points we have,
    \begin{equation*}
        \Delta^y_{t,k}\leq\Delta^y_{t-1,k}+\frac{t^{\ell-1}}{N}\cdot\Delta^y_{t-1,k-1}.
    \end{equation*}
    Thus 
    \begin{equation*}
        \Delta^y_{t,k}\leq \binom{t}{k}\brackets{\frac{t^{\ell-1}}{N}}^{k}.
    \end{equation*}
    By union bound, the probability that an algorithm outputs $K$ distinct $\ell$-tuples with the same sum is bounded by the sum of probability that this algorithm outputs $K$ distinct $\ell$-tuples with sum $y$ over all $y\in[N]$. If an algorithm outputs $K$ distinct $\ell$-tuples with sum $y$, by above analysis, either one can find out $K^{\frac{\ell-1}{\ell}}$ distinct $\ell-1$ tuples with the same sum in its queried points or there exists $K^{\frac{1}{\ell}}$ queries out of these $t$ queries such that the queried point with all the previous queried points can form an $\ell$-tuple of sum $y$ that contains the queried point. Thus, the overall success probability of any algorithm is bounded by
    \begin{align*}
        &\sum_{y\in[N]}\Delta^y_{T,K^{\frac{1}{\ell}}}+O\brackets{\frac{T^{\ell-1}}{K^\frac{1}{\ell}N}}^{K^\frac{1}{\ell}}\\
        \leq& N\binom{T}{K^{\frac{1}{\ell}}}\brackets{\frac{T^{\ell-1}}{N}}^{K^{\frac{1}{\ell}}}+O\brackets{\frac{T^{\ell-1}}{K^\frac{1}{\ell}N}}^{K^\frac{1}{\ell}}\\
        =& O\brackets{\frac{T^{\ell}}{K^\frac{1}{\ell}N}}^{K^\frac{1}{\ell}}+O\brackets{\frac{T^{\ell-1}}{K^\frac{1}{\ell}N}}^{K^\frac{1}{\ell}}\\
        =& O\brackets{\frac{T^{\ell}}{K^\frac{1}{\ell}N}}^{K^\frac{1}{\ell}}.
    \end{align*}
    Note that since $\ell>0$ is a constant, doing induction on it preserves the $O(\cdot)$ notation.
\end{proof}
\begin{theorem}\label{thm:classical-time-bound-for-finding-tuples-final}
    Let $H:[M]\rightarrow[N]$ be a random oracle with sufficiently large polynomial $M$ in $N$ and let $K=\Omega(\log N)^{\ell}$. For any classical algorithm with $T=\Omega(K)$ oracle queries, the probability that it outputs $K$ distinct $\ell$-tuples with the same sum is at most $O\brackets{\frac{T^{\ell}}{K^\frac{1}{\ell}N}}^{K^\frac{1}{\ell}}$.
\end{theorem}
\begin{proof}
    For any algorithm $\cal{A}$ with $T=\Omega(K)$ oracle queries, we can construct an algorithm $\cal{A}'$ such that all outputted points are queried at some point by running $\cal{A}$ first and querying any index that does not exist in its database at the end of the $\cal{A}$. By \Cref{thm:classical-time-bound-for-finding-tuples-database}, we prove the statement.
\end{proof}
\subsection{The quantum bound}\label{sec:time_bound_l_tuples}
In this section, we are going to prove a similar result in the quantum setting. Suppose that we start with an \textbf{arbitrary} state:
\begin{equation*}
    \ket{\psi_{{\sf st}}}=\sum_{x,u,w,r,D}\alpha_{x,u,w,r,D}\ket{x}_{\qryreg}\ket{u}_{\ansreg}\ket{w}_{\auxreg}\ket{r}_{\outreg}\ket{D}_{\dbreg}.
\end{equation*}
Define the following operators:
\begin{itemize}
    \item We extend the phase oracle unitary $\cpho$ and standard decompose unitary $\stddecomp_x$ to the case where $x=\bot$. In this case these unitaries acts as an identity over all registers. As a special reminder, not like the case in the compressed oracle setting, here $\ket{\bot}_{\qryreg}$ is orthogonal to any $\ket{x}_{\qryreg}$ for $x\in[M]$.
    \item Define $\Pi_{\geq k}$ as the projector that projects to the subspace spanned by all state $\ket{x,u,w,r}\ket{D}$ such that $D$ contains at least $k$ non-bot entries. Define $\Pi_{=k},\Pi_{<k}$ in a similar manner.
    \item Let $y\in[N]$, define $\Pi_{\geq k}^{\ell,y}$ as the projector that projects to the subspace spanned by all state $\ket{x,u,w,r}\ket{D}$ such that $D$ contains at least $k$ distinct $\ell$-tuples with sum $y$. Define $\Pi_{=k}^{\ell,y},\Pi_{<k}^{\ell,y}$ in the similar manner. Also if $y=*$ it means any $y$, for example, $\Pi_{\geq k}^{\ell,*}$ projects to the subspace spanned by all state $\ket{x,u,w,r}\ket{D}$ such that $D$ contains at least $k$ distinct $\ell$-tuples with the same sum.
    \item Let $y\in[N]$ and $\cal{O}$ be either the compressed phase oracle $\cpho$ or $\stddecomp$, define $\Pi_{{\sf inc}}^{y,\cal{O}}$ as the operator that ``projects'' to the state that after calling $\cal{O}$ the number of distinct $\ell$-tuples with sum $y$ increases. That is, 
    $\Pi_{{\sf inc}}^{y,\cal{O}}=\sum_k\brackets{\cal{O}^{-1}\Pi_{>k}^{\ell,y}\cal{O}\Pi_{=k}^{\ell,y}}$. \textbf{Note that although we use $\Pi$, this may not be a projection}.\zikuan{I'm not sure if this is a projection or not. But this affects nothing.}
\end{itemize}
\begin{theorem}\label{thm:quantum-time-bound-for-finding-tuples-database}
    For $\ell>0$ be an integer. Let $k_0,k_1,k_2,\cdots,k_\ell\geq 0$ be integers. Assume the following:
    \begin{itemize}
        \item The number of non-bot entries in the database register is at most $k_0$. In other words,
        \begin{equation*}
            \Pi_{>k_0}\ket{\psi_{{\sf st}}}=0.
        \end{equation*}
        \item For $i=1,2,\cdots,\ell$, the number of $i$-tuples with the same sum is at most $k_i$. In other words, for $i=1,2,\cdots,\ell-1$
        \begin{equation*}
            \Pi_{>k_i}^{i,*}\ket{\psi_{{\sf st}}}=0.
        \end{equation*}
    \end{itemize}
    Let $M$ be a sufficiently large and let $K>\max_{i=1}^{\ell}k_i$. Let $K_{{\sf sol}}$ be the only real non-negative root of the following function
    \begin{equation*}
        f_{K,k_1,k_2,\cdots,k_\ell}(x)=x^\ell+\sum_{i=0}^{\ell-1}k_{\ell-i}x^i-K.
    \end{equation*}
    Note that when $x=0$, $f_{K,k_1,k_2,\cdots,k_\ell}(x)=k_\ell-K\leq 0$ and when $x\rightarrow\infty$, $f(x)\rightarrow\infty$. Thus there must be a real non-negative root. Also we know that this root is unique since all coefficients except the constant term are positive, meaning that $f_{K,k_1,k_2,\cdots,k_\ell}(x)$ is monotonically increasing when $x\geq 0$.\\
    If $K_{{\sf sol}}=\Omega(\log N)$ then for any quantum algorithm with $T=\Omega(K)$ compressed phase oracle queries, and $T$ additional $\stddecomp$ after that, the probability that \textbf{at some point during the algorithm} there are $K$ distinct $\ell$-tuples with the same sum in the database register under the view of compressed oracle is at most $O\brackets{\frac{T^2(T+k_0)^{\ell-1}}{K_{{\sf sol}}^2N}}^{K_{{\sf sol}}}$.
    More specifically,
    \begin{equation*}
        \abs{\brackets{\sum_{i=0}^{2T-1}\prod_{i+1}^{j=2T-1}\brackets{\cal{O}_jU_j}\Pi_{\geq K}^{\ell,*}\cal{O}_iU_i\prod_{0}^{j=i-1}\brackets{\Pi_{<K}^{\ell,*}\cal{O}_jU_j}}\ket{\psi_{{\sf st}}}}^2=O\brackets{\frac{T^{\ell+1}}{K_{{\sf sol}}^2N}}^{K_{{\sf sol}}}.
    \end{equation*}
    where we additionally define $\cal{O}_i=\cpho$ when $i<T$ and $\cal{O}_i=\stddecomp $ otherwise. Here we use $\prod_{0}^{j=i-1}$ to denote the product of terms in the reverse order. For example for unitaries $\set{U_j}$ we have $\prod_{0}^{j=i-1}U_j=U_{i-1}U_{i-2}\cdots U_0$. To explain this term, $i$ (0-index) goes over all possible positions that indicate the first time that we can find $K$ $\ell$-tuples of the same sum. The later part $\prod_{0}^{j=i-1}\brackets{\Pi_{<K}^{\ell,*}\cal{O}_jU_j}$ denotes the process before the $i$-th query where we cannot find $K$ $\ell$-tuples of the same sum. The middle term $\Pi_{\geq K}^{\ell,*}\cal{O}_iU_i$ means that the first time that we can find $K$ $\ell$-tuples of the same sum is after the $i$-th query.
\end{theorem}
\begin{proof}
    We prove by induction on $\ell$. For $\ell=1$, define $\Delta^y_{t,k}$ as the amplitude on states after $t$ queries (from now on, if $t>T$ it means after the first $t-T$ $\stddecomp$) and at some point one can find out at least $k$ distinct points with value $y$ in the database,
    \begin{equation*}
        \Delta^y_{t,k}=\abs{\brackets{\sum_{i=0}^{t-1}\prod_{i+1}^{j=t-1}\brackets{\cal{O}_jU_j}\Pi_{\geq k}^{1,y}\cal{O}_iU_i\prod_{0}^{j=i-1}\brackets{\Pi_{<k}^{1,y}\cal{O}_jU_j}}\ket{\psi_{{\sf st}}}}.
    \end{equation*}
    Since we assume that the number of distinct points with every value $y\in[N]$ in the database of $\ket{\psi_{{\sf st}}}$ are at most $k_1$ we have:
    \begin{equation}\label{eq:delta-l=1-initial-1s}
        \Delta^{y}_{0,0}=\Delta^{y}_{0,1}=\cdots=\Delta^{y}_{0,k_1}=1.
    \end{equation}
    And 
    \begin{equation}\label{eq:delta-l=1-initial-0s}
        0=\Delta^{y}_{0,k_1+1}=\Delta^{y}_{0,k_1+2}=\cdots.
    \end{equation}
    Now we try to calculate the recursion formula for $\Delta^y_{t,k}$:
    \begin{align*}
        \Delta^y_{t,k}=&\abs{\brackets{\sum_{i=0}^{t-1}\prod_{i+1}^{j=t-1}\brackets{\cal{O}_jU_j}\Pi_{\geq k}^{1,y}\cal{O}_iU_i\prod_{0}^{j=i-1}\brackets{\Pi_{<k}^{1,y}\cal{O}_jU_j}}\ket{\psi_{{\sf st}}}}\\
        \leq&\abs{\brackets{\sum_{i=0}^{t-2}\prod_{i+1}^{j=t-1}\brackets{\cal{O}_jU_j}\Pi_{\geq k}^{1,y}\cal{O}_iU_i\prod_{0}^{j=i-1}\brackets{\Pi_{<k}^{1,y}\cal{O}_jU_j}}\ket{\psi_{{\sf st}}}}+\abs{\Pi_{\geq k}^{1,y}\cal{O}_{t-1}U_{t-1}\prod_{0}^{j=t-2}\brackets{\Pi_{<k}^{1,y}\cal{O}_jU_j}\ket{\psi_{{\sf st}}}}\\
        =&\abs{\brackets{\sum_{i=0}^{t-2}\prod_{i+1}^{j=t-2}\brackets{\cal{O}_jU_j}\Pi_{\geq k}^{1,y}\cal{O}_iU_i\prod_{0}^{j=i-1}\brackets{\Pi_{<k}^{1,y}\cal{O}_jU_j}}\ket{\psi_{{\sf st}}}}+\abs{\Pi_{\geq k}^{1,y}\cal{O}_{t-1}U_{t-1}\prod_{0}^{j=t-2}\brackets{\Pi_{<k}^{1,y}\cal{O}_jU_j}\ket{\psi_{{\sf st}}}}\\
        =&\Delta^y_{t-1,k}+\abs{\Pi_{\geq k}^{1,y}\cal{O}_{t-1}\Pi_{<k}^{1,y}U_{t-1}\cal{O}_{t-2}U_{t-2}\prod_{0}^{j=t-3}\brackets{\Pi_{<k}^{1,y}\cal{O}_jU_j}\ket{\psi_{{\sf st}}}}\\
        \leq&\Delta^y_{t-1,k}+\abs{\Pi_{\geq k}^{1,y}\cal{O}_{t-1}\Pi_{<k}^{1,y}\ket{\psi_t}}\abs{\prod_{0}^{j=t-2}\brackets{\Pi_{<k}^{1,y}\cal{O}_jU_j}\ket{\psi_{{\sf st}}}}
    \end{align*}
    where $\ket{\psi_t}\propto\prod_{0}^{j=t-2}\brackets{\Pi_{<k}^{1,y}\cal{O}_jU_j}\ket{\psi_{{\sf st}}}$ is the normalized state before the $t$-th oracle query. The two terms in the first inequality correspond to the case that $i\leq t-2$ and the case that $i=t-1$. Define the following projectors:
    \begin{itemize}
        \item $Q^{\ell,y}_{=k}$ projects to all states spanned by $\ket{x,u,w,r}\ket{D}$ such that there are $k$ distinct $\ell$-tuples of sum $y$, $D(x)=\bot$ and $u\neq 0$.
        \item $R^{\ell,y}_{=k}$ projects to all states spanned by $\ket{x,u,w,r}\ket{D}$ such that there are $k$ distinct $\ell$-tuples of sum $y$, $D(x)\neq\bot$ and $u\neq 0$.
        \item $S^{\ell,y}_{=k}$ projects to all states spanned by $\ket{x,u,w,r}\ket{D}$ such that there are $k$ distinct $\ell$-tuples of sum $y$ and $u=0$.
    \end{itemize}
    Because each oracle query add at most one entry in the database we have that
    \begin{align*}
        \Delta^y_{t,k}\leq&\Delta^y_{t-1,k}+\abs{\Pi_{\geq k}^{1,y}\cal{O}_{t-1}\Pi_{<k}^{1,y}\ket{\psi_t}}\abs{\prod_{0}^{j=t-2}\brackets{\Pi_{<k}^{1,y}\cal{O}_jU_j}\ket{\psi_{{\sf st}}}}\\
        \leq&\Delta^y_{t-1,k}+\abs{\Pi_{\geq k}^{1,y}\cal{O}_{t-1}\brackets{Q_{=k-1}^{1,y}+R_{=k-1}^{1,y}+S_{=k-1}^{1,y}}\ket{\psi_t}}\cdot\Delta^y_{t-1,k-1}.
    \end{align*}
    The reason that we can substitute the term $\abs{\prod_{0}^{j=t-2}\brackets{\Pi_{<k}^{1,y}\cal{O}_jU_j}\ket{\psi_{{\sf st}}}}$ with $\Delta^y_{t-1,k-1}$ is that the only part of $\ket{\psi_t}$ that can pass $\Pi_{\geq k}^{1,y}$ is inside $\Pi_{=k-1}^{1,y}$ (which can be decomposed into $Q_{=k-1}^{1,y}+R_{=k-1}^{1,y}+S_{=k-1}^{1,y}$) since each query adds at most one entry in the database. Thus we can project $\ket{\psi_t}$ onto $\Delta^y_{t-1,k-1}$ which has an amplitude of $\Delta^y_{t-1,k-1}$. To bound $\abs{\Pi_{\geq k}^{1,y}\cal{O}_{t-1}\brackets{Q_{=k-1}^{1,y}+R_{=k-1}^{1,y}+S_{=k-1}^{1,y}}\ket{\psi_t}}$, we need to consider two cases:
    \begin{itemize}
        \item When $t\leq T$, $\cal{O}_{t-1}=\cpho$ so we invoke \Cref{each-query-in-compress-oracle}. For $\ket{x,u,w,r}\ket{D}$ in $S_{=k-1}^{1,y}$ the state is unchanged after $\cal{O}_{t-1}$. For $\ket{x,u,w,r}\ket{D}$ in $Q_{=k-1}^{1,y}$, we have $\ket{D(x)}_{\dbreg_x}=\sum_{y'\in[N]}\frac{(-1)^{\langle u,y'\rangle}}{\sqrt{N}}\ket{y'}$ after $\cal{O}_{t-1}$. Since only one uniform random value $y'$ is added, we have $\abs{\Pi_{\geq k}^{1,y}\cal{O}_{t-1}Q_{=k-1}^{1,y}\ket{\psi_t}}\leq\sqrt{\frac{1}{N}}$. For $\ket{x,u,w,r}\ket{D}$ in $R_{=k-1}^{1,y}$, we have $\ket{D(x)}_{\dbreg_x}=\frac{(-1)^{\langle u,D(x)\rangle}}{\sqrt{N}}\ket{\bot}+\frac{1+(-1)^{\langle u,D(x)\rangle}(N-2)}{N}\ket{D(x)}+\sum_{y'\in[N]\backslash\set{D(x)}}\frac{1-(-1)^{\langle u,y'\rangle}-(-1)^{\langle u,D(x)\rangle}}{N}\ket{y'}$ after $\cal{O}_{t-1}$. Since the only case where a new value $y'$ is added has amplitude $\frac{1-(-1)^{\langle u,y'\rangle}-(-1)^{\langle u,D(x)\rangle}}{N}$, we have $\abs{\Pi_{\geq k}^{1,y}\cal{O}_{t-1}R_{=k-1}^{1,y}\ket{\psi_t}}\leq\frac{3}{N}$.
        \item When $t>T$, $\cal{O}_{t-1}=\stddecomp $. For $\ket{x,u,w,r}\ket{D}$ in $S_{=k-1}^{1,y}$ the state is unchanged after $\cal{O}_{t-1}$. For $\ket{x,u,w,r}\ket{D}$ in $Q_{=k-1}^{1,y}$, we have $\ket{D(x)}_{\dbreg_x}=\frac{1}{\sqrt{N}}\sum_{y'\in[N]}\ket{y'}$. Since only one uniform random value $y'$ is added, we have $\abs{\Pi_{\geq k}^{1,y}\cal{O}_{t-1}Q_{=k-1}^{1,y}\ket{\psi_t}}\leq\sqrt{\frac{1}{N}}$. For $\ket{x,u,w,r}\ket{D}$ in $R_{=k-1}^{1,y}$, we have $\ket{D(x)}_{\dbreg_x}=\frac{N-1}{N}\ket{D(x)}+\frac{1}{\sqrt{N}}\ket{\bot}-\frac{1}{N}\sum_{y'\in[N]\backslash\set{D(x)}}\ket{y'}$ after $\cal{O}_{t-1}$. Since the only case where a new value $y'$ is added has amplitude $\frac{1}{N}$, we have $\abs{\Pi_{\geq k}^{1,y}\cal{O}_{t-1}R_{=k-1}^{1,y}\ket{\psi_t}}\leq\frac{1}{N}$.
    \end{itemize}  
    Combining all cases above, we obtain the following recursion formula for $\Delta^y_{t,k}$:
    \begin{equation*}
        \Delta^y_{t,k}\leq\Delta^y_{t-1,k}+4\sqrt{\frac{1}{N}}\cdot\Delta^y_{t-1,k-1}
    \end{equation*}
    combine with \Cref{eq:delta-l=1-initial-1s} and \Cref{eq:delta-l=1-initial-0s} we get
    \begin{equation*}
        \Delta^{y}_{t,k}\leq \binom{t}{k-k_1}\brackets{4\sqrt{\frac{1}{N}}}^{k-k_1}
    \end{equation*}
    for $k>k_1$ and $y\in[N]$.
    The probability that at some point during the algorithm there are $K$ distinct points $x_1,x_2,\cdots,x_K$ of the same value in the database is at most
    \begin{align*}
        &\abs{\brackets{\sum_{i=0}^{2T-1}\prod_{i+1}^{j=2T-1}\brackets{\cal{O}_jU_j}\Pi_{\geq K}^{1,*}\cal{O}_iU_i\prod_{0}^{j=i-1}\brackets{\Pi_{<K}^{1,*}\cal{O}_jU_j}}\ket{\psi_{{\sf st}}}}^2\\
        \leq&\abs{\brackets{\sum_{y\in[N]}\sum_{i=0}^{2T-1}\prod_{i+1}^{j=2T-1}\brackets{\cal{O}_jU_j}\Pi_{\geq K}^{1,y}\cal{O}_iU_i\prod_{0}^{j=i-1}\brackets{\Pi_{<K}^{1,y}\cal{O}_jU_j}}\ket{\psi_{{\sf st}}}}^2\\
        \leq&N\sum_{y\in[N]}\abs{\brackets{\sum_{i=0}^{2T-1}\prod_{i+1}^{j=2T-1}\brackets{\cal{O}_jU_j}\Pi_{\geq K}^{1,y}\cal{O}_iU_i\prod_{0}^{j=i-1}\brackets{\Pi_{<K}^{1,y}\cal{O}_jU_j}}\ket{\psi_{{\sf st}}}}^2\\
        \leq&N\sum_{y\in[N]}\brackets{\Delta^y_{2T,K}}^2\\
        \leq& N^2\binom{2T}{K-k_1}^2\brackets{\frac{16}{N}}^{K-k_1}\\
        =&O\brackets{\frac{T^2}{(K-k_1)^2N}}^{K-k_1}.
    \end{align*}
    $N$ is absorbed into the $O(\cdot)$ notation because $K-k_1=\Omega(\log N)$.\\
    Suppose that the statement holds for $\ell=\ell^*-1$, now we prove that the statement holds for $\ell=\ell^*$. Imagine an algorithm that at some point one can find out at least $K$ distinct $\ell$ tuples with the same sum $y$ in its database. This algorithm must satisfy:
    \begin{itemize}
        \item Either at some point of this algorithm, one can find out at least $K_{{\sf thd}}$ distinct $\ell-1$ tuples with the same sum.
        \item Or there exists at least $\frac{K-k_\ell}{K_{{\sf thd}}}$ ``effective'' queries, meaning that after that query the queried point, with known points in the database forms at least one new $\ell$ tuple of sum $y$.
    \end{itemize}
    Intuitively, this is because if an algorithm never stores $K_{{\sf thd}}$ distinct $\ell-1$-tuples with the same sum in its database, each ``effective'' query can at most generate $K_{{\sf thd}}-1$ $\ell$-tuples with sum $y$ thus one need at least $\frac{K-k_\ell}{K_{{\sf thd}}}$ of them.\\
    Now let us formalize our proof. Define $\Gamma_{t,k}^y$ as the amplitude on states after $t$ queries and at some point there one can find at least $K$ $\ell$-tuples with sum $y$ in the database. In the following computation, we iterate over all $c_0,c_1,c_2,\cdots,c_t$ indicating the number of $\ell$-tuples of sum $y$ that one can find in the database.
    \begin{align*}
        \Gamma^y_{t,k}=&\abs{\brackets{\sum_{i=0}^{t-1}\prod_{i+1}^{j=t-1}\brackets{\cal{O}_jU_j}\Pi_{\geq k}^{\ell,y}\cal{O}_iU_i\prod_{0}^{j=i-1}\brackets{\Pi_{<k}^{\ell,y}\cal{O}_jU_j}}\ket{\psi_{{\sf st}}}}\\
        =&\abs{\brackets{\sum_{\substack{c_0,c_1,\cdots,c_t\\ c_0\leq k_\ell\text{ and }\exists i\text{ s.t. }c_i\geq k}}\prod_{0}^{j=t-1}\brackets{\Pi_{=c_{j+1}}^{\ell,y}\cal{O}_jU_j}\Pi_{=c_0}^{\ell,y}}\ket{\psi_{{\sf st}}}}\\
        \leq&\abs{\brackets{\sum_{\substack{c_0,c_1,\cdots,c_t\\ \exists i\text{ s.t. }c_i\geq c_{i-1}+K_{{\sf thd}}}}\prod_{0}^{j=t-1}\brackets{\Pi_{=c_{j+1}}^{\ell,y}\cal{O}_jU_j}\Pi_{=c_0}^{\ell,y}}\ket{\psi_{{\sf st}}}}\\
        &+\abs{\brackets{\sum_{\substack{c_0,c_1,\cdots,c_t\\ \forall i,c_i<c_{i-1}+K_{{\sf thd}}\\ \exists i\text{ s.t. }c_i\geq k}}\prod_{0}^{j=t-1}\brackets{\Pi_{=c_{j+1}}^{\ell,y}\cal{O}_jU_j}\Pi_{=c_0}^{\ell,y}}\ket{\psi_{{\sf st}}}}\\
        \leq&\abs{\brackets{\sum_{i=0}^{t-1}\prod_{i+1}^{j=t-1}\brackets{\cal{O}_jU_j}\Pi_{\geq K_{{\sf thd}}}^{\ell-1,*}\cal{O}_iU_i\prod_{0}^{j=i-1}\brackets{\Pi_{<K_{{\sf thd}}}^{\ell-1,*}\cal{O}_jU_j}}\ket{\psi_{{\sf st}}}}\\
        &+\abs{\brackets{\sum_{\substack{c_0,c_1,\cdots,c_t\\ \sum_{i=1}^{t}\mathbb{I}[c_i>c_{i-1}]\geq \frac{k-k_\ell}{K_{{\sf thd}}}}}\prod_{0}^{j=t-1}\brackets{\Pi_{=c_{j+1}}^{\ell,y}\cal{O}_jU_j}\Pi_{=c_0}^{\ell,y}}\ket{\psi_{{\sf st}}}}
    \end{align*}
    The first term is bounded by induction (we will show it later). Now we bound the second term which are states that experience $\frac{k-k_\ell}{K_{\sf thd}}$ ``jumps''. We redefine $\Delta^y_{t,k}$ as the amplitude on states after $t$ queries and there exists $k$ queries such that the number of distinct $\ell$-tuples with sum $y$ increases after it,
    \begin{equation*}
        \Delta^y_{t,k}=\abs{\brackets{\sum_{\substack{0=p_0<p_1< \\ \cdots<p_k\leq t}}\prod_{p_k}^{j=t-1}\brackets{\cal{O}_jU_j}\prod_{0}^{i=k-1}\brackets{{\sf Inc}^y_{p_{i+1}-1}\prod_{p_i}^{j=p_{i+1}-2}{\sf NonInc}^y_j}}\ket{\psi_{{\sf st}}}}
    \end{equation*}
    where ${\sf Inc}^y_i=\cal{O}_i\Pi_{{\sf Inc}}^{y,\cal{O}_i}U_i$, ${\sf NonInc}^y_i=\cal{O}_i\brackets{I-\Pi_{{\sf Inc}}^{y,\cal{O}_i}}U_i$.
    Iterating over $p_1,\cdots,p_k$ where ``jumps'' happen, we have
    \begin{align*}
        \Delta^y_{t,k}=&\abs{\brackets{\sum_{\substack{0=p_0<p_1< \\ \cdots<p_k\leq t}}\prod_{p_k}^{j=t-1}\brackets{\cal{O}_jU_j}\prod_{0}^{i=k-1}\brackets{{\sf Inc}^y_{p_{i+1}-1}\prod_{p_i}^{j=p_{i+1}-2}{\sf NonInc}^y_j}}\ket{\psi_{{\sf st}}}}\\
        \leq&\abs{\brackets{\sum_{\substack{0=p_0<p_1< \\ \cdots<p_k<t}}\prod_{p_k}^{j=t-1}\brackets{\cal{O}_jU_j}\prod_{0}^{i=k-1}\brackets{{\sf Inc}^y_{p_{i+1}-1}\prod_{p_i}^{j=p_{i+1}-2}{\sf NonInc}^y_j}}\ket{\psi_{{\sf st}}}}\\
        &+\abs{\brackets{\sum_{\substack{0=p_0<p_1< \\ \cdots<p_k=t}}\prod_{0}^{i=k-1}\brackets{{\sf Inc}^y_{p_{i+1}-1}\prod_{p_i}^{j=p_{i+1}-2}{\sf NonInc}^y_j}}\ket{\psi_{{\sf st}}}}\\
        =&\abs{\brackets{\sum_{\substack{0=p_0<p_1< \\ \cdots<p_k\leq t-1}}\prod_{p_k}^{j=t-2}\brackets{\cal{O}_jU_j}\prod_{0}^{i=k-1}\brackets{{\sf Inc}^y_{p_{i+1}-1}\prod_{p_i}^{j=p_{i+1}-2}{\sf NonInc}^y_j}}\ket{\psi_{{\sf st}}}}\\
        &+\abs{\brackets{\sum_{\substack{0=p_0<p_1< \\ \cdots<p_{k-1}<t}}\cal{O}_{t-1}\Pi_{{\sf inc}}^{y,\cal{O}_{t-1}}U_{t-1}\prod_{p_{k-1}}^{j=t-2}{\sf NonInc}^y_j\prod_{0}^{i=k-2}\brackets{{\sf Inc}^y_{p_{i+1}-1}\prod_{p_i}^{j=p_{i+1}-2}{\sf NonInc}^y_j}}\ket{\psi_{{\sf st}}}}\\
        \leq&\Delta^y_{t-1,k}+\abs{\cal{O}_{t-1}\Pi_{{\sf inc}}^{y,\cal{O}_{t-1}}\ket{\psi_t}}\\
        &\cdot\abs{\sum_{\substack{0=p_0<p_1< \\ \cdots<p_{k-1}< t}}{\prod_{p_{k-1}}^{j=t-2}\brackets{\cal{O}_jU_j}\prod_{0}^{i=k-2}\brackets{{\sf Inc}^y_{p_{i+1}-1}\prod_{p_i}^{j=p_{i+1}-2}{\sf NonInc}^y_j}}\ket{\psi_{{\sf st}}}}
    \end{align*}
    where $\ket{\psi_t}\propto\sum_{\substack{0=p_0<p_1< \\ \cdots<p_{k-1}<t}}\brackets{U_{t-1}\prod_{p_{k-1}}^{j=t-2}\brackets{\cal{O}_jU_j}\prod_{0}^{i=k-2}\brackets{{\sf Inc}^y_{p_{i+1}-1}\prod_{p_i}^{j=p_{i+1}-2}{\sf NonInc}^y_j}}\ket{\psi_{{\sf st}}}$ is the state (normalized) before the $t$-th oracle query. Define the following projectors:
    \begin{itemize}
        \item $Q$ projects to all states spanned by $\ket{x,u,w,r}\ket{D}$ such that $D(x)=\bot$ and $u\neq 0$.
        \item $R$ projects to all states spanned by $\ket{x,u,w,r}\ket{D}$ such that $D(x)\neq\bot$ and $u\neq 0$.
        \item $S$ projects to all states spanned by $\ket{x,u,w,r}\ket{D}$ such that $u=0$.
    \end{itemize}
    Use these projectors we have
    \begin{align*}
        \Delta^y_{t,k}\leq&\Delta^y_{t-1,k}+\abs{\cal{O}_{t-1}\Pi_{{\sf inc}}^{y,\cal{O}_{t-1}}\ket{\psi_t}}\\
        &\cdot\abs{\sum_{\substack{0=p_0<p_1< \\ \cdots<p_{k-1}< t}}{\prod_{p_{k-1}}^{j=t-2}\brackets{\cal{O}_jU_j}\prod_{0}^{i=k-2}\brackets{{\sf Inc}^y_{p_{i+1}-1}\prod_{p_i}^{j=p_{i+1}-2}{\sf NonInc}^y_j}}\ket{\psi_{{\sf st}}}}\\
        \leq&\Delta^y_{t-1,k}+\abs{\cal{O}_{t-1}\Pi_{{\sf inc}}^{y,\cal{O}_{t-1}}(Q+R+S)\ket{\psi_t}}\cdot\Delta^y_{t-1,k-1}.
    \end{align*}
    To bound $\abs{\cal{O}_{t-1}\Pi_{{\sf inc}}^{y,\cal{O}_{t-1}}(Q+R+S)\ket{\psi_t}}$, we need to consider two cases:
    \begin{itemize}
        \item When $t\leq T$, $\cal{O}_{t-1}=\cpho$ so we invoke \Cref{each-query-in-compress-oracle}. For $\ket{x,u,w,r}\ket{D}$ in $S$ the state is unchanged after $\cal{O}_{t-1}$. For $\ket{x,u,w,r}\ket{D}$ in $R$, we have $\ket{D(x)}_{\dbreg_x}=\sum_{y'\in[N]}\frac{(-1)^{\langle u,y'\rangle}}{\sqrt{N}}\ket{y'}$ after $\cal{O}_{t-1}$. Since only one uniform random value $y'$ is added, there are at most $(t+k_0)^{\ell-1}$ possible value of $y'$ such that it can form an $\ell$-tuple with sum $y$ with $\ell-1$ elements already in the database. Thus we have $\abs{\cal{O}_{t-1}\Pi_{{\sf inc}}^{y,\cal{O}_{t-1}}Q\ket{\psi_t}}\leq\sqrt{\frac{(t+k_0)^{\ell-1}}{N}}$. For $\ket{x,u,w,r}\ket{D}$ in $R$, we have $\ket{D(x)}_{\dbreg_x}=\frac{(-1)^{\langle u,D(x)\rangle}}{\sqrt{N}}\ket{\bot}+\frac{1+(-1)^{\langle u,D(x)\rangle}(N-2)}{N}\ket{D(x)}+\sum_{y'\in[N]\backslash\set{D(x)}}\frac{1-(-1)^{\langle u,y'\rangle}-(-1)^{\langle u,D(x)\rangle}}{N}\ket{y'}$ after $\cal{O}_{t-1}$. Since the only case where a new value $y'$ is added has amplitude $\frac{1-(-1)^{\langle u,y'\rangle}-(-1)^{\langle u,D(x)\rangle}}{N}$ and there are at most $(t+k_0)^{\ell-1}$ possible value of $y'$ such that it can form an $\ell$-tuple with sum $y$ with $\ell-1$ elements already in the database, we have $\abs{\cal{O}_{t-1}\Pi_{{\sf inc}}^{y,\cal{O}_{t-1}}R\ket{\psi_t}}\leq\frac{3(t+k_0)^{\ell-1}}{N}$.
        \item When $t>T$, $\cal{O}_{t-1}=\stddecomp $. For $\ket{x,u,w,r}\ket{D}$ in $S$ the state is unchanged after $\cal{O}_{t-1}$. For $\ket{x,u,w,r}\ket{D}$ in $Q$, we have $\ket{D(x)}_{\dbreg_x}=\frac{1}{\sqrt{N}}\sum_{y'\in[N]}\ket{y'}$ after $\cal{O}_{t-1}$. Since only one uniform random value $y'$ is added, there are at most $(t+k_0)^{\ell-1}$ possible value of $y'$ such that it can form an $\ell$-tuple with sum $y$ with $\ell-1$ elements already in the database. Thus we have $\abs{\cal{O}_{t-1}\Pi_{{\sf inc}}^{y,\cal{O}_{t-1}}Q\ket{\psi_t}}\leq\sqrt{\frac{(t+k_0)^{\ell-1}}{N}}$. For $\ket{x,u,w,r}\ket{D}$ in $R$, we have $\ket{D(x)}_{\dbreg_x}=\frac{N-1}{N}\ket{D(x)}+\frac{1}{\sqrt{N}}\ket{\bot}-\frac{1}{N}\sum_{y'\in[N]\backslash\set{D(x)}}\ket{y'}$ after $\cal{O}_{t-1}$. Since the only case where a new value $y'$ is added has amplitude $\frac{1}{N}$ and there are at most $(t+k_0)^{\ell-1}$ possible value of $y'$ such that it can form an $\ell$-tuple with sum $y$ with $\ell-1$ elements already in the database, we have $\abs{\cal{O}_{t-1}\Pi_{{\sf inc}}^{y,\cal{O}_{t-1}}R\ket{\psi_t}}\leq\frac{(t+k_0)^{\ell-1}}{N}$.
    \end{itemize}
    Combining all cases above, we obtain the following recursion formula for $\Delta^y_{t,k}$:
    \begin{equation*}
        \Delta^y_{t,k}\leq\Delta^y_{t-1,k}+\Delta^y_{t-1,k-1}\cdot \brackets{4\sqrt{\frac{(t+k_0)^{\ell-1}}{N}}}.
    \end{equation*}
    Thus 
    \begin{equation*}
        \Delta^y_{t,k}\leq \binom{t}{k-k_\ell}\brackets{\sqrt{4\frac{(t+k_0)^{\ell-1}}{N}}}^{k-k_\ell}.
    \end{equation*}
    The probability that at some point during the algorithm there are $K$ distinct $\ell$-tuples of sum $y\in[N]$ in the database is at most
    \begin{align*}
        &\abs{\brackets{\sum_{i=0}^{2T-1}\prod_{i+1}^{j=2T-1}\brackets{\cal{O}_jU_j}\Pi_{\geq K}^{\ell,*}\cal{O}_iU_i\prod_{0}^{j=i-1}\brackets{\Pi_{<K}^{\ell,*}\cal{O}_jU_j}}\ket{\psi_{{\sf st}}}}^2\\
        \leq&\abs{\brackets{\sum_{y\in[N]}\sum_{i=0}^{2T-1}\prod_{i+1}^{j=2T-1}\brackets{\cal{O}_jU_j}\Pi_{\geq K}^{\ell,y}\cal{O}_iU_i\prod_{0}^{j=i-1}\brackets{\Pi_{<K}^{\ell,y}\cal{O}_jU_j}}\ket{\psi_{{\sf st}}}}^2\\
        \leq&N\sum_{y\in[N]}\abs{\brackets{\sum_{i=0}^{2T-1}\prod_{i+1}^{j=2T-1}\brackets{\cal{O}_jU_j}\Pi_{\geq K}^{\ell,y}\cal{O}_iU_i\prod_{0}^{j=i-1}\brackets{\Pi_{<K}^{\ell,y}\cal{O}_jU_j}}\ket{\psi_{{\sf st}}}}^2\\
        \leq&N\sum_{y\in[N]}\brackets{\Gamma^{y}_{2T,K}}^2\\
        \leq&N\abs{\brackets{\sum_{i=0}^{2T-1}\prod_{i+1}^{j=2T-1}\brackets{\cal{O}_jU_j}\Pi_{\geq K_{{\sf thd}}}^{\ell-1,*}\cal{O}_iU_i\prod_{0}^{j=i-1}\brackets{\Pi_{<K_{{\sf thd}}}^{\ell-1,*}\cal{O}_jU_j}}\ket{\psi_{{\sf st}}}}^2\\
        &+N\abs{\brackets{\sum_{\substack{c_0,c_1,\cdots,c_t\\ \sum_{i=1}^{T}\mathbb{I}[c_i>c_{i-1}]\geq \frac{K-k_\ell}{K_{{\sf thd}}}}}\prod_{0}^{j=2T-1}\brackets{\Pi_{=c_{j+1}}^{\ell,y}\cal{O}_jU_j}\Pi_{=c_0}^{\ell,y}}\ket{\psi_{{\sf st}}}}^2\\
        \leq&NO\brackets{\frac{T^2(T+k_0)^{\ell-2}}{\brackets{K'_{{\sf sol}}}^2N}}^{K'_{{\sf sol}}}+N\brackets{\Delta^y_{2T,\frac{K-k_\ell}{K_{{\sf thd}}}}}^2\\
        \leq&NO\brackets{\frac{T^2(T+k_0)^{\ell-2}}{\brackets{K'_{{\sf sol}}}^2N}}^{K'_{{\sf sol}}}+N\binom{2T}{\frac{K-k_\ell}{K_{{\sf thd}}}}^2\brackets{\frac{16(2T+k_0)^{\ell-1}}{N}}^{\frac{K-k_\ell}{K_{{\sf thd}}}}\\
        \leq&NO\brackets{\frac{T^2(T+k_0)^{\ell-2}}{\brackets{K'_{{\sf sol}}}^2N}}^{K'_{{\sf sol}}}+NO\brackets{\frac{T^2(T+k_0)^{\ell-1}}{\brackets{\frac{K-k_\ell}{K_{{\sf thd}}}}^2N}}^{\frac{K-k_\ell}{K_{{\sf thd}}}}.
    \end{align*}
    We set $K_{{\sf thd}}$ in a way that $K_{{\sf sol}}:=\frac{K-k_\ell}{K_{{\sf thd}}}$ is the only real non-negative root of the function
    \begin{equation*}
        f_{K,k_1,k_2,\cdots,k_\ell}(x)=x^\ell+\sum_{i=0}^{\ell-1}k_{\ell-i}x^i-K.
    \end{equation*}
    Since we know that $K'_{{\sf sol}}$ is the only real non-negative root of the function
    \begin{equation*}
        f_{K_{{\sf thd}},k_1,k_2,\cdots,k_{\ell-1}}(x)=x^{\ell-1}+\sum_{i=0}^{\ell-2}k_{\ell-i-1}x^i-K_{{\sf thd}}.
    \end{equation*}
    We have that $K'_{{\sf sol}}=\frac{K-k_\ell}{K_{{\sf thd}}}$. In other words, $\frac{K-k_\ell}{K_{{\sf thd}}}$ is also the only real non-negative root of the function $f_{K_{{\sf thd}},k_1,k_2,\cdots,k_{\ell-1}}(x)$. Proved by
    \begin{align*}
        &f_{K_{{\sf thd}},k_1,k_2,\cdots,k_{\ell-1}}\brackets{\frac{K-k_\ell}{K_{{\sf thd}}}}\\
        =&\brackets{\frac{K-k_\ell}{K_{{\sf thd}}}}^{\ell-1}+\sum_{i=0}^{\ell-2}k_{\ell-i-1}\brackets{\frac{K-k_\ell}{K_{{\sf thd}}}}^i-K_{{\sf thd}}\\
        =&\brackets{\frac{K_{{\sf thd}}}{K-k_\ell}}\brackets{\brackets{\frac{K-k_\ell}{K_{{\sf thd}}}}^{\ell}+\sum_{i=1}^{\ell-1}k_{\ell-i}\brackets{\frac{K-k_\ell}{K_{{\sf thd}}}}^i-K+k_\ell}\\
        =&\brackets{\frac{K_{{\sf thd}}}{K-k_\ell}}\brackets{\brackets{\frac{K-k_\ell}{K_{{\sf thd}}}}^{\ell}+\sum_{i=0}^{\ell-1}k_{\ell-i}\brackets{\frac{K-k_\ell}{K_{{\sf thd}}}}^i-K}\\
        =&\brackets{\frac{K_{{\sf thd}}}{K-k_\ell}}f_{K,k_1,k_2,\cdots,k_\ell}\brackets{\frac{K-k_\ell}{K_{{\sf thd}}}}=0.
    \end{align*}
    Finishing the computation, the probability that at some point during the algorithm there are $K$ distinct $\ell$-tuples with the same sum in the database is at most
    \begin{align*}
        &NO\brackets{\frac{T^2(T+k_0)^{\ell-2}}{\brackets{K'_{{\sf sol}}}^2N}}^{K'_{{\sf sol}}}+NO\brackets{\frac{T^2(T+k_0)^{\ell-1}}{\brackets{\frac{K-k_\ell}{K_{{\sf thd}}}}^2N}}^{\frac{K-k_\ell}{K_{{\sf thd}}}}\\
        \leq& NO\brackets{\frac{T^2(T+k_0)^{\ell-2}}{K_{{\sf sol}}^2N}}^{K_{{\sf sol}}}+NO\brackets{\frac{T^2(T+k_0)^{\ell-1}}{K_{{\sf sol}}^2N}}^{K_{{\sf sol}}}\\
        \leq& O\brackets{\frac{T^2(T+k_0)^{\ell-1}}{K_{{\sf sol}}^2N}}^{K_{{\sf sol}}}\\
    \end{align*}
    The $N$ factor can be absorbed into the big-$O$ notation because $K_{{\sf sol}}=\Omega(\log N)$. Note that since $\ell>0$ is a constant, doing induction on it preserves the $O(\cdot)$ notation.
\end{proof}
\section{\texorpdfstring{The time bound for finding $K$ $\ell$-tuples with the same sum with advice}{The time bound for finding K l-tuples with the same sum with advice}}
\label{sec:time-bound-for-finding-K-l-tuples-with-the-same-sum-with-advice}

\newcommand{\seedreg}{\mathbf{S}}
\newcommand{\memreg}{\mathbf{B}}
\newcommand{\store}{{\sf Store}}
\renewcommand{\gen}{{\sf Gen}}
\newcommand{\ancillareg}{\mathbf{A}}

\qipeng{A global todo list here: optimize some subscripts/superscripts}

Now we consider an algorithm starting with $S$-qubit advice state 
\begin{equation*}
    \ket{\psi_{\sf st}}=\frac{1}{\sqrt{N^M2^{M^\ell}}}\sum_{H,\hat{r}}\ket{\psi_{H,\hat{r}}}_{\qryreg\ansreg\auxreg}\ket{\hat{r}}_{\outreg}\ket{H}_{\dbreg},
\end{equation*} 
where $\ket{\psi_{H,\hat{r}}}_{\qryreg\ansreg\auxreg}$ is of size $S$. $\hat{r}$ is in the Hadamard basis representation and $H$ is a standard oracle representation. $\outreg$ contains $M^\ell$ qubits. Let $\hat{r}_{(x_1,x_2,\cdots,x_\ell)}$/$r_{(x_1,x_2,\cdots,x_\ell)}$ be the value on the qubit in $\outreg$ corresponding to the tuple $(x_1,x_2,\cdots,x_\ell)$. The structure of $\ket{\psi_{\sf st}}$ means that $\outreg$, under the Hadamard basis is a classical control register. We will use this fact in the next section. We also introduce another random challenge register $\seedreg$ initialized to $\ket{0}_{\seedreg}$ and a memory qubit $\memreg$ initialized to $\ket{0}_{\memreg}$.\\
For a sequence of target values $\set{y_H}_{H\in[N]^M}$ (we will just use $\set{y_H}$ later) and an algorithm $\cal{A}$:
\begin{itemize}
    \item $\cal{A}$ is an isometry that receives $S$ qubit advice $\ket{\psi_{\sf st}}$. Its purification $\cal{A}'$ is  a unitary (with access to the oracle and the output register) on $\qryreg\ansreg\auxreg\ancillareg\outreg\dbreg$ where $\ancillareg$ is the purification register of arbitrary size initialized to $\ket{0}_{\ancillareg}$.
    \item $\cal{A}'$ only queries the oracle register by $\pho$ at most $T$ times.
    \item $\cal{A}'$ only interacts with the output register by $\output$ at most $T$ times. It can interact with both parts of the output register.
\end{itemize}
\begin{remark}
    In general, the advice may not be a pure state. But later on we will see that all experiments are linear which means we can take an average when the state is mixed. This will be used in later sections.
\end{remark}
We may also consider other inputs for $\cal{A}'$ other than $\ket{\psi_{\sf st}}$ later in this section. Define
\begin{equation*}
    \rho_{\sf mixed}=\brackets{\frac{I}{2^S}}_{\qryreg\ansreg\auxreg},
\end{equation*}
\begin{align*}
    \rho_{\sf standard}=&\rho_{\sf mixed}\otimes\ket{0}_{\ancillareg}\bra{0}\otimes\ket{0}_{\outreg}\bra{0}\otimes\frac{1}{N^M}\brackets{\sum_{H}\ket{H}_{\dbreg}}\brackets{\sum_{H}\bra{H}_{\dbreg}}\otimes\ket{0}_{\seedreg}\bra{0}\otimes\ket{0}_{\memreg}\bra{0}\\
    =&\rho_{\sf mixed}\otimes\ket{0}_{\ancillareg}\bra{0}\otimes\frac{1}{N^M2^{M^\ell}}\brackets{\sum_{H,\hat{r}}\ket{\hat{r}}_{\outreg}\ket{H}_{\dbreg}}\brackets{\sum_{H,\hat{r}}\bra{\hat{r}}_{\outreg}\bra{H}_{\dbreg}}\otimes\ket{0}_{\seedreg}\bra{0}\otimes\ket{0}_{\memreg}\bra{0}
\end{align*}
and
\begin{equation*}
    \rho_{\sf compressed}=\rho_{\sf mixed}\otimes\ket{0}_{\ancillareg}\bra{0}\otimes\ket{0}_{\outreg}\bra{0}\otimes\ket{\bot,\bot,\cdots,\bot}_{\dbreg}\bra{\bot,\bot,\cdots,\bot}\otimes\ket{0}_{\seedreg}\bra{0}\otimes\ket{0}_{\memreg}\bra{0}
\end{equation*}
which are inputs that the advice is removed.

\hypertarget{step1}{\textbf{Step 1: }} We are interested in the probability $P^{(1)}_{\sf FlipTest}:=\abs{\Pi_{{\sf FlipTest}}^{\cal{A}',\set{y_H}}\brackets{\ket{\psi_{\sf st}}\ket{0}_{\ancillareg}\ket{0}_{\seedreg}\ket{0}_{\memreg}}}^2$ where $\Pi_{{\sf FlipTest}}^{\cal{A}',\set{y_H}}$ is shown in \Cref{proj:single-game}.

\protocol
{$\Pi_{{\sf FlipTest}}^{\cal{A}',\set{y_H}}$:}
{Challenge an algorithm for finding a random $\ell$-tuple.}
{proj:single-game}
{
\begin{description}
\setlength{\parskip}{0.3mm} 
\setlength{\itemsep}{0.3mm} 
\item[Domain:] $\qryreg\ansreg\auxreg\outreg\dbreg\seedreg\memreg$.
\end{description}
\begin{enumerate}
    \item Apply the unitary $\gen$ that maps $\ket{0}_{\seedreg}$ to $\sum_{(x_1,x_2,\cdots,x_\ell)\in [M]^\ell}\ket{(x_1,x_2,\cdots,x_\ell)}$.
    \item Apply unitary $\store$:
    \begin{equation*}
        \ket{r}_{\outreg}\ket{(x_1,x_2,\cdots,x_\ell)}_{\seedreg}\ket{b}_{\memreg}\rightarrow\ket{r}_{\outreg}\ket{(x_1,x_2,\cdots,x_\ell)}_{\seedreg}\ket{b\oplus r_{(x_1,x_2,\cdots,x_\ell)}}_{\memreg}.
    \end{equation*}
    Which is a unitary that copies the output bit (in computational basis) on tuple $(x_1,x_2,\cdots,x_\ell)$ into the memory register.
    \item Run $\cal{A}'$.
    \item Apply unitary $\store^\dagger$.
    \item Project the state onto the state where
    \begin{enumerate}
        \item The value on $\memreg$ is $1$ under computational basis.
        \item The $\ell$-tuple on $\seedreg$ is valid and has sum $y_H$ where $H$ is the oracle on $\dbreg$.
    \end{enumerate}
    \item Apply unitary $\store$ again.
    \item Run $\cal{A}'^\dagger$ (run $\cal{A}'$ in the reverse order and change every unitary into its conjugate)
    \item Apply $\gen^\dagger\store^\dagger$.
\end{enumerate}
}


The idea of this experiment is that we want to know the probability that the algorithm $\cal{A}'$ recognizes a random $\ell$-tuple. A random $\ell$-tuple is stored in the register $\seedreg$ and the projector projects to the subspace where the bit corresponding to that $\ell$-tuple is flipped after running $\cal{A}'$ and that $\ell$-tuples has sum $y_H$. Ideally, we want to prove that the probability that any algorithm $\cal{A}'$ on any starting state $\ket{\psi_{\sf st}}$ recognizing a random $\ell$-tuple of sum $y_H$ is extremely small when $T$ is small. For example, if an algorithm $\cal{A}'$ can only output $K$ valid $\ell$-tuples of sum $y_H$, its success probability in ${\sf FlipTest}$ is just $\frac{K}{M^\ell}$. From the last section, we learn that for any uniform algorithm with $T$ queries, the chance of finding too many $\ell$-tuples of the same sum is very small. Even if we have an $S$-qubit advice, it should only increase the success probability by $2^S$ multiplicatively, which is not enough. The above reasoning seems to work out; however, in fact it takes many steps to implement it. Here is a road map that provides an intuition of our proof with links on these steps and connections between these steps \Cref{fig:roadmap} .

\begin{figure}[p]
\centering 
\scalebox{1
}{
\begin{tikzpicture}[
    primitive/.style={rectangle, draw, minimum width=3cm, minimum height=1cm, 
                      align=center, fill=blue!10},
    strong/.style={primitive, fill=red!10},
    assumptions/.style={rectangle, draw,  minimum width=3cm, minimum height=1cm, 
                      align=center, fill=gray!10},
    implies/.style={-{Stealth[length=4mm]},double, double distance=2pt},
    equivalent/.style={<->{Stealth[length=3mm]}, thick, double, double distance=2pt},
    dashed_implies/.style={-{Stealth[length=3mm]}, thick, dashed},
    node distance=2.5cm and 3cm
]

\node[primitive,draw=none] (OneFoldedNonUniformGame) {\hyperlink{step1}{\textbf{Step 1: }}The probability that $\cal{A}'$ on any \\starting state $\ket{\psi_{\sf st}}$ recognizing a \\random $\ell$-tuple  of sum $y_H$ which is \\$P^{(1)}_{\sf FlipTest}:=\abs{\Pi_{{\sf FlipTest}}^{\cal{A}',\set{y_H}}\brackets{\ket{\psi_{\sf st}}\ket{0}_{\ancillareg}\ket{0}_{\seedreg}\ket{0}_{\memreg}}}^2$};
\node[primitive, draw=none,right=1.5cm of OneFoldedNonUniformGame] (KFoldedNonUniformGame) {\hyperlink{step2}{\textbf{Step 2: }}The probability that $\cal{A}'$ on any starting \\state $\ket{\psi_{\sf st}}$ wins  \cite{marriott2005quantum}-form alternating \\game. The game is defined as an alternation\\ between $\Pi_{{\sf FlipTest}}^{\cal{A}',\set{y_H}}$ and $\Pi_{\sf reset}$ which reset\\ randomness that is needed in the challenge.};
\node[primitive, draw=none,below=1.2cm of KFoldedNonUniformGame] (KFoldedUniformGame) {\hyperlink{step3}{\textbf{Step 3: }}The probability that $\cal{A}'$ on maximally \\mixed state $\rho_{\sf standard}$ wins  \cite{marriott2005quantum}-form \\alternating experiment. The advice is \\now removed and a factor of $2^S$ is \\multiplied but absorbed by the exponent.};
\node[primitive, draw=none, left=1.5cm of KFoldedUniformGame] (KFoldedUniformGameOnlyLastRound) {\hyperlink{step4}{\textbf{Step 4: }}The probability that condition on $\cal{A}'$ \\on maximally mixed state $\rho_{\sf standard}$ \\winning all but the last round of the \\ \cite{marriott2005quantum}-form alternating experiment, it also \\wins the last round. Which is the probability \\that condition on ${\sf Expt}_{\sf FlipTest}^{\cal{A}',\set{y_H}}(S)$ outputting \\${\sf PredicatePassed}$, also outputs ${\sf Accepted}$.};
\node[primitive, draw=none, below=0.75cm of KFoldedUniformGameOnlyLastRound] (KFoldedUniformGameLastRoundProb) {\hyperlink{step5}{\textbf{Step 5: }}The probability that condition \\on $\cal{A}'$ winning all but the last round \\of the alternating experiment, it also \\succeeds in finding $K$ $\ell$-tuples\\ with sum $y_H$ in the last round.};
\node[primitive, draw=none, right=1.5cm of KFoldedUniformGameLastRoundProb] (KFoldedUniformGameOnlyLastRoundCompressed) {\hyperlink{step6}{\textbf{Step 6: }}The probability that condition on $\cal{A}'$ \\winning all but the last round of the \\alternating experiment, it also finds \\$K$ $\ell$-tuples with sum $y_H$ in the last round. \\(in the compressed oracle model)\hyperlink{step7}{\textbf{Step 7: }}It \\is at most the sum of two probabilities below.};
\node[primitive, draw=none,below = 0.75cm of KFoldedUniformGameOnlyLastRoundCompressed,xshift=-9cm] (BadStructure) {\hyperlink{step8a}{\textbf{Step 8(a): }}The probability that condition on \\$\cal{A}'$ winning all but the last round of the \\alternating experiment, the state before the \\last round has a bad structure where \\there are too many $i$-tuples with the \\same sum in the database for some $i$.\\ This is bounded by \Cref{lem:bound_with_bad_structure}.};
\node[primitive, draw=none,below = 1cm of KFoldedUniformGameOnlyLastRoundCompressed] (GoodStructureAndSucceed) {\hyperlink{step8b}{\textbf{Step 8(b): }}The probability that condition on \\$\cal{A}'$ winning all but the last round of the \\alternating experiment and a good structure \\on the resulting state, it also succeeds\\ in outputting $K$-tuples with sum $y_H$.\\ This is bounded by \Cref{lem:bound_with_good_structure}.};

\draw[implies] (KFoldedNonUniformGame) -- (OneFoldedNonUniformGame) node[midway,above,font=\small]{\Cref{lem:single_game_bound_1}};
\draw[implies] (KFoldedUniformGame) -- (KFoldedNonUniformGame) node[midway, right, font=\small] {\Cref{lem:single_game_bound_2}};
\draw[implies] (KFoldedUniformGameOnlyLastRound) -- (KFoldedUniformGame) node[midway, above, font=\small] {\Cref{lem:alternative_game_monotone}};
\draw[implies] (KFoldedUniformGameLastRoundProb) -- (KFoldedUniformGameOnlyLastRound) node[midway, right, font=\small] {\Cref{lem:switch_between_expectation_and_tail_bound}};
\draw[implies] (KFoldedUniformGameOnlyLastRoundCompressed) -- (KFoldedUniformGameLastRoundProb) node[midway, above, font=\small] {\Cref{lem:compressed-oracle-equivalent-to-normal-oracle}};
\draw[implies] (BadStructure) -- (KFoldedUniformGameOnlyLastRoundCompressed) node[midway, right, font=\small] {\quad\quad\Cref{lem:split_good_bad_structure}};
\draw[implies] (GoodStructureAndSucceed) -- (KFoldedUniformGameOnlyLastRoundCompressed) node[midway, right, font=\small] { };

\end{tikzpicture}
}
\caption{An overview of this section.}
\label{fig:roadmap}
\end{figure}
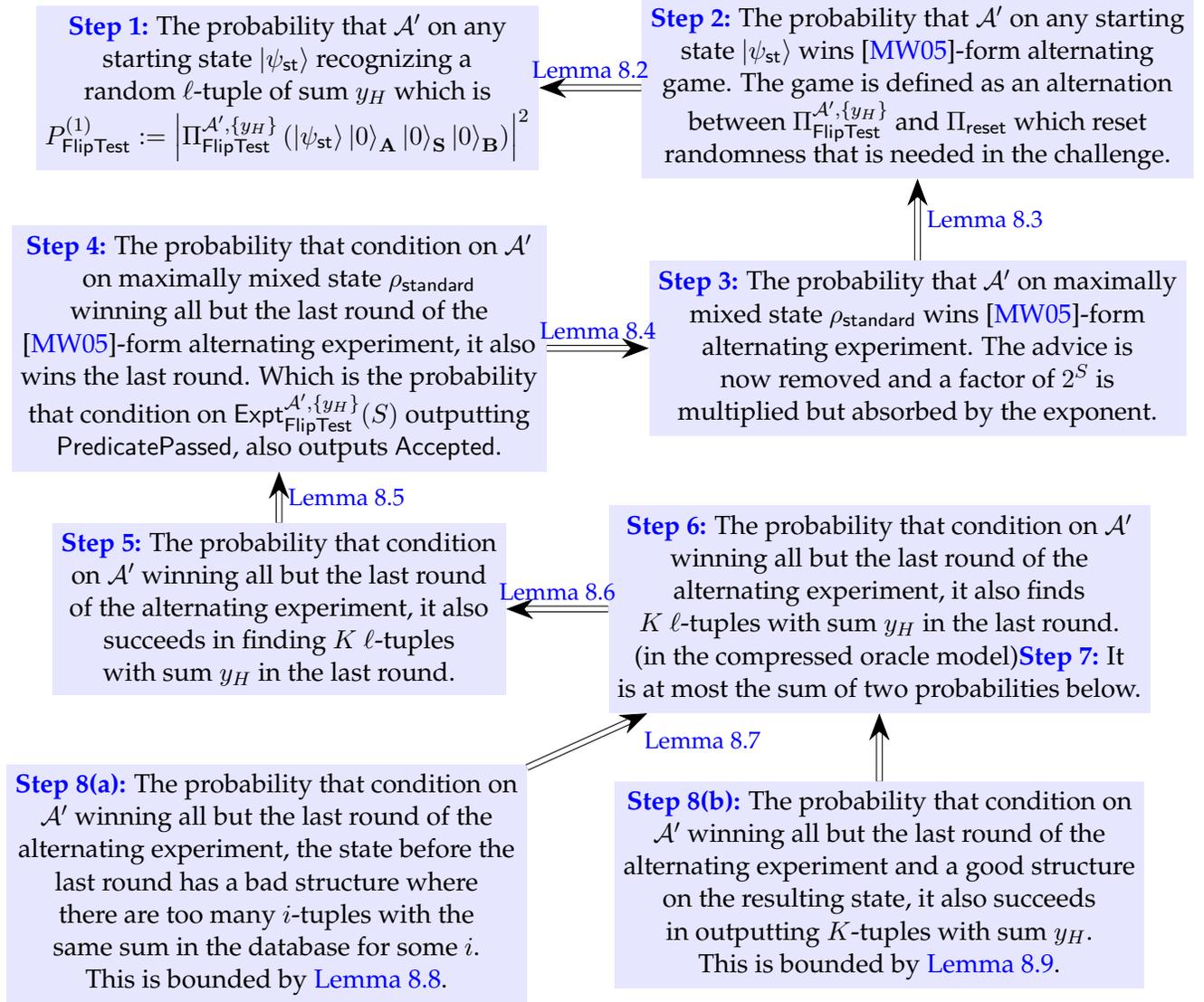 



\hypertarget{step2}{\textbf{Step 2: }} We first bound the probability $P^{(1)}_{\sf FlipTest}$ by the success probability of an alternating experiment. We introduce a \cite{marriott2005quantum}-form alternating experiment and prove the following result. \Cref{lem:single_game_bound_1}, \Cref{lem:single_game_bound_2} and \Cref{lem:alternative_game_monotone} are lemmas from \cite{liu2022nonuniformityquantumadvicequantum}, we include the proof in \Cref{sec:appendix_lemmas_alternating_game} for completeness.

\begin{lemma}[\cite{liu2022nonuniformityquantumadvicequantum}, Lemma 6.5]\label{lem:single_game_bound_1}
    Let $P^{(k)}_{\sf FlipTest}:=\abs{\brackets{\Pi_{{\sf FlipTest}}^{\cal{A}',\set{y_H}}\Pi_{\sf reset}}^k\brackets{\ket{\psi_{\sf st}}\ket{0}_{\ancillareg}\ket{0}_{\seedreg}\ket{0}_{\memreg}}}^2$ where $\Pi_{\sf reset}:=\ket{0}_{\seedreg}\bra{0}\otimes\ket{0}_{\memreg}\bra{0}$. We have 
    \begin{equation*}
        P^{(1)}_{\sf FlipTest}\leq\brackets{P^{(k)}_{\sf FlipTest}}^{\frac{1}{2k-1}}.
    \end{equation*}
\end{lemma}
\hypertarget{step3}{\textbf{Step 3: }} Now we try to substitute the advice state with an $S$-qubit maximally mixed state and absorb the $2^S$ factor by the exponent.
\begin{lemma}\label{lem:single_game_bound_2}
    Let $P^{{\sf uniform},(k)}_{\sf FlipTest}=\trace{\brackets{\Pi_{{\sf FlipTest}}^{\cal{A}',\set{y_H}}\Pi_{\sf reset}}^k\rho_{\sf standard}\brackets{\Pi_{\sf reset}\Pi_{{\sf FlipTest}}^{\cal{A}',\set{y_H}}}^k}$. We have
    \begin{equation*}
        P^{(1)}_{\sf FlipTest}\leq2\brackets{P^{{\sf uniform},(S)}_{\sf FlipTest}}^{\frac{1}{2S-1}}.
    \end{equation*}
\end{lemma}

\hypertarget{step4}{\textbf{Step 4: }} Now we try to bound $P^{{\sf uniform},(S)}_{\sf FlipTest}$ and link it to a new experiment. Before that, we first show a ``monotone'' theorem on each alternation.
\begin{lemma}[\cite{liu2022nonuniformityquantumadvicequantum}, Corollary 6.10]\label{lem:alternative_game_monotone}
    For $k\geq 2$ we have
    \begin{equation*}
        \frac{P^{{\sf uniform},(k)}_{\sf FlipTest}}{P^{{\sf uniform},(k-1)}_{\sf FlipTest}}\leq \frac{P^{{\sf uniform},(k+1)}_{\sf FlipTest}}{P^{{\sf uniform},(k)}_{\sf FlipTest}}.
    \end{equation*}
\end{lemma}
From this theorem, we can first bound $P^{{\sf uniform},(S)}_{\sf FlipTest}$ in the following way:
\begin{equation*}
    P^{{\sf uniform},(S)}_{\sf FlipTest}=\prod_{k=1}^{S}\brackets{\frac{P^{{\sf uniform},(k)}_{\sf FlipTest}}{P^{{\sf uniform},(k-1)}_{\sf FlipTest}}}\leq\brackets{\frac{P^{{\sf uniform},(S)}_{\sf FlipTest}}{P^{{\sf uniform},(S-1)}_{\sf FlipTest}}}^S.
\end{equation*}
Define $\Pi_{\sf alter}^{(k)}=\brackets{\Pi_{\sf reset}\Pi_{{\sf FlipTest}}^{\cal{A}',\set{y_H}}}^{k}$. Define the following experiment \Cref{fig:find_random_tuple_with_predicate}. We can notice that 
\begin{align*}
    \frac{P^{{\sf uniform},(S)}_{\sf FlipTest}}{P^{{\sf uniform},(S-1)}_{\sf FlipTest}}=&\frac{\trace{\brackets{\Pi_{{\sf FlipTest}}^{\cal{A}',\set{y_H}}\Pi_{\sf reset}}^S\rho_{\sf standard}\brackets{\Pi_{\sf reset}\Pi_{{\sf FlipTest}}^{\cal{A}',\set{y_H}}}^S}}{\trace{\brackets{\Pi_{{\sf FlipTest}}^{\cal{A}',\set{y_H}}\Pi_{\sf reset}}^{S-1}\rho_{\sf standard}\brackets{\Pi_{\sf reset}\Pi_{{\sf FlipTest}}^{\cal{A}',\set{y_H}}}^{S-1}}}\\
    \leq&\frac{\trace{\Pi_{{\sf FlipTest}}^{\cal{A}',\set{y_H}}\Pi_{\sf alter}^{(S-1)}\rho_{\sf standard}\Pi_{\sf alter}^{(S-1)}\Pi_{{\sf FlipTest}}^{\cal{A}',\set{y_H}}}}{\trace{\Pi_{\sf alter}^{(S-1)}\rho_{\sf standard}\Pi_{\sf alter}^{(S-1)}}}\\
    =&\frac{\Pr\left[\set{\sf PredicatePassed, Accepted}\subseteq{\sf Expt}_{\sf FlipTest}^{\cal{A}',\set{y_H}}(S)\right]}{\Pr\left[{\sf PredicatePassed}\in{\sf Expt}_{\sf FlipTest}^{\cal{A}',\set{y_H}}(S)\right]}.
\end{align*}
\protocol
{${\sf Expt}_{\sf FlipTest}^{\cal{A}',\set{y_H}}$:}
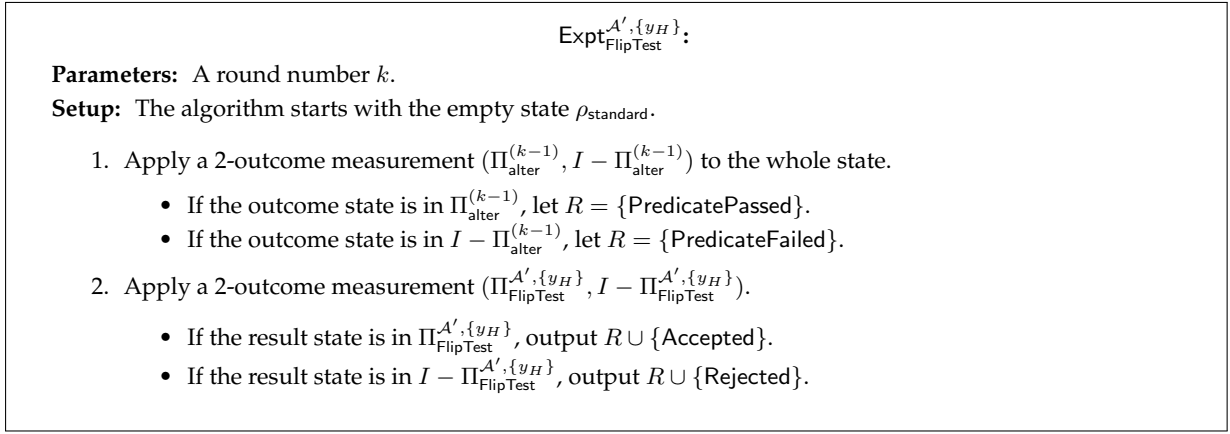
{Challenge an algorithm for finding a random $\ell$-tuple with predicate.}
{fig:find_random_tuple_with_predicate}
{
\begin{description}
\setlength{\parskip}{0.3mm} 
\setlength{\itemsep}{0.3mm} 
\item[Parameters:] A round number $k$.
\item[Setup:] The algorithm starts with the empty state $\rho_{\sf standard}$.
\end{description}
\begin{enumerate}
    \item Apply a 2-outcome measurement $(\Pi_{\sf alter}^{(k-1)},I-\Pi_{\sf alter}^{(k-1)})$ to the whole state. 
    \begin{itemize}
        \item If the outcome state is in $\Pi_{\sf alter}^{(k-1)}$, let $R=\set{{\sf PredicatePassed}}$.
        \item If the outcome state is in $I-\Pi_{\sf alter}^{(k-1)}$, let $R=\set{{\sf PredicateFailed}}$.
    \end{itemize}
    \item Apply a 2-outcome measurement $(\Pi_{{\sf FlipTest}}^{\cal{A}',\set{y_H}},I-\Pi_{{\sf FlipTest}}^{\cal{A}',\set{y_H}})$.
    \begin{itemize}
        \item If the result state is in $\Pi_{{\sf FlipTest}}^{\cal{A}',\set{y_H}}$, output $R\cup\set{\sf Accepted}$.
        \item If the result state is in $I-\Pi_{{\sf FlipTest}}^{\cal{A}',\set{y_H}}$, output $R\cup\set{\sf Rejected}$.
    \end{itemize}
\end{enumerate}
}
\hypertarget{step5}{\textbf{Step 5: }}Now we turn to another new experiment \Cref{fig:standard_oracle_success_prob_estimation}. Instead of examining whether a single $\ell$-tuple is recognized, this experiment now cares about the probability of $\cal{A}'$ outputting at least $K$ $\ell$-tuples of the same sum. Here we say that $\cal{A}'$ outputs an $\ell$-tuple $(x_1,x_2,\cdots,x_\ell)$ if $r_{(x_1,x_2,\cdots,x_\ell)}$ is non-zero.
\protocol
{${\sf Expt}_{{\sf Standard}}^{\cal{A}'}$:}
{Verifying whether an algorithm succeeds in finding $K$ distinct $\ell$-tuples the same sum with a predicate.}
{fig:standard_oracle_success_prob_estimation}
{
\begin{description}
\setlength{\parskip}{0.3mm} 
\setlength{\itemsep}{0.3mm} 
\item[Parameters:] A projector $\Pi_{\sf predicate}$ and a threshold $K$.
\item[Setup:] The algorithm starts with the state $\rho_{\sf standard}$.
\end{description}
\begin{enumerate}
    \item Apply a 2-outcome measurement $(\Pi_{\sf predicate},I-\Pi_{\sf predicate})$ to the whole state. 
    \begin{itemize}
        \item If the outcome state is in $\Pi_{\sf predicate}$, let $R=\set{{\sf PredicatePassed}}$.
        \item If the outcome state is in $I-\Pi_{\sf predicate}$, let  $R=\set{{\sf PredicateFailed}}$.
    \end{itemize}
    \item For every $y\in[N]$, count how many distinct $\ell$-tuples outputted by the prover have sum $y$ for every $y$. If there exists $y\in[N]$ such that this number is at least $K$ then let $R\leftarrow R\cup\set{{\sf Accept0}}$. Otherwise, let $R\leftarrow R\cup\set{{\sf Rejected\_0}}$.
    \item Run the second stage of the algorithm $\cal{A}'$. All oracle queries use $\pho$, interacting with the oracle register $H$.
    \item For every $y\in[N]$, count how many distinct $\ell$-tuples outputted by the prover have sum $y$ for every $y$. If there exists $y\in[N]$ such that this number is at least $K$ then output $R\cup\set{{\sf Accept1}}$. Otherwise, output $R\cup\set{{\sf Rejected\_1}}$.
\end{enumerate}
}
Intuitively, ${\sf PredicatePassed}$/${\sf PredicateFailed}$ in the output of ${\sf Expt}_{{\sf Standard}}^{\cal{A}'}$ means that the algorithm passes/fails in the predicate test and ${\sf Accepted}$/${\sf Rejected}$ in the output of ${\sf Expt}_{{\sf Standard}}^{\cal{A}'}$ means the algorithm is/isn't able to find $K$ $\ell$-tuples of the same sum. 
\begin{lemma}\label{lem:switch_between_expectation_and_tail_bound}
    For any $k\geq 1$ and $K$ we have
    \begin{align*}
        &\Pr\left[{\sf PredicatePassed}\in{\sf Expt}_{\sf FlipTest}^{\cal{A}',\set{y_H}}(k)\right]\\
        =&\Pr\left[{\sf PredicatePassed}\in{\sf Expt}_{{\sf Standard}}^{\cal{A}'}\brackets{\Pi_{\sf alter}^{(k-1)},K}\right]
    \end{align*}
    and
    \begin{align*}
        &\Pr\left[\set{\sf PredicatePassed, Accepted}\subseteq{\sf Expt}_{\sf FlipTest}^{\cal{A}',\set{y_H}}(k)\right]\\
        \leq&\frac{8K}{M^\ell}\Pr\left[{\sf PredicatePassed}\in{\sf Expt}_{{\sf Standard}}^{\cal{A}'}\brackets{\Pi_{\sf alter}^{(k-1)},K}\right]\\
        &+4\Pr\left[\set{{\sf PredicatePassed,Accept0}}\subseteq{\sf Expt}_{{\sf Standard}}^{\cal{A}'}\brackets{\Pi_{\sf alter}^{(k-1)},K}\right]\\
        &+4\Pr\left[\set{{\sf PredicatePassed,Accept1}}\subseteq{\sf Expt}_{{\sf Standard}}^{\cal{A}'}\brackets{\Pi_{\sf alter}^{(k-1)},K}\right].
    \end{align*}
\end{lemma}
\begin{proof}
    Let $\rho_{\sf pass}=\Pi_{\sf alter}^{(k-1)}\rho_{\sf standard}\Pi_{\sf alter}^{(k-1)}$ (not normalized). This is the state after step 1 in ${\sf Expt}_{\sf FlipTest}^{\cal{A}',\set{y_H}}(k)$ and it is also the state after step 1 in ${\sf Expt}_{{\sf Standard}}^{\cal{A}'}$. Thus the two probabilities in the first equation are exactly the same:
    \begin{align*}
        &\Pr\left[{\sf PredicatePassed}\in{\sf Expt}_{\sf FlipTest}^{\cal{A}',\set{y_H}}(k)\right]\\
        =&\trace{\rho_{\sf pass}}\\
        =&\Pr\left[{\sf PredicatePassed}\in{\sf Expt}_{{\sf Standard}}^{\cal{A}'}\brackets{\Pi_{\sf alter}^{(k-1)},K}\right].
    \end{align*}
    Now we prove the second inequality. Define 
    \begin{equation*}
        \Pi_b=\sum_{\substack{H,(x_1,x_2,\cdots,x_\ell)\text{ is a }\\ \text{valid $\ell$-tuple of sum }y_H}}\ket{H}_{\dbreg}\bra{H}\otimes\ket{(x_1,x_2,\cdots,x_\ell)}_{\seedreg}\bra{(x_1,x_2,\cdots,x_\ell)}\otimes\ket{b}_{\outreg_{(x_1,x_2,\cdots,x_\ell)}}\bra{b}
    \end{equation*}
    and
    \begin{equation*}
        \Pi_{\sf Accepted}=\sum_{\substack{H,r\text{ such that there are $K$}\\ \text{non-zero entries on $r$ that correspond}\\ \text{to a valid $\ell$-tuple of sum $y_H$}}}\ket{H}_{\dbreg}\bra{H}\otimes\ket{r}_{\outreg}\bra{r}.
    \end{equation*}
    We have
    \begin{align*}
        &\Pr\left[\set{\sf PredicatePassed, Accepted}\subseteq{\sf Expt}_{\sf FlipTest}^{\cal{A}',\set{y_H}}(k)\right]\\
        =&2\trace{\Pi_0\cal{A}'\Pi_1\gen\cdot\rho_{\sf pass}\cdot\gen^{\dagger}\Pi_1\cal{A}'^{\dagger}\Pi_0}+2\trace{\Pi_1\cal{A}'\Pi_0\gen\cdot\rho_{\sf pass}\cdot\gen^{\dagger}\Pi_0\cal{A}'^{\dagger}\Pi_1}\\
        \leq&2\trace{\Pi_1\gen\cdot\rho_{\sf pass}\cdot\gen^{\dagger}\Pi_1}+2\trace{\Pi_1\cal{A}'\cdot\gen\cdot\rho_{\sf pass}\cdot\gen^{\dagger}\cal{A}'^{\dagger}\Pi_1}\\
        \leq&4\trace{\Pi_1\gen(I-\Pi_{\sf Accepted})\rho_{\sf pass}(I-\Pi_{\sf Accepted})\gen^{\dagger}\Pi_1}\\
        &+4\trace{\Pi_1(I-\Pi_{\sf Accepted})\cal{A}'\cdot\gen\cdot\rho_{\sf pass}\cdot\gen^{\dagger}\cal{A}'^{\dagger}(I-\Pi_{\sf Accepted})\Pi_1}\\
        &+4\trace{\Pi_1\gen\Pi_{\sf Accepted}\cdot\rho_{\sf pass}\Pi_{\sf Accepted}\gen^{\dagger}\Pi_1}\\
        &+4\trace{\Pi_1\Pi_{\sf Accepted}\cal{A}'\cdot\gen\cdot\rho_{\sf pass}\cdot\gen^{\dagger}\cal{A}'^{\dagger}\Pi_{\sf Accepted}\Pi_1}\\
        \leq&\frac{8K}{M^\ell}\trace{\rho_{\sf pass}}+4\Pr\left[\set{{\sf PredicatePassed,Accept0}}\subseteq{\sf Expt}_{{\sf Standard}}^{\cal{A}'}\brackets{\Pi_{\sf alter}^{(k-1)},K}\right]\\
        &+4\Pr\left[\set{{\sf PredicatePassed,Accept1}}\subseteq{\sf Expt}_{{\sf Standard}}^{\cal{A}'}\brackets{\Pi_{\sf alter}^{(k-1)},K}\right]\\
        \leq&\frac{8K}{M^\ell}\Pr\left[{\sf PredicatePassed}\in{\sf Expt}_{{\sf Standard}}^{\cal{A}'}\brackets{\Pi_{\sf alter}^{(k-1)},K}\right]\\
        &+4\Pr\left[\set{{\sf PredicatePassed,Accept0}}\subseteq{\sf Expt}_{{\sf Standard}}^{\cal{A}'}\brackets{\Pi_{\sf alter}^{(k-1)},K}\right]\\
        &+4\Pr\left[\set{{\sf PredicatePassed,Accept1}}\subseteq{\sf Expt}_{{\sf Standard}}^{\cal{A}'}\brackets{\Pi_{\sf alter}^{(k-1)},K}\right].
    \end{align*}
    The reason why
    \begin{equation*}
        \trace{\Pi_1\gen(I-\Pi_{\sf Accepted})\rho_{\sf pass}(I-\Pi_{\sf Accepted})\gen^{\dagger}\Pi_1}\leq\frac{K}{M^\ell}\trace{\rho_{\sf pass}}
    \end{equation*}
    and
    \begin{equation*}
        \trace{\Pi_1(I-\Pi_{\sf Accepted})\cal{A}'\cdot\gen\cdot\rho_{\sf pass}\cdot\gen^{\dagger}\cal{A}'^{\dagger}(I-\Pi_{\sf Accepted})\Pi_1}\leq\frac{K}{M^\ell}\trace{\rho_{\sf pass}}
    \end{equation*}
    is that once the state is in the span of $(I-\Pi_{\sf Accepted})$ then the number of $1$ that corresponds to a valid $\ell$-tuple of sum $y_H$ is less than or equal to $K$. For a random challenge on $\seedreg$ the probability that the challenge hits one of the valid $\ell$-tuples of sum $y_H$ is at most $\frac{K}{M^\ell}$.
\end{proof}

\hypertarget{step6}{\textbf{Step 6: }}Now we switch the standard oracle into a compressed oracle.
\protocol
{${\sf Expt}_{{\sf Compressed}}^{\cal{A}'}$:}
{Verifying whether an algorithm succeeds in finding $K$ distinct $\ell$-tuples the same sum with a predicate.}
{fig:compressed_oracle_success_prob_estimation}
{
\begin{description}
\setlength{\parskip}{0.3mm} 
\setlength{\itemsep}{0.3mm} 
\item[Parameters:] A projector $\Pi_{\sf predicate}$ and a threshold $K$.
\item[Setup:] The algorithm starts with the state $\rho_{\sf compressed}$.
\end{description}
\begin{enumerate}
    \item Define $\Pi'_{{\sf predicate}}={\sf DecompressAll}^\dagger\Pi_{\sf predicate}{\sf DecompressAll}$ where ${\sf DecompressAll}$ is the unitary that applies $\stddecomp_x$ for every $x\in[M]$. Apply a 2-outcome measurement $(\Pi'_{{\sf predicate}},I-\Pi'_{{\sf predicate}})$ to the whole state. 
    \begin{itemize}
        \item If the outcome state is in $\Pi'_{\sf predicate}$, let $R=\set{{\sf PredicatePassed}}$.
        \item If the outcome state is in $I-\Pi'_{\sf predicate}$, let  $R=\set{{\sf PredicateFailed}}$.
    \end{itemize}
    \item For every $y\in[N]$, count how \qipeng{how?} many distinct $\ell$-tuples outputted by the prover have sum $y$ for every $y$. To do this, first run ${\sf DecompressAll}$ to convert the compressed oracle into a standard oracle, count it and run ${\sf DecompressAll}^\dagger$ to recover. If there exists $y\in[N]$ such that this number is at least $K$ then let $R\leftarrow R\cup\set{{\sf Accept0}}$. Otherwise, let $R\leftarrow R\cup\set{{\sf Rejected\_0}}$.
    \item Run the second stage of the algorithm $\cal{A}'$. All oracle queries use $\cpho$, interacting with the database $\dbreg$.
    \item For every $y\in[N]$, count how \qipeng{how? similarly in every other game} many distinct $\ell$-tuples outputted by the prover have sum $y$ for every $y$. To do this, first run ${\sf DecompressAll}$ to convert the compressed oracle into a standard oracle, count it and run ${\sf DecompressAll}^\dagger$ to recover. If there exists $y\in[N]$ such that this number is at least $K$ then output $R\cup\set{{\sf Accept1}}$. Otherwise, output $R\cup\set{{\sf Rejected\_1}}$.
\end{enumerate}
}
\begin{lemma}\label{lem:compressed-oracle-equivalent-to-normal-oracle}
    For any algorithm $\cal{A}'$ the probability that any event occurs in the experiment ${\sf Expt}_{{\sf Standard}}$ is the same as that in the experiment ${\sf Expt}_{{\sf Compressed}}$. That is, for any $K$ and any subset of results
    \begin{equation*}
        R\subseteq\set{\sf PredicatePassed,PredicateFailed,Accept0,Accept1,Rejected\_0,Rejected\_1},
    \end{equation*}
    we have
    \begin{align*}
        &\Pr\left[R\subseteq{\sf Expt}_{{\sf Standard}}^{\cal{A}'}(\Pi_{\sf predicate},K)\right]\\
        =&\Pr\left[R\subseteq{\sf Expt}_{{\sf Compressed}}^{\cal{A}'}(\Pi_{\sf predicate},K)\right].
    \end{align*}
\end{lemma}
\begin{proof}
    We prove that at any time the state on $\dbreg_x$ in ${\sf Expt}_{{\sf Compressed}}$ and in ${\sf Expt}_{{\sf Standard}}$ differs by exactly $\stddecomp_x$. At some point, the state on $\dbreg_x$ in ${\sf Expt}_{{\sf Compressed}}$ is
    \begin{align*}
        &\stddecomp_x\cdot\prod_j\cpho_{u_j}\ket{\bot}\\
        =&\stddecomp_x\cdot\prod_j\brackets{\stddecomp_x\cdot\cpho'_{u_j}\cdot\stddecomp_x}\ket{\bot}\\
        =&\prod_j\brackets{\cpho'_{u_j}}\stddecomp_x\ket{\bot}\\
        =&\prod_j\cpho'_{u_j}\brackets{\sum_{y\in[N]}\frac{1}{\sqrt{N}}\ket{y}}
    \end{align*}
    where $\cpho_{u_j}$ are all gates that act on this register.\\
    One can see that using compressed oracle queries in the middle and a $\stddecomp$ at the end is equivalent to initializing that database with a uniform random value and then using $\cpho'$ for each oracle query. If you interpret the initialization $\ket{\bot}\rightarrow\brackets{\sum_{y\in[N]}\frac{1}{\sqrt{N}}\ket{y}}$ as the initialization of $H$ you will find that these two approaches are exactly the same. Thus, these two experiments must have the same result.
\end{proof}
\hypertarget{step7}{\textbf{Step 7: }}Now we consider another experiment \Cref{fig:compressed_oracle_success_prob_estimation_with_structural_test} that is similar to \Cref{fig:compressed_oracle_success_prob_estimation} but another structure projector is added between step 1 and step 2. The intuition on this experiment is that it splits the state into two types. In the case where the state has a good structure, which means that the number of $i$-tuples in the database is small for all $i$, the success probability in later stages will be bounded due to the lack of information on the tuples. We argue that it is unlikely for the state to have a bad structure, which means that there are too many $i$-tuples of the same sum for some $i$ that provide to much information in later stages. \\
\protocol
{${\sf Expt}_{{\sf Structured}}^{\cal{A}'}$:}
{Verifying whether an algorithm succeeds in finding $K$ distinct $\ell$-tuples the same sum with a predicate and a structural test.}
{fig:compressed_oracle_success_prob_estimation_with_structural_test}
{
\begin{description}
\setlength{\parskip}{0.3mm} 
\setlength{\itemsep}{0.3mm} 
\item[Parameters:] A projector $\Pi_{\sf predicate}$, a threshold $K$ and integers $k_0,k_1,\cdots,k_\ell$.
\item[Setup:] The algorithm starts with the state $\rho_{\sf compressed}$.
\end{description}
\begin{enumerate}
    \item Define $\Pi'_{{\sf predicate}}={\sf DecompressAll}^\dagger\Pi_{\sf predicate}{\sf DecompressAll}$ where ${\sf DecompressAll}$ is the unitary that applies $\stddecomp_x$ for every $x\in[M]$. Apply a 2-outcome measurement $(\Pi'_{{\sf predicate}},I-\Pi'_{{\sf predicate}})$ to the whole state. 
    \begin{itemize}
        \item If the outcome state is in $\Pi'_{\sf predicate}$, let $R=\set{{\sf PredicatePassed}}$.
        \item If the outcome state is in $I-\Pi'_{\sf predicate}$, let $R=\set{{\sf PredicateFailed}}$.
    \end{itemize}
    \item Define $\Pi_{{\sf structure}}=\Pi^{>k_0}\prod_{i=1}^{\ell}\Pi_{>k_i}^{i,*}$. Apply a 2-outcome measurement $(\Pi_{{\sf structure}},I-\Pi_{{\sf structure}})$ to the whole state. 
    \begin{itemize}
        \item If the outcome state is in $\Pi_{{\sf structure}}$, let $R\leftarrow R\cup\set{\sf GoodStructure}$.
        \item If the outcome state is in $I-\Pi_{{\sf structure}}$, let $R\leftarrow R\cup\set{\sf BadStructure}$.
    \end{itemize}
    \item For every $y\in[N]$, count how many distinct $\ell$-tuples outputted by the prover have sum $y$ for every $y$. To do this, first run ${\sf DecompressAll}$ to convert the compressed oracle into a standard oracle, count it and run ${\sf DecompressAll}^\dagger$ to recover. If there exists $y\in[N]$ such that this number is at least $K$ then let $R\leftarrow R\cup\set{{\sf Accept0}}$. Otherwise, let $R\leftarrow R\cup\set{{\sf Rejected\_0}}$.
    \item Run the second stage of the algorithm $\cal{A}'$. All oracle queries use $\cpho$, interacting with the database $\dbreg$.
    \item For every $y\in[N]$, count how many distinct $\ell$-tuples outputted by the prover have sum $y$ for every $y$. To do this, first run ${\sf DecompressAll}$ to convert the compressed oracle into a standard oracle, count it and run ${\sf DecompressAll}^\dagger$ to recover. If there exists $y\in[N]$ such that this number is at least $K$ then output $R\cup\set{{\sf Accept1}}$. Otherwise, output $R\cup\set{{\sf Rejected\_1}}$.
\end{enumerate}
}
Intuitively, ${\sf GoodStructure}$/${\sf BadStructure}$ in the output of ${\sf Expt}_{{\sf Structured}}^{\cal{A}'}(\Pi_{\sf predicate},K,k_0,k_1,\cdots,k_\ell)$ means that there doesn't/does exist $i\in\set{0,1,\cdots,\ell}$ and $y\in[N]$ such that there are more than $k_i$ $i$-tuples of sum $y$ in the database after step 1.
\newcommand{\prfreg}{\mathbf{P}}
\begin{lemma}\label{lem:split_good_bad_structure}
    For any $k_1,k_2,\cdots,k_\ell<K$. We can bound the success probability of the algorithm in the standard experiment condition on the predicate test passed by, intuitively, the probability of acceptance condition on a good structure on the database plus the probability of having a bad structure condition on predicate test passed:
    \begin{align*}
        &\frac{\Pr\left[\set{{\sf PredicatePassed,Accept0}}\subseteq{\sf Expt}_{{\sf Standard}}^{\cal{A}'}(\Pi_{\sf alter}^{(k)},K)\right]}{\Pr\left[{\sf PredicatePassed}\in{\sf Expt}_{{\sf Standard}}^{\cal{A}'}(\Pi_{\sf alter}^{(k)},K)\right]}\\
        &+\frac{\Pr\left[\set{{\sf PredicatePassed,Accept1}}\subseteq{\sf Expt}_{{\sf Standard}}^{\cal{A}'}(\Pi_{\sf alter}^{(k)},K)\right]}{\Pr\left[{\sf PredicatePassed}\in{\sf Expt}_{{\sf Standard}}^{\cal{A}'}(\Pi_{\sf alter}^{(k)},K)\right]}\\
        \leq&\frac{2\Pr\left[\set{{\sf PredicatePassed,GoodStructure,Accept0}}\subseteq{\sf Expt}_{{\sf Structured}}^{\cal{A}'}(\Pi_{\sf alter}^{(k)},K,k_0,k_1,\cdots,k_\ell)\right]}{\Pr\left[\set{{\sf PredicatePassed,GoodStructure}}\subseteq{\sf Expt}_{{\sf Structured}}^{\cal{A}'}(\Pi_{\sf alter}^{(k)},K,k_0,k_1,\cdots,k_\ell)\right]}\\
        &+\frac{2\Pr\left[\set{{\sf PredicatePassed,GoodStructure,Accept1}}\subseteq{\sf Expt}_{{\sf Structured}}^{\cal{A}'}(\Pi_{\sf alter}^{(k)},K,k_0,k_1,\cdots,k_\ell)\right]}{\Pr\left[\set{{\sf PredicatePassed,GoodStructure}}\subseteq{\sf Expt}_{{\sf Structured}}^{\cal{A}'}(\Pi_{\sf alter}^{(k)},K,k_0,k_1,\cdots,k_\ell)\right]}\\
        &+\frac{4\Pr\left[\set{{\sf PredicatePassed,BadStructure}}\subseteq{\sf Expt}_{{\sf Structured}}^{\cal{A}'}(\Pi_{\sf alter}^{(k)},K,k_0,k_1,\cdots,k_\ell)\right]}{\Pr\left[{\sf PredicatePassed}\in{\sf Expt}_{{\sf Structured}}^{\cal{A}'}(\Pi_{\sf alter}^{(k)},K,k_0,k_1,\cdots,k_\ell)\right]}.
    \end{align*}
\end{lemma}
\begin{proof}
    \begin{align*}
        &\frac{\Pr\left[\set{{\sf PredicatePassed,Accept0}}\subseteq{\sf Expt}_{{\sf Standard}}^{\cal{A}'}(\Pi_{\sf alter}^{(k)},K)\right]}{\Pr\left[{\sf PredicatePassed}\in{\sf Expt}_{{\sf Standard}}^{\cal{A}'}(\Pi_{\sf alter}^{(k)},K)\right]}\\
        &+\frac{\Pr\left[\set{{\sf PredicatePassed,Accept1}}\subseteq{\sf Expt}_{{\sf Standard}}^{\cal{A}'}(\Pi_{\sf alter}^{(k)},K)\right]}{\Pr\left[{\sf PredicatePassed}\in{\sf Expt}_{{\sf Standard}}^{\cal{A}'}(\Pi_{\sf alter}^{(k)},K)\right]}\\
        =&\frac{\Pr\left[\set{{\sf PredicatePassed,Accept0}}\subseteq{\sf Expt}_{{\sf Compressed}}^{\cal{A}'}(\Pi_{\sf alter}^{(k)},K)\right]}{\Pr\left[{\sf PredicatePassed}\in{\sf Expt}_{{\sf Compressed}}^{\cal{A}'}(\Pi_{\sf alter}^{(k)},K)\right]}\\
        &+\frac{\Pr\left[\set{{\sf PredicatePassed,Accept1}}\subseteq{\sf Expt}_{{\sf Compressed}}^{\cal{A}'}(\Pi_{\sf alter}^{(k)},K)\right]}{\Pr\left[{\sf PredicatePassed}\in{\sf Expt}_{{\sf Compressed}}^{\cal{A}'}(\Pi_{\sf alter}^{(k)},K)\right]}\\
        \leq&\frac{2\Pr\left[\set{{\sf PredicatePassed,Accept0}}\subseteq{\sf Expt}_{{\sf Structured}}^{\cal{A}'}(\Pi_{\sf alter}^{(k)},K,k_0,k_1,\cdots,k_\ell)\right]}{\Pr\left[{\sf PredicatePassed}\in{\sf Expt}_{{\sf Structured}}^{\cal{A}'}(\Pi_{\sf alter}^{(k)},K,k_0,k_1,\cdots,k_\ell)\right]}\\
        &+\frac{2\Pr\left[\set{{\sf PredicatePassed,Accept1}}\subseteq{\sf Expt}_{{\sf Structured}}^{\cal{A}'}(\Pi_{\sf alter}^{(k)},K,k_0,k_1,\cdots,k_\ell)\right]}{\Pr\left[{\sf PredicatePassed}\in{\sf Expt}_{{\sf Structured}}^{\cal{A}'}(\Pi_{\sf alter}^{(k)},K,k_0,k_1,\cdots,k_\ell)\right]}\\
        \leq&\frac{2\Pr\left[\set{{\sf PredicatePassed,GoodStructure,Accept0}}\subseteq{\sf Expt}_{{\sf Structured}}^{\cal{A}'}(\Pi_{\sf alter}^{(k)},K,k_0,k_1,\cdots,k_\ell)\right]}{\Pr\left[\set{{\sf PredicatePassed,GoodStructure}}\subseteq{\sf Expt}_{{\sf Structured}}^{\cal{A}'}(\Pi_{\sf alter}^{(k)},K,k_0,k_1,\cdots,k_\ell)\right]}\\
        &+\frac{2\Pr\left[\set{{\sf PredicatePassed,GoodStructure,Accept1}}\subseteq{\sf Expt}_{{\sf Structured}}^{\cal{A}'}(\Pi_{\sf alter}^{(k)},K,k_0,k_1,\cdots,k_\ell)\right]}{\Pr\left[\set{{\sf PredicatePassed,GoodStructure}}\subseteq{\sf Expt}_{{\sf Structured}}^{\cal{A}'}(\Pi_{\sf alter}^{(k)},K,k_0,k_1,\cdots,k_\ell)\right]}\\
        &+\frac{4\Pr\left[\set{{\sf PredicatePassed,BadStructure}}\subseteq{\sf Expt}_{{\sf Structured}}^{\cal{A}'}(\Pi_{\sf alter}^{(k)},K,k_0,k_1,\cdots,k_\ell)\right]}{\Pr\left[{\sf PredicatePassed}\in{\sf Expt}_{{\sf Structured}}^{\cal{A}'}(\Pi_{\sf alter}^{(k)},K,k_0,k_1,\cdots,k_\ell)\right]}.
    \end{align*}
    The constant $2$ in the second equation is due to the additional 2-outcome measurement on the structure. 
\end{proof}
\hypertarget{step8}{\textbf{Step 8(a): }} Now we prove that it is unlikely for the state to have a bad structure.
\begin{lemma}\label{lem:bound_with_bad_structure}
    For $k_0=2kT$ and $k_1,\cdots,k_\ell$, if $T=\Omega(k_i)$ for every $k_i$ then
    \begin{equation*}
        \Pr\left[{\sf BadStructure}\in{\sf Expt}_{{\sf Structured}}^{\cal{A}'}\brackets{\Pi_{\sf alter}^{(k)},K,k_0,k_1,\cdots,k_\ell}\right]\leq\sum_{i=1}^{\ell}O\brackets{\frac{(2kT)^{\ell+1}}{k_i^{\frac{2}{i}}N}}^{k_i^{\frac{1}{i}}}.
    \end{equation*}
\end{lemma}
\begin{proof}
    Consider the following algorithm $\cal{B}$:
    \begin{enumerate}
        \item Start with the state $\rho_{\sf compressed}\otimes\ket{0,0,\cdots,0}_{\prfreg}\bra{0,0,\cdots,0}$ where $\prfreg$ is a register that is used to purify measurements. 
        \item Run $\Pi_{\sf alter}^{(k)}$, but for every projection (step 5(a) in \Cref{proj:single-game} and $\Pi_{\sf reset}$) write the result of that projection on $\prfreg$ to purify the process.
    \end{enumerate}
    The probability
    \begin{equation*}
        \Pr\left[{\sf BadStructure}\in{\sf Expt}_{{\sf Structured}}^{\cal{A}'}\brackets{\Pi_{\sf alter}^{(k)},K,k_0,k_1,\cdots,k_\ell}\right]
    \end{equation*}
    is the probability that after running $\cal{B}$ on $\rho_{\sf compressed}$ one of the following happens:
    \begin{enumerate}
        \item The number of non-$\bot$ entries on $\dbreg$ is more than $k_0$.
        \item There exists $i\in 1,2,\cdots,\ell$ and $y\in[N]$ such that the number of $i$-tuple of sum $y$ is more than $k_i$.
    \end{enumerate}
    Note that algorithm $\cal{B}$ queries the compressed oracle for at most $2kT$ times and the initial state $\rho_{\sf compressed}$ has an empty database. Thus, the first bullet above never happens. The second bullet for a fixed $i$ happens with probability $O\brackets{\frac{(2kT)^{\ell+1}}{k_i^{\frac{2}{i}}N}}^{k_i^{\frac{1}{i}}}$ by \Cref{thm:quantum-time-bound-for-finding-tuples-database}. By union bound, we can bound the probability by
    \begin{equation*}
        \Pr\left[{\sf BadStructure}\in{\sf Expt}_{{\sf Structured}}^{\cal{A}'}\brackets{\Pi_{\sf alter}^{(k)},K,k_0,k_1,\cdots,k_\ell}\right]\leq\sum_{i=1}^{\ell}O\brackets{\frac{(2kT)^{\ell+1}}{k_i^{\frac{2}{i}}N}}^{k_i^{\frac{1}{i}}}.
    \end{equation*}
\end{proof}
\hypertarget{step8}{\textbf{Step 8(b): }} Now we prove that if the state has a good structure, the success probability of later stages is bounded.
\begin{lemma}\label{lem:bound_with_good_structure}
    Condition on a the predicate passed and the structure being good, the success probability of the last round is bounded if $T=\Omega(K)$:
    \begin{align*}
        &\frac{\Pr\left[\set{{\sf PredicatePassed,GoodStructure,Accept0}}\subseteq{\sf Expt}_{{\sf Structured}}^{\cal{A}'}(\Pi_{\sf alter}^{(k)},K,k_0,k_1,\cdots,k_\ell)\right]}{\Pr\left[\set{{\sf PredicatePassed,GoodStructure}}\subseteq{\sf Expt}_{{\sf Structured}}^{\cal{A}'}(\Pi_{\sf alter}^{(k)},K,k_0,k_1,\cdots,k_\ell)\right]}\\
        &+\frac{\Pr\left[\set{{\sf PredicatePassed,GoodStructure,Accept1}}\subseteq{\sf Expt}_{{\sf Structured}}^{\cal{A}'}(\Pi_{\sf alter}^{(k)},K,k_0,k_1,\cdots,k_\ell)\right]}{\Pr\left[\set{{\sf PredicatePassed,GoodStructure}}\subseteq{\sf Expt}_{{\sf Structured}}^{\cal{A}'}(\Pi_{\sf alter}^{(k)},K,k_0,k_1,\cdots,k_\ell)\right]}\\
        \leq&O\brackets{\frac{(kT)^2(kT+k_0)^{\ell-1}}{K_{{\sf sol}}^2N}}^{K_{{\sf sol}}}.
    \end{align*}
    for any $k_0$ and $k_1,\cdots,k_\ell<K$ where $K_{{\sf sol}}$ is the only real non-negative root of the following function
    \begin{equation*}
        f_{K,k_1,k_2,\cdots,k_\ell}(x)=x^\ell+\sum_{i=0}^{\ell-1}k_{\ell-i}x^i-K.
    \end{equation*}
\end{lemma}
\begin{proof}
    To bound
    \begin{equation*}
        \frac{\Pr\left[\set{{\sf PredicatePassed,GoodStructure,Accept0}}\subseteq{\sf Expt}_{{\sf Structured}}^{\cal{A}'}(\Pi_{\sf alter}^{(k)},K,k_0,k_1,\cdots,k_\ell)\right]}{\Pr\left[\set{{\sf PredicatePassed,GoodStructure}}\subseteq{\sf Expt}_{{\sf Structured}}^{\cal{A}'}(\Pi_{\sf alter}^{(k)},K,k_0,k_1,\cdots,k_\ell)\right]}
    \end{equation*}
    in the last inequality, we consider the meaning of this probability. This probability is less than or equal to the success probability of this algorithm $\cal{B}_0$:
    \begin{enumerate}
        \item The initial state is 
        \begin{equation*}
            \rho=\frac{\Pi_{\sf structure}\Pi'_{\sf predicate}\rho_{\sf compressed}\Pi'_{\sf predicate}\Pi_{\sf structure}}{\trace{\Pi_{\sf structure}\Pi'_{\sf predicate}\rho_{\sf compressed}\Pi'_{\sf predicate}\Pi_{\sf structure}}}
        \end{equation*}
        where $\Pi'_{\sf predicate}={\sf DecompressAll}^\dagger\Pi_{\sf alter}^{(k)}{\sf DecompressAll}$.
        \item Controlled by the content on $\outreg$, let $Q=\bigcup_{r_{(x_1,x_2,\cdots,x_\ell)}=1}\set{x_1,x_2,\cdots,x_\ell}$ be all the points that are contained in the algorithm's output, swap them one by one to $\qryreg$ and run $\stddecomp$ every time to completely decompress the oracle into the standard oracle in these points.
        \item Count the maximum number of $\ell$-tuples that we can find in the compressed oracle register $\dbreg$. The algorithm succeeds if this number is at least $K$.
    \end{enumerate}
    Note that algorithm $\cal{B}_0$ uses at most $O(kT)$ $\stddecomp$ because the number of $1$s on $\outreg$ is at most $O(kT)$ according to the definition of $\Pi_{\sf alter}^{(k)}$. Thus by \Cref{thm:quantum-time-bound-for-finding-tuples-database} the success probability of $\cal{B}_0$ is at most $O\brackets{\frac{(kT)^2(kT+k_0)^{\ell-1}}{K_{{\sf sol}}^2N}}^{K_{{\sf sol}}}$.
    Similarly, 
    \begin{equation*}
        \frac{\Pr\left[\set{{\sf PredicatePassed,GoodStructure,Accept1}}\subseteq{\sf Expt}_{{\sf Structured}}^{\cal{A}'}(\Pi_{\sf alter}^{(k)},K,k_0,k_1,\cdots,k_\ell)\right]}{\Pr\left[\set{{\sf PredicatePassed,GoodStructure}}\subseteq{\sf Expt}_{{\sf Structured}}^{\cal{A}'}(\Pi_{\sf alter}^{(k)},K,k_0,k_1,\cdots,k_\ell)\right]}
    \end{equation*}
    is less than or equal to two times (due to the 2-outcome measurement on step 3) of the success probability of this algorithm:
    \begin{enumerate}
        \item The initial state is 
        \begin{equation*}
            \rho=\frac{\Pi_{\sf structure}\Pi'_{\sf predicate}\rho_{\sf compressed}\Pi'_{\sf predicate}\Pi_{\sf structure}}{\trace{\Pi_{\sf structure}\Pi'_{\sf predicate}\rho_{\sf compressed}\Pi'_{\sf predicate}\Pi_{\sf structure}}}
        \end{equation*}
        where $\Pi'_{\sf predicate}={\sf DecompressAll}^\dagger\Pi_{\sf alter}^{(k)}{\sf DecompressAll}$.
        \item Run $\cal{A}'$. All oracle queries use $\cpho$.
        \item Controlled by the content on $\outreg$, let $Q=\bigcup_{r_{(x_1,x_2,\cdots,x_\ell)}=1}\set{x_1,x_2,\cdots,x_\ell}$ be all the points that are contained in the algorithm's output, swap them one by one to $\qryreg$ and run $\stddecomp$ every time to completely decompress the oracle into the standard oracle in these points.
        \item Count the maximum number of $\ell$-tuples that we can find in the compressed oracle register $\dbreg$. The algorithm succeeds if this number is at least $K$.
    \end{enumerate}
    Note that algorithm $\cal{B}_1$ uses at most $O((k+1)T)$ $\stddecomp$ because the number of $1$s on $\outreg$ is at most $O((k+1)T)$ and at most $O(T)$ $\cpho$ during step 2. Thus, by \Cref{thm:quantum-time-bound-for-finding-tuples-database} the success probability of $\cal{B}_0$ is also at most $O\brackets{\frac{(kT)^2(kT+k_0)^{\ell-1}}{K_{{\sf sol}}^2N}}^{K_{{\sf sol}}}$.
\end{proof}
\begin{remark}
    One may wonder what if after projecting the state on $\Pi_{\sf structure}$ and run ${\sf DecompressAll}$, there still exists $\bot$ entries in the oracle. This will not happen because the Hadamard basis is the eigenbasis for $\Pi_{\sf structure}$.
\end{remark}
\begin{theorem}[Success probability in finding a random $\ell$-tuple with advice]\label{thm:non-uniform-random-challenge}
    For all $\set{y_H}$ and all algorithms $\cal{A}'$ with $T$ oracle query and $T$ $\output$ operations where $T=\Omega(S^{2\ell})$. Let $\ket{\psi_{\sf st}}$ be any $S$-qubit advice for the algorithm. The probability of $\cal{A}'$ succeeds in finding randomly generated $\ell$-tuple of sum $y_H$ is at most $O\brackets{\frac{S^{2\ell}}{M^\ell}}+O\brackets{\frac{(ST)^{\ell+1}}{N}}^{S}$. \zihan{is it the tighter bound or typo?}\zikuan{Typo.} That is
    \begin{equation*}
        P^{(1)}_{\sf FlipTest}=\abs{\Pi_{{\sf FlipTest}}^{\cal{A}',\set{y_H}}\brackets{\ket{\psi_{\sf st}}\ket{0}_{\ancillareg}\ket{0}_{\seedreg}\ket{0}_{\memreg}}}^2\leq O\brackets{\frac{S^{2\ell}}{M^\ell}}+O\brackets{\frac{(ST)^{\ell+1}}{N}}^{S}
    \end{equation*}
    where $\Pi_{{\sf FlipTest}}^{\cal{A}',\set{y_H}}$ is defined in \Cref{proj:single-game}.
\end{theorem}
\begin{proof}
    Set $k_0=2ST$, $k_i=S^{\ell+i}$ and $K=2S^{2\ell}$.
    \begin{claim}
        Let $K_{\sf sol}$ be the only non-negative real root for function
        \begin{equation*}
            f(x)=x^\ell+\sum_{i=0}^{\ell-1}k_{\ell-i}x^i-K=x^\ell+\sum_{i=0}^{\ell-1}S^{2\ell-i}x^i-S^{2\ell}.
        \end{equation*}
        We have that $K_{\sf sol}\geq S$.
    \end{claim}
    \begin{proof}
        \begin{align*}
            f(S+1)=&(S+1)^{\ell}+S^\ell\frac{(S+1)^{\ell}-S^\ell}{(S+1)-S}-2S^{2\ell}\\
            \leq&(S+1)^{\ell}S^\ell(S+1)^{\ell}-2S^{2\ell}<0.
        \end{align*}
        Since $f(x)$ is monotonically increasing on $(0,+\infty)$, $K_{\sf sol}\geq S$.
    \end{proof}
    Now we assume that
    \begin{equation*}
        P^{{\sf uniform},(S)}_{\sf FlipTest}\geq\brackets{\frac{16K}{M^\ell}}^{2S}
    \end{equation*}
    cause Otherwise, this theorem holds immediately by
    \begin{equation*}
        P^{(1)}_{\sf FlipTest}\leq2\brackets{P^{{\sf uniform},(S)}_{\sf FlipTest}}^{\frac{1}{2S-1}}\leq2\brackets{\frac{16K}{M^\ell}}^{\frac{2S}{2S-1}}\leq\frac{32K}{M^\ell}.
    \end{equation*}
    Now we begin the proof:
    \begin{align*}
        &P^{(1)}_{\sf FlipTest}\leq2\brackets{P^{{\sf uniform},(S)}_{\sf FlipTest}}^{\frac{1}{2S-1}}\leq\frac{2P^{{\sf uniform},(S)}_{\sf FlipTest}}{P^{{\sf uniform},(S-1)}_{\sf FlipTest}}\\
        \leq&\frac{2\Pr\left[\set{\sf PredicatePassed, Accepted}\subseteq{\sf Expt}_{\sf FlipTest}^{\cal{A}',\set{y_H}}(S)\right]}{\Pr\left[{\sf PredicatePassed}\in{\sf Expt}_{\sf FlipTest}^{\cal{A}',\set{y_H}}(S)\right]}\\
        \leq&\frac{16K}{M^\ell}+\frac{8\Pr\left[\set{{\sf PredicatePassed,Accept0}}\subseteq{\sf Expt}_{{\sf Standard}}^{\cal{A}'}\brackets{\Pi_{\sf alter}^{(S-1)},K}\right]}{\Pr\left[{\sf PredicatePassed}\in{\sf Expt}_{{\sf Standard}}^{\cal{A}'}\brackets{\Pi_{\sf alter}^{(S-1)},K}\right]}\\
        &+\frac{8\Pr\left[\set{{\sf PredicatePassed,Accept1}}\subseteq{\sf Expt}_{{\sf Standard}}^{\cal{A}'}\brackets{\Pi_{\sf alter}^{(S-1)},K}\right]}{\Pr\left[{\sf PredicatePassed}\in{\sf Expt}_{{\sf Standard}}^{\cal{A}'}\brackets{\Pi_{\sf alter}^{(S-1)},K}\right]}\\
        \leq&\frac{16K}{M^\ell}+\frac{16\Pr\left[\set{{\sf PredicatePassed,GoodStructure,Accept0}}\subseteq{\sf Expt}_{{\sf Structured}}^{\cal{A}'}(\Pi_{\sf alter}^{(S-1)},K,k_0,k_1,\cdots,k_\ell)\right]}{\Pr\left[\set{{\sf PredicatePassed,GoodStructure}}\subseteq{\sf Expt}_{{\sf Structured}}^{\cal{A}'}(\Pi_{\sf alter}^{(S-1)},K,k_0,k_1,\cdots,k_\ell)\right]}\\
        &+\frac{16\Pr\left[\set{{\sf PredicatePassed,GoodStructure,Accept1}}\subseteq{\sf Expt}_{{\sf Structured}}^{\cal{A}'}(\Pi_{\sf alter}^{(S-1)},K,k_0,k_1,\cdots,k_\ell)\right]}{\Pr\left[\set{{\sf PredicatePassed,GoodStructure}}\subseteq{\sf Expt}_{{\sf Structured}}^{\cal{A}'}(\Pi_{\sf alter}^{(S-1)},K,k_0,k_1,\cdots,k_\ell)\right]}\\
        &+\frac{32\Pr\left[{\sf BadStructure}\in{\sf Expt}_{{\sf Structured}}^{\cal{A}'}\brackets{\Pi_{\sf alter}^{(S-1)},K,k_0,k_1,\cdots,k_\ell}\right]}{\Pr\left[{\sf PredicatePassed}\in{\sf Expt}_{{\sf Structured}}^{\cal{A}'}\brackets{\Pi_{\sf alter}^{(S-1)},K,k_0,k_1,\cdots,k_\ell}\right]}\\
        \leq&\frac{16K}{M^\ell}+O\brackets{\frac{(ST)^2(ST+k_0)^{\ell-1}}{K_{{\sf sol}}^2N}}^{K_{{\sf sol}}}+\frac{\sum_{i=1}^{\ell}O\brackets{\frac{(2ST)^{i+1}}{k_i^{\frac{2}{i}}N}}^{k_i^{\frac{1}{i}}}}{\abs{\brackets{\Pi_{\sf reset}\Pi_{{\sf FlipTest}}^{\cal{A}',\set{y_H}}}^{S-1}\rho_{\sf compressed}\brackets{\Pi_{{\sf FlipTest}}^{\cal{A}',\set{y_H}}\Pi_{\sf reset}}^{S-1}}}\\
        \leq&\frac{16K}{M^\ell}+O\brackets{\frac{(ST)^2(ST+k_0)^{\ell-1}}{K_{{\sf sol}}^2N}}^{K_{{\sf sol}}}+\frac{\sum_{i=1}^{\ell}O\brackets{\frac{(2ST)^{i+1}}{k_i^{\frac{2}{i}}N}}^{k_i^{\frac{1}{i}}}}{P^{{\sf uniform},(S)}_{\sf FlipTest}}\\
        \leq&\frac{16K}{M^\ell}+O\brackets{\frac{(ST)^2(ST+k_0)^{\ell-1}}{K_{{\sf sol}}^2N}}^{K_{{\sf sol}}}+\frac{\sum_{i=1}^{\ell}O\brackets{\frac{(2ST)^{i+1}}{k_i^{\frac{2}{i}}N}}^{k_i^{\frac{1}{i}}}}{\brackets{\frac{16K}{M^\ell}}^{2S}}\\
        \leq&\frac{32S^{2\ell}}{M^\ell}+O\brackets{\frac{(ST)^{\ell+1}}{N}}^{S}+\frac{O\brackets{\frac{(ST)^{i+1}}{S^4N}}^{S^2}}{\brackets{\frac{32S^{2\ell}}{M^\ell}}^{2S}}\\
        \leq&O\brackets{\frac{S^{2\ell}}{M^\ell}}+O\brackets{\frac{(ST)^{\ell+1}}{N}}^{S}+O\brackets{\frac{S^{\ell-3}T^{\ell+1}}{N}}^{S^2}\\
        \leq&O\brackets{\frac{S^{2\ell}}{M^\ell}}+O\brackets{\frac{(ST)^{\ell+1}}{N}}^{S}.
    \end{align*}
    Here are the reasons on why these inequalities hold:
    \begin{itemize}
        \item The first and second inequality follows from \Cref{lem:single_game_bound_1}, \Cref{lem:single_game_bound_2} and \Cref{lem:alternative_game_monotone}.
        \item The third inequality follows from a direct interpretation.
        \item The fourth inequality follows from \Cref{lem:switch_between_expectation_and_tail_bound}.
        \item The fifth inequality follows from \Cref{lem:split_good_bad_structure}.
        \item The sixth and seventh inequality follows from \Cref{lem:bound_with_bad_structure}, \Cref{lem:bound_with_good_structure} and an interpretation of the term
        \begin{equation*}
            \Pr\left[{\sf PredicatePassed}\in{\sf Expt}_{{\sf Structured}}^{\cal{A}'}\brackets{\Pi_{\sf alter}^{(S-1)},K,k_0,k_1,\cdots,k_\ell}\right].
        \end{equation*}
        \item The eighth inequality holds because we assume $P^{{\sf uniform},(S)}_{\sf FlipTest}\geq\brackets{\frac{16K}{M^\ell}}^{2S}$.
    \end{itemize}
\end{proof}
\begin{corollary}[Success probability in finding a random $\ell$-tuple with mixed advice]\label{cor:non-uniform-random-challenge}
    For all $\set{y_H}$ and all algorithms $\cal{A}'$ with $T$ oracle query and $T$ $\output$ operations where $T=\Omega(S^{2\ell})$. Let 
    \begin{equation*}
        \rho_{\sf st}=\frac{1}{N^M2^{M^\ell}}\sum_{H,\hat{r}}(\rho_{H,\hat{r}})_{\qryreg\ansreg\auxreg}\otimes\ket{\hat{r}}_{\outreg}\bra{\hat{r}}\otimes\ket{H}_{\dbreg}\bra{H},
    \end{equation*} 
    be any $S$-qubit mixed state advice for the algorithm. The probability of $\cal{A}'$ succeeds in finding a randomly generated $\ell$-tuple of sum $y_H$ is at most $O\brackets{\frac{S^{2\ell}}{M^\ell}}+O\brackets{\frac{(ST)^{\ell+1}}{N}}^{S}$. 
\end{corollary}
\begin{proof}
    In \Cref{thm:non-uniform-random-challenge}, the probability is linear in the input state $\ket{\psi_{\sf st}}$. By averaging, you get the same bound for any mixed state advice.
\end{proof}
\section{\texorpdfstring{The time-space bound for finding $K$ $\ell$-tuples with the same sum}{The time-space bound for finding K l-tuples with the same sum}}

\newcommand{\cB}{\mathcal{B}}
\newcommand{\fliptest}{\mathsf{FlipTest}}
\newcommand{\avgtest}{\mathsf{AvgFlipTest}}
\newcommand{\cnot}{\mathsf{CNOT}}
\newcommand{\swap}{\mathsf{SWAP}}

\newcommand{\yh}{\{y_H\}}
\newcommand{\bigk}{\widehat{\mathbf{K}}}
\newcommand{\capacity}{\mathsf{Cap}}
\newcommand{\allreg}{\qryreg\ansreg\auxreg\outreg\dbreg}
\newcommand{\workreg}{\qryreg\ansreg\auxreg}

\newcommand{\boundedflip}{$p_e$-bounded flipping}
\newcommand{\maxcapacity}{same-sum $\ell$-tuple capacity}
\newcommand{\subsetcapacity}{sum-$\{y_H\}$ $\ell$-tuple capacity}
\newcommand{\qubitsubsetcapacity}[1]{sum-$\{y_H\}$ $#1$-th qubit $\ell$-tuple capacity}

\newcommand{\game}[4]{
\begin{boxfig}{H}{\footnotesize 
\centering{\textbf{#1}}
    #4
\vspace{0.2em} } \caption{\label{#3} #2}
\end{boxfig}
}

\subsection{The classical bound}
In this section we bound the probability of a classical algorithm outputting $K$ $\ell$-tuples with the same sum using $S$-bit space and $T$ oracle queries. The model we are considering is that there is a write-only $M^\ell$-bit classical memory where each bit indicates whether its corresponding $\ell$-tuple is added to the output ($0$ for no and $1$ for yes). The algorithm can only access this part of the memory by making an oracle call to the ${\sf FLIP}(x_1,x_2,\cdots,x_\ell)$ functionality that flips the bit that corresponds to the $\ell$-tuple $(x_1,x_2,\cdots,x_\ell)$. The function ${\sf FLIP}$ will also reply to the algorithm with a random generated value $g_{(x_1,x_2,\cdots,x_\ell)}$ that is only generated once for the first time that the algorithm calls ${\sf FLIP}(x_1,x_2,\cdots,x_\ell)$ and remains the same for later calls. We say that the algorithm finds $K$ $\ell$-tuples with the same sum iff all bits that are $1$ correspond to tuples with the same sum and the number of such bits is at least $K$. Here we provide a statement stronger than what we need. Instead of proving a time-space bound for finding $K$ $\ell$-tuples of sum $y$, we prove a time-space bound for finding $K$ $\ell$-tuples with the same sum.
\begin{theorem}\label{thm:classical_time_space_bound_for_finding_tuples}
    Let $H:[M]\rightarrow[N]$ be a random oracle with sufficiently large $M$ polynomial in $N$. For $\ell\geq 2$ and for any classical algorithm with space $S$ and $T$ oracle queries on either $H$ or ${\sf FLIP}$. Then the expectation of the maximum number of $\ell$-tuples with the same sum found by the algorithm is $O\brackets{\frac{S^{\frac{\ell^2-1}{\ell}}T}{N^{\frac{1}{\ell}}}}$ if $S=\Omega(\log N)$ and $S=O\brackets{N^{\frac{1}{\ell^2-1}}}$.
\end{theorem}
\begin{proof}
    Divide the algorithm into $L=T/T'$ segments with each of them querying the random oracle for $T'=cS^{\frac{1}{\ell}}N^{\frac{1}{\ell}}=\Omega(S^\ell)$ times. By setting $c$ properly, by \Cref{thm:classical-time-bound-for-finding-tuples-final}, we can guarantee that the probability of outputting at least $S^{\ell}$ $\ell$-tuples with the same sum is at most $\frac{1}{M^\ell}$ for all segments since $M$ is a polynomial of $N$. This is because the advice for each segment (i.e. the internal state at the start of that segment) is of size $S$ and thus only provides a boost of $2^S$ on the probability in \Cref{thm:classical-time-bound-for-finding-tuples-final}. The exponent, $K^{\frac{1}{\ell}}=S=\Omega(\log N)$ makes it possible to absorb $2^S$ into any factor that is polynomial in $N$ into the constant $c$. With such $c$, the expected number of $\ell$-tuples with the same sum found in every segment is at most $S^{\ell}+1$. Thus the overall expectation of the maximum number of $\ell$-tuples with the same sum found by the algorithm is at most $O\brackets{LS^{\ell}}=O\brackets{\frac{S^{\frac{\ell^2-1}{\ell}}T}{N^{\frac{1}{\ell}}}}$.
\end{proof}

\subsection{The quantum bound}
Now we bound the \subsetcapacity{}(\Cref{def:subset-capacity}) of the state after an algorithm with $S$-qubit space and $T$ oracle queries to two oracles $H:[M]\to[N] ,G:[M]^\ell\to [N_0]$. That is to say, we take the imaginary process of measuring $\dbreg$ to obtain $H$, and measure $\outreg$ to obtain $\widehat{r}$, then look at the entries with indices $(x_1\cdots,x_{\ell})$ that have sum $y_H$, and then count the number of non-zero entries among these entries.  We specifically ask that the algorithm queries $G$ through $\Hadamardo$. 

We first describe the settings of all registers. Considering that a space-bounded algorithm can introduce ancilla qubits into the working registers and discard some of them on the fly, we specify an ancilla register $\ancillareg$ to include these qubits. $\ancillareg$ can be of arbitrary size, and an algorithm could apply unitary across $\ancillareg$ and $\workreg$, but it can only pass the state in working registers $\workreg$ to the next timestep but not $\ancillareg$. In other words, different steps of algorithms can only access disjoint parts of $\ancillareg$. The registers $\workreg\ancillareg\dbreg\outreg$ are initialized as:
\begin{equation*}
    \ket{\phi_{\sf st}}_{\workreg\ancillareg\dbreg\outreg} =\frac{1}{\sqrt{N^M N_0^{M^\ell}}}\sum_{H,G}\ket{\psi^{(0)}_{H,G}}_{\qryreg\ansreg\auxreg\ancillareg}\ket{H}_{\dbreg}\ket{\widehat{G}}_{\outreg}
\end{equation*}
where $\ket{\psi^{(0)}_{H,G}}_{\workreg\ancillareg}$ is the purification of the advice under $H,G$, while the algorithm could only access the mixed state in $\workreg$. For a uniform algorithm, we set $\ket{\psi^{(0)}_{H,G}}_{\workreg\ancillareg} = \ket{0}_{\workreg\ancillareg}$ for all $H,G$.
Here, $\ket{\widehat{G}}$ is the Hadamard basis state corresponding to $G\in\set{0,1}^{M^\ell}$, we can also write it as a tensor product of entry-wise Hadamard basis,
\begin{equation*}
    \ket{\widehat{G}}_{\outreg} = \frac{1}{\sqrt{N_0^{M^\ell}}}\sum_{R\in \set{0,1}^{M^\ell}} (-1)^{\langle G,R \rangle}\ket{R}_{\outreg} = \bigotimes_{i\in[M]^\ell} \frac{1}{\sqrt{N_0}}\sum_{r\in\set{0,1}} (-1)^{\langle G(i),R(i)\rangle} \ket{R(i)}_{\outreg_i} = \bigotimes_{i\in[M]^\ell} \ket{\widehat{G(i)}}_{\outreg_i},
\end{equation*}
with $\outreg_i$ being the $i$-th entry of the register with size $N_0$. 


We specify that all queries made to $G$ need to be performed through $\Hadamardo$ (in this section all $\Hadamardo$ refers to the Hadamard oracle query to $G$), which is consistent with the setting in the last section. Specifically, here we consider the $\Hadamardo$ oracle applied on $\outreg$,  $\Hadamardo\ket{x,u,w,a,H}_{\workreg\ancillareg\dbreg}\ket{R}_{\outreg} = \ket{x,u,w,a,H}_{\workreg\ancillareg\dbreg}\ket{R\oplus (x,u)}_{\outreg}$. Then, if we compare $\Hadamardo$ on $\outreg$ and the $\output$ operator for outputting $\ell$-tuples, they are identical. The initialization of registers here is also consistent with the last section. Specifically, for uniform algorithms, $\outreg$ is initialized as $\ket{0,\cdots,0} = \frac{1}{\sqrt{N_0^{M^\ell}}}\sum_{G\in\set{0,1}^{M^\ell}} \ket{\widehat{G}}$. 

Therefore, we could transform all results on algorithms with $H$ queries and $\output$ operations to results on algorithms with $H$ queries and $\Hadamardo$ $G$ queries. We denote the final state in the registers $\workreg\ancillareg\dbreg\outreg$ as:
\begin{equation*}
    \ket{\phi_{\mathsf{fin}}}_{\workreg\ancillareg\dbreg\outreg} =\frac{1}{\sqrt{N^M N_0^{M^\ell}}}\sum_{H,G}\ket{\psi^{\sf{fin}}_{H,G}}_{\qryreg\ansreg\auxreg\ancillareg}\ket{H}_{\dbreg}\ket{\widehat{G}}_{\outreg},
\end{equation*}
where $\ket{\psi^{\sf{fin}}_{H,G}}$ is the state in $\workreg\ancillareg$ after running algorithm $\cA$ under oracle $H,G$. We describe the final state as a pure state, since we can always extend the ancilla register $\ancillareg$ to purify the entire state. We further define our target rigorously, which we call \textbf{\subsetcapacity}, as the expected value of an observable on our final state $\ket{\phi_{\sf fin}}$.

\begin{definition}[\subsetcapacity]
    \label{def:subset-capacity}
    For any state in registers $\workreg\ancillareg\dbreg\outreg$, we define its \textbf{\subsetcapacity}:
    \begin{itemize}
        \item For $R\in [N_0]^{M^\ell}$, define its \textbf{subset-$Y$ capacity} relative to a subset $Y\subseteq [M]^\ell$,
        \begin{equation*}
            \capacity_{Y}(R) = \sum_{i\in Y} \mathds{1}\{R_i\neq 0\}.
        \end{equation*}

        \item For $R\in [N_0]^{M^\ell}$, view it as a $M^\ell$ entry table, then each entry has $n_0=\log N_0$ qubits, define its \textbf{subset-$Y$ $j$-th qubit capacity} relative to $Y\subseteq [M]^\ell$,
        \begin{equation*}
            \capacity^{(j)}_{Y}(R) = \sum_{i\in Y} \mathds{1}\{\textnormal{the $j$-th bit of } R_i=1\}.
        \end{equation*}

        \item Relative to a sequence $\yh$, define observable $\bigk_{\yh}$, as well as observable $\bigk_{\yh}^{(j)}$ for $j=1,\cdots,n_0$,
        \begin{equation*}
        \begin{split}
            \bigk_{\yh} =& I_{\qryreg\ansreg\auxreg\ancillareg}\otimes \sum_{H,R} \capacity_{Y_H}(R)\cdot \ketbra{H}{H}_{\dbreg}\otimes\ketbra{R}{R}_{\outreg},\\
            \bigk^{(j)}_{\yh} =& I_{\qryreg\ansreg\auxreg\ancillareg}\otimes \sum_{H,R} \capacity^{(j)}_{Y_H}(R)\cdot \ketbra{H}{H}_{\dbreg}\otimes\ketbra{R}{R}_{\outreg},
        \end{split}
        \end{equation*}
        while recall $Y_H$ is the set of $\ell$-tuples with sum $y_H$ decided by $\yh$ and $H$.
        
        

        \item For algorithm $\cA$, let the final state on $\workreg\ancillareg\dbreg\outreg$ be $\ket{\phi_{\sf fin}}$. We define the \textbf{\subsetcapacity{}} $K_{\yh}(\cA)$ and \textbf{\qubitsubsetcapacity{j}} $K^{(j)}_{\yh}(\cA)$ of $\cA$ relative to $\yh$, by taking the expected value of the observables on its final state,
        \begin{equation*}
            K_{\yh}(\cA):= \bra{\phi_{\sf fin}}\bigk_{\yh}\ket{\phi_{\sf fin}},\ \  K^{(j)}_{\yh}(\cA):= \bra{\phi_{\sf fin}}\bigk^{(j)}_{\yh}\ket{\phi_{\sf fin}}.
        \end{equation*}
        

    \end{itemize}
    
\end{definition}


In other words, $K^{(j)}_{\yh}(\cA)$ essentially measures the output register $\outreg$ after running the algorithm $\cA$, and counts the number of entries with $j$-th qubit being flipped and indices constrained by $\yh$, while $K_{\yh}(\cA)$ counts the index-constrained entries with any bit being flipped. To bound this property, we adopt a reordering procedure for oracle $G$, to help the analysis of $K_{\yh}(\cA)$.

\zihan{brief introduction to the reordering process}
The process follows an observation that $K_{\yh}(\cA)\leq\sum_{j=1}^{n_0}K_{\yh}^{(j)}(\cA)$, as a non-zero entry implies at least one non-zero bit. Therefore we turn to bound each $K_{\yh}^{(j)}(\cA)$ individually, by viewing $G$ as a concatenation of $n_0$ one-bit oracles, $g^{(1)},\cdots,g^{(n_0)}\in \bit^{M^\ell}$, which leads to the first step of reordering, splitting output oracle $G$. We split the register $\outreg$ into $n_0$ registers $\outreg^{(1)}, \cdots, \outreg^{(n_0)}$, where they each have the same number of entries as $\outreg$, and one qubit per entry. $\outreg^{(j)}$ simply takes the $j$-th qubit from every entry in $\outreg$, and therefore stores $g^{(j)}$ in Hadamard basis. The second step is fixing all but one of these $n_0$ oracles. Let $\outreg^{(-j)}=\otimes_{j'\neq j} \outreg^{(j')}$ be the $n_0-1$ registers other than $\outreg^{(j)}$, notice that the measurement of observable $K_{\yh}^{(j)}(\cA)$ acts trivially on $\outreg^{(-j)}$, so we could fix these $n_0-1$ oracles $\{g^{(j')}\}_{j'\neq j}$ as $G' \in \bit^{n_0-1}$ and examine the observable value for each $G'$, and then take expectation to get $K_{\yh}^{(j)}(\cA)$. For fixed $G'$, we remain the $\Hadamardo$ queries to $g^{(j)}$, and replace all $\Hadamardo$ queries to other $g^{(j')}$ ($j' \neq j$) with a fixed unitary hardcoded by $g^{(j')}$, then the algorithm $\cA$ could be simplified as a $G'$-hardcoded algorithm $\cA[j, G']$ that queries $g^{(j)}$ and $H$. This way, we focus on bounding $K_{\yh}(\cA[j, G'])$, reducing the problem to the case when the oracle $g^{(j)}$ in the output register $\outreg^{(j)}$ (stored in computation basis) is a one-bit oracle. 


We then introduce our theorem in bounding $K_{\yh}(\cA)$.

\begin{theorem}
\label{thm:two-oracle-K-l-tuple-time-space-tradeoff}
    Let $H:[M]\rightarrow[N]$, $G:[M]^\ell\rightarrow[N_0] $ be two random oracles. For any $\yh$ and any uniform quantum algorithm $\cA$ with space $S$ , $T$ queries to $H$, and $T$ $\Hadamardo$ queries to $G$. Then, we have the \subsetcapacity{} $K_{\yh}(\cA)$ is upper bounded by $\widetilde{O}\brackets{S^{2\ell+2}T^2 N^{-\frac{2}{\ell+1}}} $ if $S = \Omega(\log N)$ and $S=O\brackets{N^\frac{1}{(\ell+1)(2\ell+1)}}$. 
\end{theorem}

\begin{proof}

    We first apply the \textbf{reordering process} to algorithm $\cA$. First, we split oracle $G$ into $n_0$ one-bit oracles, where $n_0 = \log N_0$, viewing $G$ as an $n_0$-bit oracle. Define one-bit oracles $g^{(1)}, g^{(2)}, \cdots, g^{(n_0)} : [M]^\ell \to \{0, 1\}$, where $g^{(j)}$ maps any input $x$ to the $j$-th bit of $G(x)$. Denote $g^{(j)}_x := g^{(j)}(x)$ as the one-bit entry of $g^{(j)}$ indexed by $x$, then $G(x)$ can be represented by a $\mathbb{F}_2^{n_0}$ vector $(g^{(1)}_x, g^{(2)}_x, \cdots, g^{(n_o)}_x)$. We also split the output register $\outreg$ into $\outreg^{(1)}, \cdots, \outreg^{(n_0)}$, where for each $j$, $\outreg^{(j)}$ is an $M^\ell$-entry, one-qubit-per-entry register that takes the $j$-th qubit of $\outreg$.

    For each $\Hadamardo$ query to $G$, we view it under the Hadamard basis. Let $G_x = G(x)$, and $G_{-x}$ denotes the concatenation of values in all $M^\ell - 1$ entries other than the $x$-th entry, then
    \begin{equation*}
    \begin{split}
    \Hadamardo\cdot \ket{x,u,w,a,H}_{\workreg\ancillareg\dbreg} \ket{\widehat{G}}_{\outreg} 
    &= \Hadamardo\cdot \ket{x,u,w,a,H} \ket{\widehat{G_x}}_{\outreg_x} \ket{\widehat{G_{-x}}}_{\outreg_{-x}}\\
    &= \Hadamardo\cdot \ket{x,u,w,a,H} \otimes\frac{1}{\sqrt{N_0}}\sum_{r\in[N_0]} (-1)^{\langle r,G_x\rangle}\ket{r}_{\outreg_x} \ket{\widehat{G_{-x}}}_{\outreg_{-x}}\\
    &= \ket{x,u,w,a,H}\otimes\frac{1}{\sqrt{N_0}}\sum_{r\in[N_0]}(-1)^{\langle r,G_x\rangle}\ket{r\oplus u}_{\outreg_x} \ket{\widehat{G_{-x}}}_{\outreg_{-x}}\\
    &= (-1)^{\langle u,G_x\rangle} \ket{x,u,w,a,H}\ket{\widehat{G_x}}_{\outreg_x}\ket{\widehat{G_{-x}}}_{\outreg_{-x}}\\
    &= (-1)^{\langle u,G_x\rangle} \ket{x,u,w,a,H}\ket{\widehat{G}}_{\outreg}. 
    \end{split}
    \end{equation*}
    Therefore, $\Hadamardo$ actually applies a phase kick-back onto the state controlled by the oracle $G$ in register $\outreg$. Moreover, after we split $\outreg$ into $\outreg^{(1)},\outreg^{(2)},\cdots,\outreg^{(n_0)}$ , we have,
    \begin{equation*}
        \Hadamardo\cdot \ket{x,u,w,a,H}_{\workreg\ancillareg\dbreg} \bigotimes_{j=1}^{n_0}\ket{\widehat{g^{(j)}}}_{\outreg^{(j)}} = \prod_{j=1}^{n_0}(-1)^{u^{(j)}\cdot g^{(j)}_x} \ \ket{x,u,w,a,H}_{\workreg\ancillareg\dbreg} \bigotimes_{j=1}^{n_0}\ket{\widehat{g^{(j)}}}_{\outreg^{(j)}},
    \end{equation*}
    where $u^{(j)}$ is the $j$-th bit of $u\in[N_0]$. Then, it could be viewed as the process of querying $g^{(1)}, g^{(2)}, \cdots, g^{(n_0)}$ in parallel, with phase kick-back multiplied together.

    
    We then switch our target from bounding $K_{\yh}(\cA)$ to bounding $K^{(j)}_{\yh}(\cA)$ for each $j=1,\cdots,n_0$. Notice that for any $R \in [N_0]^{M^\ell}$, $\capacity_Y(R) \leq \sum_{j=1}^{n_0} \capacity_Y^{(j)}(R)$. Then, we have $K_{\yh}(\cA) \leq \sum_{j=1}^{n_0} K^{(j)}_{\yh}(\cA)$. We therefore only need to bound $K^{(j)}_{\yh}(\cA)$ for each $j=1,\cdots,n_0$. We further denote $R^{(j)} \in \{0,1\}^{M^\ell}$ as the one-bit table to take the $j$-th bit of $R$ on each entry, then
    \begin{equation*}
    \begin{split}
    \bigk^{(j)}_{\yh} &= I_{\workreg\ancillareg} \otimes I_{\outreg^{(-j)}} \otimes \sum_{H, R^{(j)}} \capacity_{Y_H}(R^{(j)}) \cdot \ketbra{H}{H}_{\dbreg}\otimes\ketbra{R^{(j)}}{R^{(j)}}_{\outreg^{(j)}} \\
    &=: I_{\workreg\ancillareg} \otimes I_{\outreg^{(-j)}} \otimes \left(\Pi_{(j),\yh}^{K}\right)_{\outreg^{(j)}},
    \end{split}
    \end{equation*}
    where we define $\Pi_{(j),\yh}^{K}$ as its nontrivial projector on register $\dbreg\outreg^{(j)}$.
    
    Now that we are working on $K_{\yh}^{(j)}(\cA)$ for some fixed $j$, we apply the second step of reordering, fixing the oracle in $\outreg^{(-j)}$. For each $G' \in \{0,1\}^{(n_0-1)\cdot M^\ell}$, we define a $G'$-hardcoded algorithm $\cA[j,G']$, which essentially effects the same as running $\cA$ with $\outreg^{(-j)}$ initialized as $\ket{\widehat{G'}}$. We parse $G'$ as $G'=g^{(-j)} = g^{(1)} || \cdots || g^{(j-1)}|| g^{(j+1)} || \cdots || g^{(n_0)}$, viewing it as a concatenation of $n_0-1$ one-bit oracles. We further describe $\cA[j,G']$ in details:
    
    \begin{itemize}
        \item Runs on registers $\workreg\ancillareg\dbreg\outreg^{(j)}$, with $\outreg^{(j)}$ being the output register with one qubit per entry. It has query to oracle $H$, and $\Hadamardo$ query to oracle $g^{(j)}$;
        
        \item Starts with state $\ket{\phi_{\sf{st}}^{(j,G')}} = \frac{1}{\sqrt{N_M 2^{M^\ell}} } \sum_{H, g^{(j)}} \ket{\psi_{H, g^{(1)}, \cdots, g^{(n_0)}}^{(0)}}_{\workreg\ancillareg} \ket{H}_{\dbreg}  \ket{\widehat{g^{(j)}}}_{\outreg^{(j)}}$;
       
        \item For the operations, whenever $\cA$ applies a local unitary or an $H$-query, $\cA[j,G']$ does the same; whenever $\cA$ queries $G$ by $\Hadamardo$, $\cA[j,G']$ applies a $\Hadamardo$ query to $g^{(j)}$ (denoted as $\Hadamardo^{g^{(j)}}$ here) and a unitary $U_{\Hadamardo}[G']$. We define the $G'$-hardcoded unitary $U_{\Hadamardo}[G']: \ket{x,u,w}_{\workreg}\mapsto (-1)^{\langle u^{(-j)}, G'_x \rangle}\ket{x,u,w}_{\workreg}$, where $u^{(-j)}$ denotes the concatenation of the $n_0-1$ bits of $u$ other than its $j$-th bit. Notice that $U_{\Hadamardo}[G']$ is a local unitary that only acts on $\workreg$, it also extends to other registers trivially. We show that these two operations essentially simulate a $\Hadamardo$ query to $G$, where $G$ is the $n_0$-bit oracle induced by $g^{(j)}$ and $G'$, 
            \begin{equation*}
            \begin{split}
            U_{\Hadamardo}[G']\cdot \Hadamardo^{g^{(j)}} \cdot \ket{x,u,w,a,H}_{\workreg\ancillareg\dbreg}\ket{\widehat{g^{(j)}}}_{\outreg^{(j)}}
            &= (-1)^{\langle u^{(-j)}, G'_x \rangle}\cdot (-1)^{ u^{(j)}\cdot g_x^{(j)} }\cdot \ket{x,u,w,a,H}\ket{\widehat{g^{(j)}}}_{\outreg^{(j)}} \\
            &= (-1)^{\langle u, G_x \rangle} \ket{x,u,w,a,H}\ket{\widehat{g^{(j)}}}_{\outreg^{(j)}},
            \end{split}
            \end{equation*}
        where we obtain the same phase kick-back.
    
        \item Has the same amount of quantum space as $\cA$, which is the size of $\workreg$; as well as $T$ queries to $H$, and $T$ $\Hadamardo$ queries to $g^{(j)}$;
        
        \item Terminates with final state $\ket{\phi_{\sf{fin}}^{(j,G')}} = \frac{1}{\sqrt{N_M 2^{M^\ell}} } \sum_{H, g^{(j)}} \ket{\psi_{H, g^{(1)}, \cdots, g^{(n_0)}}^{\sf fin}}_{\workreg\ancillareg} \ket{H}_{\dbreg}  \ket{\widehat{g^{(j)}}}_{\outreg^{(j)}}$.
    \end{itemize}
    
    We then show that it suffices to bound $K_{\yh}(\cA[j, G'])$. Since $\cA[j,G']$ has $\outreg^{(j)}$ as its output register, then here we have $\bigk_{\yh}=I_{\workreg\ancillareg}\otimes \left(\Pi_{(j),\yh}^{K}\right)_{\outreg^{(j)}}$, therefore,
    \begin{equation}
    \label{eq:bigk-expectation}
    \begin{split}
    &\ \E_{G'} \left[ K_{\yh}(\cA[j,G']) \right]\\
    =&\  \frac{1}{2^{(n_0-1)M^\ell}} \sum_{g^{(-j)}} \left| \bigk_{\yh} \ket{\phi_{\sf{fin}}^{(j,g^{(-j)})}}_{\workreg\ancillareg\dbreg\outreg^{(j)}} \right|^2 \\
    =&\  \frac{1}{2^{(n_0-1)M^\ell}} \sum_{g^{(-j)}} \left| I_{\workreg\ancillareg} \otimes \left(\Pi_{(j),\yh}^{K}\right)_{\outreg^{(j)}} \cdot \frac{1}{\sqrt{N_M 2^{M^\ell}}} \sum_{H, g^{(j)}} \ket{\psi_{H, g^{(1)},\cdots,g^{(n_0)}}^{\sf{fin}}}_{\workreg\ancillareg} \ket{H}_{\dbreg} \ket{\widehat{g^{(j)}}}_{\outreg^{(j)}} \right|^2 \\
    =&\ \frac{1}{2^{(n_0-1)M^\ell}} \sum_{g^{(-j)}} \left| \bigk^{(j)}_{\yh} \cdot \frac{1}{\sqrt{N_M 2^{M^\ell}}} \sum_{H, g^{(j)}} \ket{\psi_{H, g^{(1)},\cdots,g^{(n_0)}}^{\sf{fin}}}\ket{H}_{\dbreg} \ket{\widehat{g^{(j)}}}_{\outreg^{(j)}} \ket{\widehat{g^{(-j)}}}_{\outreg^{(-j)}} \right|^2 \\
    =&\  \left| \sum_{g^{(-j)}} \bigk^{(j)}_{\yh} \cdot \frac{1}{\sqrt{N_M 2^{n_0 M^\ell}}} \sum_{H, g^{(j)}} \ket{\psi_{H, g^{(1)},\cdots,g^{(n_0)}}^{\sf{fin}}} \ket{H}_{\dbreg} \ket{\widehat{g^{(j)}}}_{\outreg^{(j)}} \ket{\widehat{g^{(-j)}}}_{\outreg^{(-j)}} \right|^2 \\
    =&\ \left| \bigk^{(j)}_{\yh} \ket{\phi_{\sf{fin}}} \right|^2 = K_{\yh}^{(j)}(\cA).
    \end{split}
    \end{equation}
    The third equation comes from the form of $\bigk^{(j)}_{\yh}$ acting on registers with output register $\outreg$, and that $\ket{\widehat{g^{(-j)}}}$ has norm 1; the fourth equation comes from the orthogonality of $\{\ket{\widehat{g^{(-j)}}}\}$ for different $g^{(-j)} \in \{0,1\}^{(n_0-1)\cdot M^\ell}$.
    
    Therefore, we only need to focus on bounding the \subsetcapacity{} $K_{\yh}(\cA[j, G'])$ for the $G'$-hardcoded algorithm $\cA[j,G']$, with fixed $j=1,\cdots,n_0$, and $G' \in \{0,1\}^{(n_0-1)\cdot M^\ell}$. For the simplicity of the following part of the proof, we \textbf{denote $\cA[j,G']$ as $\cB$}, \textbf{denote $g^{(j)}$ as $g$}, and \textbf{denote $\outreg^{(j)}$ as $\outreg'$}. 
    
    Then, by the constraint on $\cA$ in the theorem statement, we have: $\cB$ is a uniform algorithm with space $S$ and $T$ queries to $H$, $T$ $\Hadamardo$ queries to $g$, running on registers $\workreg\ancillareg\dbreg\outreg'$. We will be working on algorithm $\cB$ until the very end of the proof. 
    
    We then apply the \textbf{L-slicing} process onto $\cB$ as follows. View $\cB$ as a sequence of local unitaries, $H$ queries and $\Hadamardo$ queries to $g$, then slice it into $L=2T/T'$ segments $\cB_1,\cB_2,\cdots,\cB_L$ where each segment $\cB_t$ has at most $T'=c S^{-1} N^{\frac{1}{\ell+1}}$ overall queries to $H$ and $g$, $c$ is some constant we will fix at the end of analysis. Denote $\ket{\phi_0}_{\workreg\ancillareg\dbreg\outreg'}$ as the initial state before running algorithm, where $\ket{\psi_{H,g}^{(0)}}_{\workreg\ancillareg}$ is all zero for any $H,g$. The registers $\workreg\ancillareg$ are initialized to be all zero since $\cB$ is a uniform algorithm. For each $t=1,2,\cdots,L$, denote
    \begin{equation*}
        \ket{\phi_t}_{\workreg\ancillareg\dbreg\outreg'} = \frac{1}{\sqrt{N^M 2^{M^\ell}}}  \sum_{H,g} \ket{\psi_{H,g}^{(t)}}_{\workreg\ancillareg} \ket{H}_{\dbreg}\ket{\widehat{g}}_{\outreg'}
    \end{equation*}
    to be the purified state after running segments from $\cB_1,\cdots$ to $\cB_t$. 
    
    We hereby specify the access to ancilla register $\ancillareg = \bigotimes_{t=0}^L \ancillareg_t$. where each sub-register $\ancillareg_t$ is disjoint and of arbitrary size. From the perspective of $\cB_t$, the state before its operation has a purification across registers $\ancillareg_0\cdots\ancillareg_{t-1}$, but it has no access to $\ancillareg_0\cdots\ancillareg_{t-1}$ and can only apply unitaries over $\workreg$ and $\ancillareg_{t}$. In the following analysis, we may just use $\ancillareg$ in subscript for simplicity. For each induced segment $\cB_t$, it is a non-uniform algorithm with $S$-qubit advice, where the advice is the mixed state in register $\workreg$ after running $\cB_1,\cdots,\cB_{t-1}$. Specifically, adopting a view back to without purifying oracles, $H,g$ is classically sampled at the beginning, then algorithm $\cB_t$ takes advice with purification $\ket{\psi_{H,g}^{(t-1)}}_{\workreg\ancillareg}$ and apply some $H,g$-integrated local unitary $(U_{H,g}^{(t)})_{\workreg\ancillareg_{t}}$ and obtain $\ket{\psi_{H,g}^{(t)}}_{\workreg\ancillareg}$. 

    We then analyze the output register $\outreg'$ after applying each slicing-induced $\cB_t$. For any $t=1,\cdots,L$, $\cB_t$ has an $S$-qubit mixed state advice and less than $T'$ queries.\zikuan{T' queries?}
    Due to the range upper bound of $S$, we have $T' = \Omega(S^{2\ell})$.

    To characterize the constraint on such a non-uniform algorithm, we define the $\fliptest$ game in \Cref{game:flip-test}. We will later see that the average-case version of this game is equivalent to the random entry challenge game \Cref{proj:single-game}, where the lazy-sampling process of oracle $g$ corresponds to viewing $\outreg'$ as an output register. This allows us to use the bound on the success probability of the game from the last section. 
    
    The game $\fliptest$ is parameterized by fixed oracles $H,g$ and a sequence $\yh$. Let $Y_{H,\yh}\subseteq [M]^\ell $ denote the subset of indices which is an $\ell$-tuple with sum $y_H$, i.e. $Y_{H,\yh} = \{i=(i_1,\cdots,i_\ell) \mid \sum_{j=1}^\ell H(i_j) = y_H \}$. For simplicity, we will refer to it as $Y_H$ in the following context. Denote $\outreg'_{-i} = \bigotimes_{j\neq i}\outreg'_j$, and $g_{-i}$ as the $M^\ell-1$ entries oracle from $[M]^\ell\backslash \{i\}$ to $\set{0,1}$ copied from $g$.
    We maintain the notation of ancilla register $\ancillareg$, parsing it into $\ancillareg = \ancillareg_0\ancillareg_1$, where the other half of advice purification is in $\ancillareg_0$, and $\ancillareg_1$ is the unbounded extra space for $\cB$.
    \zikuan{I add the remark here, you can choose whether you want to use it.}
    \begin{remark}
        The reason that $\ancillareg_0$ exists is that the $S$-qubit advice state can be a mixed state and $\ancillareg_0$ serves as its purification, though it does not participate in any computation.
    \end{remark}
    
    \game
    {$\fliptest_{H,g}^{\yh}$:}
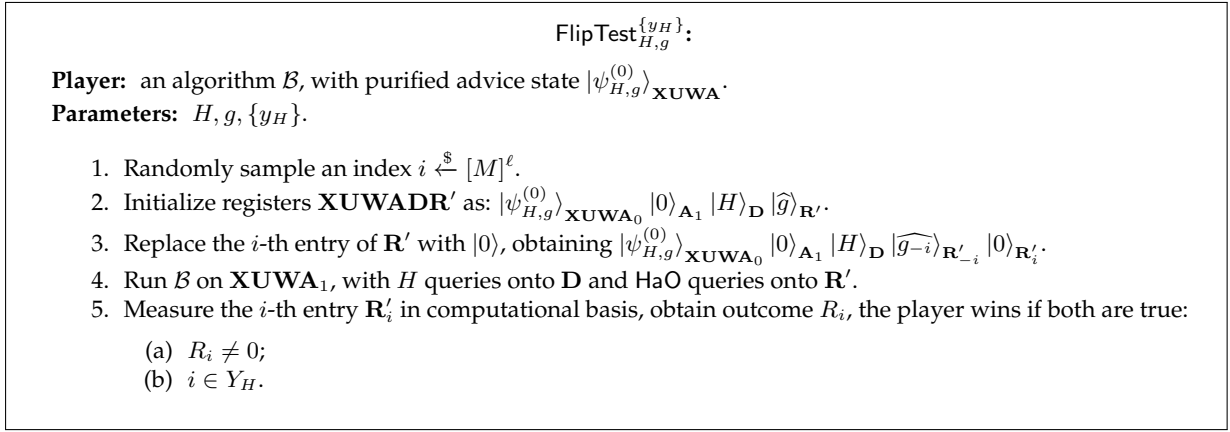
    {Challenge an algorithm for finding a random $\ell$-tuple when $H,g$ are fixed.}
    {game:flip-test}
    {
    \begin{description}
    \setlength{\parskip}{0.3mm} 
    \setlength{\itemsep}{0.3mm} 
    \item[Player:] an algorithm $\cB$, with purified advice state $\ket{\psi^{(0)}_{{H,g}}}_{\workreg\ancillareg}$.
    \item[Parameters:] $H,g,\yh$.
    \end{description}
    \begin{enumerate}
        \item Randomly sample an index $i\xleftarrow{\$} [M]^\ell$.
        \item Initialize registers $\workreg\ancillareg\dbreg\outreg'$ as: $ \ket{\psi^{(0)}_{H,g}}_{\workreg\ancillareg_0}\ket{0}_{\ancillareg_1} \ket{H}_{\dbreg} \ket{\widehat{g}}_{\outreg'}$.
        \item Replace the $i$-th entry of $\outreg'$ with $\ket{0}$, obtaining $\ket{\psi^{(0)}_{H,g}}_{\workreg\ancillareg_0}\ket{0}_{\ancillareg_1} \ket{H}_{\dbreg}\ket{\widehat{g_{-i}}}_{\outreg'_{-i}}\ket{0}_{\outreg'_i}$.
        \item Run $\cB$ on $\workreg\ancillareg_1$\zikuan{I think here it should be $\ancillareg_1$}, with $H$ queries onto $\dbreg$ and $\Hadamardo$\zikuan{Here should be $\output$, or $\Hadamardo$ when we view the Hadamard basis $\ket{\widehat{g}}_{\outreg'}$ as the standard basis.} queries onto $\outreg'$.
        \item Measure the $i$-th entry $\outreg'_i$ in computational basis, obtain outcome $R_i$, the player wins if both are true:
        \begin{enumerate}
            \item $R_i \neq 0$;
            \item $i\in Y_H$.
        \end{enumerate}
    \end{enumerate}
    }
    
    We further define game $\avgtest^{\yh}$ as the average-case version of $\fliptest$, while it runs $\fliptest$ under $\yh$ and uniformly chosen $H,g$. We now argue that $\avgtest^{\yh}$ is almost identical to the game defined in \Cref{proj:single-game} except the algorithm is equipped with $\Hadamardo$ instead of $\output$.
    
    \begin{claim}
    \label{claim:game-equivalence}
        For arbitrary $\yh$, any non-uniform algorithm $\cB$ with oracle queries to $H$ and $\Hadamardo$ queries to $g$, denote $\cB'$ as the algorithm obtained by only substituting all $\output$ operations to $\Hadamardo$ in $\cB$. define $(\rho_{\sf st})_{\allreg}$ to be the mixed state being initialized, with another half purification in $\ancillareg_0$, then we have 
        \begin{equation*}
            \Pr\left[\text{$\cB$ wins }\avgtest^{\yh}\right] =  \E_{H,g} \Pr\left[\text{$\cB$ wins }\fliptest_{H,g}^{\yh}\right] = \abs{\Pi_{{\sf FlipTest}}^{\cB',\set{y_H}}\brackets{\rho_{\sf st}\ket{0}_{\ancillareg_1\seedreg\memreg}\bra{0}}\brackets{\Pi_{{\sf FlipTest}}^{\cB',\set{y_H}}}^\dagger},
        \end{equation*}
        \zikuan{$\rho_{\sf st}$ should be tensor product with empty state on $\ancillareg_1\otimes\seedreg\otimes\memreg$?}
        where in the middle term we are taking expectation over uniformly sampled random oracles $H,g$.
    \end{claim}
    
    \begin{proof}
        We first look at the projector $\Pi_{{\sf FlipTest}}^{\cB',\set{y_H}}$ defined in \Cref{proj:single-game}. Notice that register $\seedreg$ is initialized to be the superposition over $[M]^\ell$, then controls over the choice of index in $\outreg'$, then measured in standard basis. This is equivalent to classically sample $i\xleftarrow{\$}[M]^\ell$. 
        
        Then, with random index $i$ fixed, the unitary $\sf Store$ applies a generalized $\cnot$ (we just use $\cnot$ with a slight abuse of notation) acting on two single qubit \zikuan{$2$-dimension or $\log 2$ sized I think?}registers $\outreg'_i$ and $\memreg$ before and after running the algorithm $\cB'$, where $\cnot_{\outreg'_i\memreg}$ denotes $\outreg'_i$ controlling $\memreg$. Notice that $\swap_{\outreg'_i\memreg} = \cnot_{\memreg\outreg'_i}\cnot_{\outreg'_i\memreg}\cnot_{\memreg\outreg'_i}$, also that algorithm $\cB$ only interacts with register $\outreg'_i$ by $\output$ operations, which stabilizes states in the Hadamard basis, therefore it commutes with $\cnot_{\memreg\outreg'_i}$. Also notice that the initial state on $\memreg$ is $\ket{0}$, and we measure $\memreg$ in computational basis, which also commutes with $\cnot_{\memreg\outreg'_i}$. Thus, with fixed random index $i$, replacing the two $\cnot_{\outreg'_i\memreg}$ with two $\swap_{\outreg'_i\memreg}$ does not affects the game. In this way, it is equivalent to the procedure in \Cref{game:flip-test} that swaps a $\ket{0}$ onto $\outreg'_i$ and then measure this entry.
    
        Then, consider the initialization of registers. In \Cref{proj:single-game}, $\rho_{\sf st}$ is the reduced state of some $\frac{1}{\sqrt{N^M2^{M^\ell}}}\sum_{H,\widehat{r}}\ket{\psi_{H,\widehat{r}}}_{\qryreg\ansreg\auxreg\ancillareg}\ket{\widehat{r}}_{\outreg'}\ket{H}_{\dbreg}$, with $r,H$ in uniform superposition. Since $\ket{\widehat{r}}\ket{H}$ is orthogonal for different $\widehat{r},H$, it is equivalent with averaging over the same game on all possible $\widehat{r},H$, which corresponds to $\avgtest^{\yh}$ as $\widehat{r},\widehat{g}$ are of same function in the process. 
    
        Therefore, we show the equivalence of $\cB'$ projected onto $\Pi_{{\sf FlipTest}}^{\cB',\set{y_H}}$ and $\cB$ winning $\avgtest^{\yh}$ with corresponding advice state.
        
    \end{proof}

    Then, by \Cref{cor:non-uniform-random-challenge}, any $\cB_t$ has the following property:
    \begin{equation*}
        \forall\yh,\ \ p^{(t)}_{\yh} := \Pr\left[\cB_t \text{ wins }\avgtest \right] \leq p_e = O\brackets{\frac{S^{2\ell}}{M^\ell}} + O\brackets{\frac{(ST')^{\ell+1}}{N}}^{S},
    \end{equation*}
    where $p_e$ is obtained by taking a proper constant in the big-$O$ notation. We denote such a property as \textbf{\boundedflip}, namely winning the $\avgtest^{\yh}$ with at most probability $p_e$ for any $\yh$. We then glue the property on each $\cB_t$ together with the following lemma:
    
    \begin{lemma}[Quantum union bound on random sample game]
    \label{lem:quantum-boolean-function-union-bound}
        For any sequence $\yh$, let $\cB$ be an arbitrary algorithm with $H$ queries and $\Hadamardo$ queries to one-bit oracle $g$, let $(\cB_1,\cdots,\cB_L)$ be an L-slicing of $\cB$. Assume that $\forall t=1,\cdots,L$, the non-uniform algorithm $\cB_t$ is \boundedflip{}, then the \subsetcapacity{} of $\cB$ is bounded:
        \begin{equation*}
            K_{\yh}(\cB) \leq L^2 \cdot M^\ell p_e .
        \end{equation*}
    \end{lemma}

    \begin{proof}
        For $L$-slicing induced non-uniform algorithms $\cB_1,\cdots,\cB_L$, denote $p^{(t)}_{\yh}$ as the probability of $\cB_t$ winning the game $\avgtest^{\yh}$. Recall that we have specified the behavior of $\cB_t$ under fixed $H,g$ as follows:
        \begin{itemize}
            \item takes advice $\brackets{\rho_{H,g}^{(t-1)}}_{\workreg}$, which has a purification $\ket{\psi_{H,g}^{(t-1)}}_{\workreg\ancillareg_0\ancillareg_1\cdots\ancillareg_{t-1}}$;
            \item clarifies that it can only access registers $\workreg\ancillareg_t$;
            \item applies $H,g$-integrated unitary $\brackets{U_{H,g}^{(t)}}_{\workreg\ancillareg}$;
            \item obtains $\ket{\psi_{H,g}^{(t)}}_{\workreg\ancillareg} = U_{H,g}^{(t)} \ket{\psi_{H,g}^{(t-1)}}_{\workreg\ancillareg} $, with registers $\ancillareg_{t+1}\cdots \ancillareg_{L}$ still being all zeros.
        \end{itemize}
        We then look at the algorithm $\cB$ in a whole. Under any $H,g$, the registers $\workreg\ancillareg$ is initialized as $\ket{\psi_{H,g}^{(0)}}_{\workreg\ancillareg} = \ket{0}_{\workreg\ancillareg}$, and then the algorithm applies unitary $U_{H,g} = U_{H,g}^{(L)}\cdots U_{H,g}^{(2)}U_{H,g}^{(1)} $ over registers $\workreg\ancillareg$, and obtain $\ket{\psi_{H,g}^{(L)}}_{\workreg\ancillareg} = U_{H,g}\ket{0}_{\workreg\ancillareg}$. Denote  $\ket{\phi_{L}}_{\workreg\ancillareg\dbreg\outreg'}$ as the final state after running algorithm $\cB$. Then the \subsetcapacity{} of the final state is $\bra{\phi_L} \bigk_{\yh}\ket{\phi_L}$.

        We first represent $p^{(t)}_{\yh}$ for any $1\leq t\leq L$. For any $g$ and $s\in\set{0,1}$, define $g^{i,s}$ to be the oracle obtained by replacing the $i$-th entry of $g$ by $s$ while others remain. We consider the process of $\cB_t$ playing $\fliptest_{H,g}^{\yh}$ for each pair of oracles $H,g$. On $i\xleftarrow{\$}[M]^\ell$ being sampled as the random index in step $1$ of \Cref{game:flip-test}, the entire state in $\workreg\ancillareg\dbreg\outreg'$ evolves as follows:
        \begin{equation*}
        \begin{split}
            & \ket{\psi_{H,g}^{(t-1)}}_{\workreg\ancillareg}\ket{H}_{\dbreg}\ket{\widehat{g}}_{\outreg'}\\
            \xrightarrow{\text{replace }\outreg'_i} & \ket{\psi_{H,g}^{(t-1)}}_{\workreg\ancillareg}\ket{H}_{\dbreg}\ket{\widehat{g_{-i}}}_{\outreg'_{-i}}\ket{0}_{\outreg'_i}\\
            =& \frac{1}{\sqrt{2}}\sum_{s\in\set{0,1}} \ket{\psi_{H,g}^{(t-1)}}_{\workreg\ancillareg}\ket{H}_{\dbreg}\ket{\widehat{g_{-i}}}_{\outreg'_{-i}}\ket{{\widehat{s}}}_{\outreg'_i}\\
            \xrightarrow{\text{runs }\cB_t} & \frac{1}{\sqrt{2}}\sum_{s\in\set{0,1}} U_{H,g^{i,s}}^{(t)} \ket{\psi_{H,g}^{(t-1)}}_{\workreg\ancillareg} \ket{H}_{\dbreg} \ket{\widehat{g_{-i}}}_{\outreg'_{-i}}\ket{\widehat{s}}_{\outreg'_i}\\
            =& \frac{1}{\sqrt{2}}\sum_{s\in\set{0,1}} U_{H,g^{i,s}}^{(t)} \ket{\psi_{H,g}^{(t-1)}}_{\workreg\ancillareg} \ket{H}_{\dbreg} \ket{\widehat{g_{-i}}}_{\outreg'_{-i}} \frac{1}{\sqrt{2}} \sum_{r\in\set{0,1}}(-1)^{s\cdot r}\ket{r}_{\outreg'_i}  \\
            =& \sum_{r\in\set{0,1}}\left(\frac{1}{2}\sum_{s\in\set{0,1}}(-1)^{s\cdot r} U_{H,g^{i,s}}^{(t)} \ket{\psi_{H,g}^{(t-1)}}_{\workreg\ancillareg}\right) \ket{H}_{\dbreg} \ket{\widehat{g_{-i}}}_{\outreg'_{-i}} \ket{r}_{\outreg'_i},
        \end{split}
        \end{equation*}
        then when measuring $\outreg'_i$, suppose we get $r$ with probability $\Pr\bigl[r_i=r \mid i,\cB_t,H,g \bigr]$. Then, the overall winning probability is
        \begin{equation*}
        \begin{split}
            p_{H,g,\yh}^{(t)} =& \Pr\Bigl[i\in Y_H\land r_i=1\Big|\  i,\cB_t,H,g \Bigr]\\
                              =& \frac{1}{4M^\ell}\sum_{i\in Y_H} \Big|(U_{H,g^{i,0}}^{(t)} - U_{H,g^{i,1}}^{(t)}\big) \ket{\psi_{H,g}^{(t-1)}}_{\workreg\ancillareg} \Big|^2.
        \end{split}
        \end{equation*}

        We could then represent the \boundedflip{} property of $\cB_t$ by
        \begin{equation}
        \label{eq:quantum-union-bound-prob-formula}
            p_{\yh}^{(t)} = \frac{1}{4M^\ell} \E_{H,g}  \left[ \sum_{i\in Y_H} \Big|(U_{H,g^{i,0}}^{(t)} - U_{H,g^{i,1}}^{(t)}\big) \ket{\psi_{H,g}^{(t-1)}}_{\workreg\ancillareg} \Big|^2\right] \leq p_e.
        \end{equation}

        Then, we represent $\bigk_{\yh}$ in a similar way. For any $H$ and $i\in[M]^\ell$, define projector $\Pi_{i,H} = I_{\workreg\ancillareg}\otimes \ketbra{H}{H}_{\dbreg}\otimes I_{\outreg'_{-i}}\otimes \ketbra{1}{1}_{\outreg'_i}$, then with we notice that $\bigk_{\yh}$ could be written as a sum of projectors:
        \begin{equation*}
            \bigk_{\yh} = I_{\qryreg\ansreg\auxreg\ancillareg}\otimes \sum_{H,R} \capacity_{Y_H}(R)\cdot \ketbra{H}{H}_{\dbreg}\otimes\ketbra{R}{R}_{\outreg'} = \sum_{H}\sum_{i\in Y_H} \Pi_{i,H},
        \end{equation*}
        where each $Y_H$ relative to $H$ is induced by the fixed $\yh$. We then calculate the expected value of $\Pi_{i,H}$ on the final state $\ket{\phi_L}_{\workreg\ancillareg\dbreg\outreg'}$ for any $H$ and $i\in[M]^\ell$.
        \begin{equation*}
        \begin{split}
            \bra{\phi_L} \Pi_{i,H}\ket{\phi_L}
            =&\ \bra{\phi_L} \left( I_{\workreg\ancillareg\outreg'_{-i}}\otimes \ketbra{H}{H}_{\dbreg} \otimes \ketbra{1}{1}_{\outreg'_i}\right) \ket{\phi_L}  \\ 
            =&\ \left| \frac{1}{\sqrt{N^M 2^{M^\ell}}}  \sum_{g} \ket{\psi_{H,g}^{(L)}}_{\workreg\ancillareg} \ket{H}_{\dbreg} \ket{\widehat{g_{-i}}}_{\outreg'_{-i}} \ketbra{1}{1}_{\outreg'_i}\ket{\widehat{g(i)}}_{\outreg'_{i}}  \right|^2\\
            =&\ \frac{1}{N^M 2^{M^\ell}} \left| \sum_{g_{-i}} \frac{1}{\sqrt{2}} \Big(\ket{\psi_{H,g^{i,0}}^{(L)}} - \ket{\psi_{H,g^{i,1}}^{(L)}}\Big)_{\workreg\ancillareg} \ket{H}_{\dbreg} \ket{\widehat{g_{-i}}}_{\outreg'_{-i}} \ket{1}_{\outreg'_{i}}  \right|^2\\
            =&\ \frac{1}{2 N^M 2^{M^\ell}} \sum_{g_{-i}} \left| \Big(\ket{\psi_{H,g^{i,0}}^{(L)}} - \ket{\psi_{H,g^{i,1}}^{(L)}}\Big)_{\workreg\ancillareg} \ket{H}_{\dbreg} \ket{\widehat{g_{-i}}}_{\outreg'_{-i}} \ket{1}_{\outreg'_{i}}  \right|^2\\
            =&\ \frac{1}{2 N^M 2^{M^\ell}} \sum_{g_{-i}} \left|  \ket{\psi_{H,g^{i,0}}^{(L)}}_{\workreg\ancillareg} - \ket{\psi_{H,g^{i,1}}^{(L)}}_{\workreg\ancillareg} \right|^2\\
            =&\ \frac{1}{4 N^M 2^{M^\ell}} \sum_{g} \left|  \ket{\psi_{H,g^{i,0}}^{(L)}} - \ket{\psi_{H,g^{i,1}}^{(L)}} \right|^2\\
            =&\ \frac{1}{4N^M} \E_{g} \Big| \ket{\psi_{H,g^{i,0}}^{(L)}} - \ket{\psi_{H,g^{i,1}}^{(L)}} \Big|^2.
        \end{split}
        \end{equation*}
        The fourth equation is by the orthogonality of states with different $g_{-i}$ in $\ket{g_{-i}}_{\outreg'_{-i}}$, and the fifth equation is by ignoring the norm-$1$ subsystem states in registers $\dbreg,\outreg'_{-i},\outreg'$, while we omit the subscript $\workreg\ancillareg$ afterwards. Therefore, adding the expected value of these projectors together, we have 
        \begin{equation*}
            K_{\yh}(\cB) = \bra{\phi_L} \bigk_{\yh}\ket{\phi_L} = \frac{1}{4}\E_{H,g} \left[ \sum_{i\in Y_H} \Big| \ket{\psi_{H,g^{i,0}}^{(L)}} - \ket{\psi_{H,g^{i,1}}^{(L)}} \Big|^2 \right].
        \end{equation*}

        Then, we bridge the gap from $p^{(t)}_{\yh}$ to $K_{\yh}(\cB)$. For any fixed $H,g$, we bound the $L2$-distance of the final states $\ket{\psi^{(L)}_{H,g^{i,0}}}$ and $\ket{\psi^{(L)}_{H,g^{i,1}}}$ by the distance independently imposed by each segment $\cB_t$. Here, for $t_1\leq t_2$, denote $U_{H,g}^{t_1\to t_2} := U_{H,g}^{t_2}U_{H,g}^{t_2-1}\cdots U_{H,g}^{(t_1)}$; for $t_1>t_2$, denote $U_{H,g}^{t_1\to t_2} := I$. We finally obtain
        \begin{equation*}
            \begin{split}
            &K_{\yh}(\cB)\\
            =& \frac{1}{4} \E_{H,g} \left[ \sum_{i\in Y_H} \Big| \ket{\psi_{H,g^{i,0}}^{(L)}} - \ket{\psi_{H,g^{i,1}}^{(L)}} \Big|^2 \right]\\
            =& \frac{1}{8}  \E_{H,g} \left[ \sum_{i\in Y_H} \left| \left( U_{H,g^{i,0}}^{(1\to L)} - U_{H,g^{i,1}}^{(1\to L)} \right) \ket{0} \right|^2+ \sum_{i\in Y_H} \left| \left( U_{H,g^{i,1}}^{(1\to L)} - U_{H,g^{i,0}}^{(1\to L)} \right) \ket{0} \right|^2 \right]\\
            =& \frac{1}{8}  \E_{H,g} \Bigg[ \sum_{i\in Y_H} \left| \sum_{t=1}^L \left( U_{H,g^{i,1}}^{(t+1\to L)} U_{H,g^{i,0}}^{(1\to t)} - U_{H,g^{i,1}}^{(t\to L)} U_{H,g^{i,0}}^{(1\to t-1)} \right) \ket{0} \right|^2 \\
            &+ \sum_{i\in Y_H} \left| \sum_{t=1}^L \left( U_{H,g^{i,0}}^{(t+1\to L)} U_{H,g^{i,1}}^{(1\to t)} - U_{H,g^{i,0}}^{(t\to L)} U_{H,g^{i,1}}^{(1\to t-1)} \right) \ket{0} \right|^2 \Bigg]\\
            \leq& \frac{1}{8}  \E_{H,g} \Bigg[ \sum_{i\in Y_H} L\sum_{t=1}^L \left| U_{H,g^{i,1}}^{(t+1\to L)}\left(  U_{H,g^{i,1}}^{(t)} - U_{H,g^{i,0}}^{(t)}  \right) U_{H,g^{i,0}}^{(1\to t-1)}\ket{0} \right|^2 \\
            &+ \sum_{i\in Y_H} L\sum_{t=1}^L \left| U_{H,g^{i,0}}^{(t+1\to L)}\left(  U_{H,g^{i,0}}^{(t)} - U_{H,g^{i,1}}^{(t)}  \right) U_{H,g^{i,1}}^{(1\to t-1)}\ket{0} \right|^2\Bigg]\\
            =& \frac{L}{8}  \E_{H,g} \left[ \sum_{i\in Y_H} \sum_{t=1}^L \left| \left(  U_{H,g^{i,1}}^{(t)} - U_{H,g^{i,0}}^{(t)}  \right) \ket{\psi_{H,g^{i,0}}^{(t-1)}} \right|^2 + \sum_{i\in Y_H} \sum_{t=1}^L \left| \left(  U_{H,g^{i,0}}^{(t)} - U_{H,g^{i,1}}^{(t)}  \right) \ket{\psi_{H,g^{i,1}}^{(t-1)}} \right|^2\right]\\
            =& \frac{L}{4}  \E_{H,g} \left[ \sum_{i\in Y_H} \sum_{t=1}^L \left| \left(  U_{H,g^{i,1}}^{(t)} - U_{H,g^{i,0}}^{(t)}  \right) \ket{\psi_{H,g}^{(t-1)}} \right|^2 \right]\\
            =& L \sum_{t=1}^L M^\ell p^{(t)}_{\yh} \leq L^2 \cdot M^\ell p_e.
            \end{split}
        \end{equation*}
        where the first inequality is by Cauchy-Schwarz inequality; the fourth equation is by that unitaries are norm-preserving. The last equality and the last inequality are by \Cref{eq:quantum-union-bound-prob-formula}.


        
    \end{proof}

    By \Cref{lem:quantum-boolean-function-union-bound} and the bound on $p_e$, we have $K_{\yh}(\cA[j,G']) = K_{\yh}(\cB)\leq L^2 \cdot M^\ell p_e$, for any $j\in$. Therefore, by \Cref{eq:bigk-expectation}, we also have $K^{(j)}_{\yh}(\cA) = \E_{G'} K_{\yh}(\cA[j,G'])\leq L^2 \cdot M^\ell p_e $. We now switch back to our primary target $K_{\yh}(\cA)$. By carefully picking constant $c$, we have
    \begin{equation*}
        \begin{split}
            K_{\yh}(\cA) 
            \leq&\sum_{j=1}^{n_0} K_{\yh}^{(j)}(\cA)\\
            \leq& n_0\cdot L^2 M^{\ell}\cdot p_e \\
            \leq& n_0\cdot O\brackets{\frac{T^2}{T'^2} M^\ell \left(O\brackets{\frac{S^{2\ell}}{M^\ell}} + O\brackets{\frac{(ST')^{\ell+1}}{N}}^{S} \right) }\\
            =& \widetilde{O}\left(\frac{T^2}{(cS^{-1}N^{\frac{1}{\ell+1}})^2} M^\ell \left(\frac{S^{2\ell}}{M^\ell} + (c_3\cdot c^{\ell+1})^{S} \right) \right)\\
            =& \widetilde{O}\left(S^{2\ell+2}T^2N^{-\frac{2}{\ell+1}}\right).
        \end{split}
    \end{equation*}
    For the first equation, we choose $c_3$ larger than the constant in all big-$O$ notations in \Cref{thm:non-uniform-random-challenge}, then absorb all terms into a single big-$\widetilde{O}$ notation where we omit the logarithmic term $n_0$. The second equation comes from the choice of $c$ when we set $T'=cS^{-1}N^{\frac{1}{\ell+1}}$. Since $M=\poly(N)$, $S = \Omega(\log N)$, we choose $c_1,c_2>0$ such that $M\leq N^{c_1}$, $S\geq c_2\log N$. Then, we pick $c = 2^{-c_1/c_2} c_3^{-1/(\ell+1)} $, which guarantee that $(c_3\cdot c^{\ell+1})^{S}\leq \frac{1}{M^\ell}$.
    

\end{proof}

\section{\texorpdfstring{The time-space bound for solving $\ell$-Nested Collision Finding}{The time-space bound for solving l-Nested Collision Finding}}




\subsection{The classical bound}
\begin{theorem}\label{thm:classical-overall_lowerbound}
    For $\ell\geq 2$, any classical algorithm that solves the $\ell$-Nested Collision Finding problem with constant probability with space $S$  and $T$ oracle queries must satisfy $S^{\frac{\ell^2-1}{2\ell}}T=\Omega\brackets{N^{\frac{1}{2\ell}}N_0^{\frac{1}{2}}}$ if $S=\Omega(\log N)$ and $S=O\brackets{N^{\frac{1}{\ell^2-1}}}$.
\end{theorem}
\begin{proof}
   Let $K$ to be the expectation of the number of $\ell$-tuples with the same sum one can find in all its queried tuples on $G$ after the algorithm finishes. Then it must satisfy $TK=\Omega(N_0)$. Thus constraints are:
    \begin{equation*}
        \left\{
        \begin{aligned}
            K&=O\brackets{\frac{S^{\frac{\ell^2-1}{\ell}}T}{N^{\frac{1}{\ell}}}}\\
            TK&=\Omega(N_0)
        \end{aligned}
        \right.
    \end{equation*}
    Combine two equations we have $S^{\frac{\ell^2-1}{2\ell}}T=\Omega\brackets{N^{\frac{1}{2\ell}}N_0^{\frac{1}{2}}}$.

\end{proof}
\begin{theorem}\label{classical-space-time-tradeoff}
    For any $\ell\geq 2$ and any $N_0=\Theta\brackets{N^\epsilon}$ where $\frac{1}{\ell}<\epsilon\leq\frac{2}{\ell-1}$ the corresponding $\ell$-Nested Collision Finding Problem has a space-time tradeoff for classical algorithms.
\end{theorem}
\begin{proof}
    By \Cref{thm:classical-overall_lowerbound}, any classical algorithm that solves the $\ell$-Nested Collision Finding Problem when $\ell\geq 2$ and $N_0=\Theta\brackets{N^\epsilon}$ with constant probability using $S$ bit space and $T$ oracle queries must satisfy $S^{\frac{\ell^2-1}{2\ell}}T=\Omega\brackets{N^{\frac{1}{2\ell}+\frac{\epsilon}{2}}}$. This implies that with limited space, say $S=O(\log N)$, any classical algorithm that succeeds with constant probability should use $T=\Omega\brackets{N^{\frac{1}{2\ell}+\frac{\epsilon}{2}}}$ oracle queries. But with enough space, the classical algorithm described in \Cref{thm:classical-algorithm} solves this problem when $N_0=O\brackets{N^{\frac{2}{\ell-1}}}$ with constant probability using $T=O\brackets{N^{\frac{1}{\ell}+\frac{\eps}{2\ell}}}$ oracle queries which is strictly better than the bound when $\ell\geq 2$.
\end{proof}
\subsection{The quantum bound}
\begin{theorem}\label{thm:quantum-overall-lowerbound}
    For $\ell\geq 2$, any quantum algorithm that solves the $\ell$-Nested Collision Finding problem with constant probability with space $S$ and $T$ oracle queries must satisfy $S^{\frac{\ell+1}{2}}T=\Omega\brackets{N^{\frac{1}{2(\ell+1)}}N_0^{\frac{1}{4}}}$ if $S=\Omega(\log N)$ and $S=O\brackets{N^{\frac{1}{(\ell+1)(2\ell+1)}}}$.

    
\end{theorem}
\begin{proof}
    By the definition of $\ell$-Nested Collision Finding Problem, the algorithm is given a target $y$ sampled from an arbitrary distribution $\cD$. We will prove that for any fixed $y$, the algorithm solving problem will have $S^{\frac{\ell+1}{2}}T=\widetilde\Omega\brackets{N^{\frac{1}{2(\ell+1)}}N_0^{\frac{1}{4}}}$. Then, for any distribution $\cD$, if an algorithm $\cA$ solves the $\ell$-Nested Collision Finding Problem with respect to $\cD$, take $y^*$ from the support of $\cD$ such that $\cA$ has largest winning probability with $y^*$ as target, then $\cA$ also solves the $\ell$-Nested Collision Finding Problem for $y^*$, and therefore also satisfy the bound $S^{\frac{\ell+1}{2}}T=\widetilde\Omega\brackets{N^{\frac{1}{2(\ell+1)}}N_0^{\frac{1}{4}}}$. 

    We now consider the case when the target $y$ is fixed, we fix $\yh$ such that $y_H=y$ for all $H$. We have shown that the algorithms having each of $\cpho$, $\Hadamardo$ and $\csto$ as the oracle query form are equivalent, therefore without loss of generality, we consider an arbitrary algorithm $\cA$ interacting with random oracle $G$ in $\Hadamardo$, and all registers $\workreg\ancillareg\dbreg\outreg$ are therefore initialized as $\ket{\Phi_0}=\frac{1}{\sqrt{N^M2^{M^\ell}}}\sum_{H,G}\ket{0}_{\workreg\ancillareg}\ket{\widehat{G}}_{\outreg}\ket{H}_{\dbreg}$.

    For $t=1,\cdots,T$, we look at the state in all registers after $t$ oracle queries, let $K^{(t)}$ be the expected \subsetcapacity{} of state $\ket{\Phi_t}$ with respect to our fixed $\yh$. Let $K=\frac{1}{T}\sum_{t=1}^T K^{(t)}$. By \Cref{thm:two-oracle-K-l-tuple-time-space-tradeoff}, algorithm up to time step $t$ has $t\leq T$ queries, we have
    \begin{equation}
        \label{eq:final-K-UB}
         K = \frac{1}{T}\sum_{t=1}^T\widetilde{O}\brackets{S^{2\ell+2}T^2 N^{-\frac{2}{\ell+1}}} = \widetilde{O}\brackets{S^{2\ell+2}T^2 N^{-\frac{2}{\ell+1}}}.
    \end{equation}
    We can also define the value of \subsetcapacity{} for any fixed $H$. Since we could view $H$ as classically uniform-randomly sampled, then whenever we look at the state in $\workreg\ancillareg\dbreg\outreg$ and project it to subspaces where $\dbreg$ is $\ket{H}_{\dbreg}$, their corresponding projected component all have the same amplitude. Also, if we denote $\ket{\Phi_{i,H}}_{\workreg\ancillareg\dbreg\outreg}$ as the normalized state after projecting it to the subspace where $\dbreg$ is $\ket{H}_{\dbreg}$, $\ket{\Phi_i} = \frac{1}{\sqrt{N^M}}\sum_{H} \ket{\Phi_{i,H}} $. Therefore, define $K^{(t)}_H$ as the expected \subsetcapacity{} of $\ket{\Phi_{i,H}}$, and $K_H = \frac{1}{T}\sum_{t=1}^T K_H^{(t)} $, we have $K = \E_{H} \left[K_H \right]$.

    Now we look at the probability of finding collisions in $G$ with bounded expected \maxcapacity{} by applying \Cref{remember-tuples}. Notice that in the original setting of \Cref{remember-tuples}, $G$ oracle is purified in register $\dbreg$, initialized as $\ket{\bot,\cdots,\bot}_{\dbreg} $, where the algorithm queries it using $\cpho$, $\outreg$ is used to output collisions, and there is no oracle $H$; for fixed label function, $V$ is defined by looking at register $\dbreg$ in compressed basis and count its non-$\bot$ entries among the designated label. 
    
    Here, we argue that this is equivalent to our current scenario under each fixed $H\in[N]^M$. Now output the collisions in some pre-determined positions of $\workreg$. We set the domain of $G$ oracle to be $[M]^\ell$, $G$ is purified in $\outreg$, initialized as $\ket{0,\cdots,0}_{\outreg}$, where the algorithm queries it using $\Hadamardo$. Notice that $\Hadamardo$ is essentially a $\cpho$ under Hadamard basis of register $\outreg$, the current standard basis state $\ket{0}$ is $\ket{\widehat{0}}=\ket{\bot}$ under Hadamard basis, so it is also consistent. For our fixed $H$, we define its induced label function $\lab_{H}:[M]^\ell\to [N]$, where $\lab_{H}(x_1,\cdots,x_\ell) = \sum_{j}H(x_j)\mod N$, therefore $\ell$-tuples with label $y_H$ are $\ell$-tuples with sum $y_H$. With fixed $H$, the state after $t$-th query becomes $\ket{\Phi_{i,H}}$, and $K_H$ exactly characterizes the expected same-label capacity averaging over all time steps from $1$ to $T$. 
    
    Therefore, we could safely borrow \Cref{remember-tuples}. For any fixed $H$ and label function $\lab_H$ as defined above, the probability of $\cA$ outputting a pair of labeled-collision in $G$ is exactly the probability of $\cA$ outputting a pair of colliding $\ell$-tuples in $G$ of target sum $y$, we denote this probability as $P_{H}$, then 
    \begin{equation}
        P_H = O\brackets{\frac{T^2 K_H}{N_0}}.
    \end{equation}

    Finally, the algorithm that solves the $\ell$-Nested Collision Finding problem (defined in \Cref{def:nested_col}) finds a pair of colliding $\ell$-tuples in $G$ of the target sum $y$ with randomly sampled $H$, therefore $\E_{H}\left[P_H\right] = \Omega(1)$. Then we have
    \begin{equation}
        \label{eq:final-K-LB}
        K = \E_{H} \left[K_H \right] = \Omega\brackets{T^{-2}N_0 }\cdot\E_{H}\left[P_H\right] = \Omega\brackets{T^{-2}N_0 }.
    \end{equation}
    Since the constant in the $O$-notation in \Cref{remember-tuples} does not depend on $H$, we can take this expectation over $H$. Combine \Cref{eq:final-K-LB} and \Cref{eq:final-K-UB}, we have $S^{\frac{\ell+1}{2}}T=\widetilde\Omega\brackets{N^{\frac{1}{2(\ell+1)}}N_0^{\frac{1}{4}}}$.
\end{proof}
\begin{theorem}\label{thm:quantum-space-time-tradeoff}
    For any $\ell\geq 4$ and any $N_0=\Theta\brackets{N^\epsilon}$ where $\frac{4\ell+6}{2\ell^2-\ell-3}<\epsilon\leq\frac{3}{\ell-1}$ the corresponding $\ell$-Nested Collision Finding Problem has a space-time tradeoff for quantum algorithms.
\end{theorem}
\begin{proof}
    By \Cref{thm:quantum-overall-lowerbound}, any quantum algorithm that solves the $\ell$-Nested Collision Finding Problem when $\ell\geq 2$ and $N_0=\Theta\brackets{N^\epsilon}$ with constant probability using $S$ qubit space and $T$ oracle queries must satisfy $S^{\frac{\ell+1}{2}}T=\widetilde\Omega\brackets{N^{\frac{1}{2(\ell+1)}+\frac{\eps}{4}}}$. This implies that with limited space, say $S=O(\log N)$, any quantum algorithm that succeeds with constant probability should use $T=\widetilde{\Omega}\brackets{N^{\frac{1}{2(\ell+1)}+\frac{\eps}{4}}}$ oracle queries. But with enough space, the quantum algorithm described in \Cref{thm:quantum-algorithm} solves this problem when $N_0=O\brackets{N^{\frac{3}{\ell-1}}}$ with constant probability using $T=\Theta\brackets{N^{\frac{2+\eps}{2\ell+1}}}$ oracle queries which is strictly better than the bound when $\epsilon>\frac{4\ell+6}{2\ell^2-\ell-3}$. But notice that only when $\ell\geq 4$ there exists $\epsilon$ such that $\frac{4\ell+6}{2\ell^2-\ell-3}<\epsilon\leq\frac{3}{\ell-1}$.
\end{proof}

\printbibliography
\appendix

\section{\texorpdfstring{Proof for \texorpdfstring{\Cref{lem:single_game_bound_1}}{Lemma~\ref{lem:single_game_bound_1}}, \texorpdfstring{\Cref{lem:single_game_bound_2}}{Lemma~\ref{lem:single_game_bound_2}} and \texorpdfstring{\Cref{lem:alternative_game_monotone}}{Lemma~\ref{lem:alternative_game_monotone}}}{Proof for lemmas regarding the alternating game}}\label{sec:appendix_lemmas_alternating_game}
\subsection{Proof for \texorpdfstring{\Cref{lem:single_game_bound_1}}{Lemma~\ref{lem:single_game_bound_1}}}
\begin{proof}
    Let $\set{(\ket{u_i},\ket{v_i}}_{i\in[D]}$ be the Jordan basis of projectors $\Pi_{{\sf FlipTest}}^{\cal{A},\set{y_H}}$ and $\Pi_{\sf reset}$. We can rewrite the initial state as:
    \begin{equation*}
        \ket{\psi_{\sf st}}\ket{0}_{\seedreg}\ket{0}_{\memreg}=\sum_{i\in[D]}c_i\ket{v_i}.
    \end{equation*}
    Then the final state is
    \begin{equation*}
        \brackets{\Pi_{{\sf FlipTest}}^{\cal{A},\set{y_H}}\Pi_{\sf reset}}^k\brackets{\ket{\psi_{\sf st}}\ket{0}_{\seedreg}\ket{0}_{\memreg}}=\sum_{i\in[D]}c_ip_i^{\frac{2k-1}{2}}\ket{u_i}
    \end{equation*}
    where $p_i=\abs{\braket{u_i}{v_i}}^2$. By Jensen's inequality \Cref{lem:jensen_inequal}, we have
    \begin{align*}
        P^{(k)}_{\sf FlipTest}=&\abs{\brackets{\Pi_{{\sf FlipTest}}^{\cal{A},\set{y_H}}\Pi_{\sf reset}}^k\brackets{\ket{\psi_{\sf st}}\ket{0}_{\seedreg}\ket{0}_{\memreg}}}^2\\
        =&\sum_{i\in[D]}c_i^2p_i^{2k-1}\\
        \geq&\brackets{\sum_{i\in[D]}c_i^2p_i}^{2k-1}\\
        =&\brackets{P^{(1)}_{\sf FlipTest}}^{2k-1}.
    \end{align*}
\end{proof}
\subsection{Proof for \texorpdfstring{\Cref{lem:single_game_bound_2}}{Lemma~\ref{lem:single_game_bound_2}}}
\begin{proof}
    Replacement of the advice state to a maximally mixed state can only affect the success probability by $\frac{1}{2^S}$ since the advice is an $S$-qubit state. 
    \begin{align*}
        &\brackets{P^{{\sf uniform},(S)}_{\sf FlipTest}}^{\frac{1}{2S-1}}\\
        \geq&\brackets{\frac{1}{2^S}P^{(S)}_{\sf FlipTest}}^{\frac{1}{2S-1}}\\
        \geq&\frac{1}{2}\brackets{\frac{1}{2^S}P^{(S)}_{\sf FlipTest}}^{\frac{1}{2S-1}}\\
        \geq&\frac{1}{2}P^{(1)}_{\sf FlipTest}.
    \end{align*}
\end{proof}
\subsection{Proof for \texorpdfstring{\Cref{lem:alternative_game_monotone}}{Lemma~\ref{lem:alternative_game_monotone}}}
\begin{proof}
    It is equivalent to prove for $k\geq 2$:
    \begin{equation*}
        S_{k-1}:=\frac{\sum_{i\in[D]}c_i^2p_i^{2k-1}}{\sum_{i\in[D]}c_i^2p_i^{2k-3}}\leq S_k:=\frac{\sum_{i\in[D]}c_i^2p_i^{2k+1}}{\sum_{i\in[D]}c_i^2p_i^{2k-1}}.
    \end{equation*}
    Define $\alpha_i=\frac{c_i^2p_i^{2k-3}}{\sum_{i\in[D]}c_i^2p_i^{2k-3}}$, we have $S_k=\sum_{i\in[D]}\alpha_ip_i^2$. Let $\beta_i=\frac{\alpha_ip_i^2}{\mu}$ where $\mu=\sum_{i\in[D]}\alpha_ip_i^2$.
    \begin{align*}
        \beta_i=&\frac{\alpha_ip_i^2}{\mu}\\
        =&\frac{c_i^2p_i^{2k-1}}{\sum_{i\in[D]}c_i^2p_i^{2k-3}\cdot\mu}\\
        =&\frac{c_i^2p_i^{2k-1}}{\sum_{i\in[D]}c_i^2p_i^{2k-3}\cdot\brackets{\sum_{i\in[D]}c_i^2p_i^{2k-1}/\brackets{\sum_{i\in[D]}c_i^2p_i^{2k-3}}}}\\
        =&\frac{c_i^2p_i^{2k-1}}{\sum_{i\in[D]}c_i^2p_i^{2k-1}}.
    \end{align*}
    Thus $S_k=\sum_i\beta_ip_i^2$. Now we prove that $\sum_{i\in[D]}\beta_ip_i^2\geq\sum_{i\in[D]}\alpha_ip_i^2$ instead, which is
    \begin{equation}\label{eq:variance_equation}
        \sum_{i\in[D]}\alpha_ip_i^4\geq\brackets{\sum_{i\in[D]}\alpha_ip_i^2}^2
    \end{equation}
    when multiple $\mu$ on both sides. Let $\mathbf{P}$ be a random variable that takes value $p_i^2$ w.p. $\alpha_i$. We have $\mathbb{E}[\mathbf{P}]=\sum_{i\in[D]}\alpha_ip_i^2$ and $\mathbb{E}[\mathbf{P}^2]=\sum_{i\in[D]}\alpha_ip_i^4$. \Cref{eq:variance_equation} follows from $\mathbf{Var}[\mathbf{P}]=\mathbb{E}[\mathbf{P}^2]-\mathbb{E}[\mathbf{P}]^2\geq 0$.
\end{proof}
\section{Fourier Tail-Bound Composition Problem}\label{sec:tail_bound_example}
We give a finite-dimensional counterexample to the following proposed inequality for Boolean functions:
if $W^{>C}[f]=W^{>C}[g]=\varepsilon$, must one have $W^{>2C}[f\cdot g]\leq 4\varepsilon$? The answer is no.  We construct explicit functions $f,g:\{0,1\}^{40}\to\{-1,+1\}$, both depending only on the Hamming weight modulo~$6$, and take $C=13$. Their Fourier tails are equal and satisfy
\[
    W^{>13}[f]=W^{>13}[g]=\frac{14304966732952173233665}{2^{76}},
\]
whereas
\[
    W^{>26}[f\cdot g]=\frac{14545481628327743687463}{2^{74}}>4W^{>13}[f].
\]
All Fourier coefficients are computed in closed form.  The same periodic construction also gives an asymptotic ratio tending to~$8$.
Let $f:\{0,1\}^n\to\mathbb{R}$, and write its Walsh--Fourier expansion as
\[
    f(x)=\sum_{S\subseteq[n]}\widehat f(S)\chi_S(x),\qquad \chi_S(x)=(-1)^{\sum_{i\in S}x_i}.
\]
For an integer $d\geq0$, define the Fourier weight above degree $d$ by
\[
    W^{>d}[f]=\sum_{\lvert S\rvert>d}\widehat f(S)^2.
\]
The map $x\mapsto z$, with $z_i=(-1)^{x_i}$, identifies this convention with the usual expansion on the sign cube.  Fourier levels, and therefore $W^{>d}$, are preserved.

The construction is elementary: both functions depend only on the Hamming weight modulo~$6$.  Its spectrum can therefore be evaluated by a one-parameter complex-exponential calculation.

\subsection{A Fourier formula for periodic Hamming-weight functions}

For $x=(x_1,\dots,x_n)\in\{0,1\}^n$, let
\[
    \abs{x}=x_1+\cdots+x_n.
\]
We begin with a standard factorization.

\begin{lemma}\label{lem:exponential}
    For $\theta\in\mathbb{R}$, define
    \[
        E_\theta(x)=e^{i\theta\abs{x}}.
    \]
    If $S\subseteq[n]$ and $\lvert S\rvert=j$, then
    \[
        \widehat E_\theta(S)=A_\theta^{\,n-j}B_\theta^{\,j},
        \qquad
        A_\theta=\frac{1+e^{i\theta}}{2},
        \quad
        B_\theta=\frac{1-e^{i\theta}}{2}.
    \]
\end{lemma}

\begin{proof}
    For each coordinate,
    \[
        e^{i\theta x_r}
        =\frac{1+e^{i\theta}}{2}
        +\frac{1-e^{i\theta}}{2}(-1)^{x_r}.
    \]
    Multiplying these identities over $r=1,\dots,n$ gives the Walsh expansion of $E_\theta$.  Choosing the character $(-1)^{x_r}$ exactly on the coordinates in $S$ yields the asserted coefficient.
\end{proof}

\subsection{The explicit counterexample in dimension forty}

Set
\[
    n=40,\qquad C=13.
\]
For $x\in\{0,1\}^{40}$, define
\begin{equation}
f(x)=
    \begin{cases}
        +1,&\abs{x}\bmod 6\in\{0,1,5\},\\
        -1,&\abs{x}\bmod 6\in\{2,3,4\},
    \end{cases}
    \label{eq:defF}
\end{equation}
and
\begin{equation}
    g(x)=
    \begin{cases}
        +1,&\abs{x}\bmod 6\in\{3,4,5\},\\
        -1,&\abs{x}\bmod 6\in\{0,1,2\}.
    \end{cases}
    \label{eq:defG}
\end{equation}
Their values on the six residue classes are displayed below.
    
\begin{center}
    \begin{tabular}{c@{\qquad}rrrrrr}
        \toprule
        $r=\abs{x}\bmod 6$ & 0 & 1 & 2 & 3 & 4 & 5\\
        \midrule
        $f$  & 1 & 1 & -1 & -1 & -1 & 1\\
        $g$  & -1 & -1 & -1 & 1 & 1 & 1\\
        $f\cdot g$ & -1 & -1 & 1 & -1 & -1 & 1\\
        \bottomrule
    \end{tabular}
\end{center}

\begin{theorem}\label{thm:main}
The functions $f$ and $g$ defined in \eqref{eq:defF}--\eqref{eq:defG} satisfy
\[
    W^{>13}[f]=W^{>13}[g]=:\eps,
\]
but
\[
    W^{>26}[f\cdot g]>4\eps.
\]
More precisely,
\begin{align*}
    \eps&=\frac{14304966732952173233665}{2^{76}}\approx 0.189324658315425,\\
    W^{>26}[f\cdot g]&=\frac{14545481628327743687463}{2^{74}}\approx 0.770031385804733,
\end{align*}
and hence
\[
    W^{>26}[f\cdot g]-4\eps=\frac{240514895375570453798}{2^{74}}\approx 0.0127327525430331>0.
\]
\end{theorem}

\subsubsection{The spectra of the two factors}

The residue tables in \eqref{eq:defF}--\eqref{eq:defG} are equivalently represented by
\begin{align}
    f(x)&=\frac{4}{3}\cos\!\left(\frac{\pi\abs{x}}{3}\right)-\frac{1}{3}(-1)^{\abs{x}},
    \label{eq:Ftrig}\\
    g(x)&=\frac{4}{3}\cos\!\left(\frac{\pi\abs{x}}{3}+\frac{2\pi}{3}\right)-\frac{1}{3}(-1)^{\abs{x}}.
    \label{eq:Gtrig}
\end{align}
Indeed, it suffices to check these identities for $\abs{x}=0,1,\dots,5$.

For $\theta=\pi/3$, Lemma~\ref{lem:exponential} gives
\[
    A_{\pi/3}=\frac{\sqrt{3}}{2}e^{i\pi/6},
    \qquad
    B_{\pi/3}=\frac{1}{2}e^{-i\pi/3}.
\]
Consequently, if $\lvert S\rvert=j$ and
\[
    \alpha_j=\frac{\pi(40-3j)}{6},
\]
then
\begin{align}
    \widehat f(S)&=\frac{4}{3}\,2^{-40}3^{(40-j)/2}\cos\alpha_j-\frac{1}{3}\mathds{1}[S=[40]],
    \label{eq:Fhat}\\
    \widehat g(S)&=\frac{4}{3}\,2^{-40}3^{(40-j)/2}\cos\!\left(\alpha_j+\frac{2\pi}{3}\right)-\frac{1}{3}\mathds{1}[S=[40]].
    \label{eq:Ghat}
\end{align}
Here the final terms come from $(-1)^{\abs{x}}=\chi_{[40]}(x)$. Since
\[
    \alpha_j\equiv \frac{2\pi}{3}-\frac{j\pi}{2}\pmod{2\pi},
\]
a direct check of $j\bmod4$ shows that, for $j<40$,
\begin{equation}
    \cos^2\alpha_j=\cos^2\!\left(\alpha_j+\frac{2\pi}{3}\right)=
    \begin{cases}
        \dfrac14,&j\ \text{even},\\[1mm]
        \dfrac34,&j\ \text{odd}.
    \end{cases}
    \label{eq:cosparity}
\end{equation}
At the top level $j=40$, both cosine factors in
\eqref{eq:Fhat}--\eqref{eq:Ghat} equal $-1/2$.  Thus
\[
    \widehat f(S)^2=\widehat g(S)^2\qquad\text{for every }S\subseteq[40],
\]
and in particular
\[
    W^{>13}[f]=W^{>13}[g].
\]

Let
\[
    W_j(f)=\sum_{\lvert S\rvert=j}\widehat f(S)^2.
\]
For $j<40$, equations \eqref{eq:Fhat} and \eqref{eq:cosparity} give
\begin{equation}
    W_j(f)=W_j(g)=
    \begin{cases}
        \displaystyle
        \frac{4}{9}\binom{40}{j}\frac{3^{40-j}}{4^{40}},
        &j\ \text{even},\\[3mm]
        \displaystyle
        \frac{4}{3}\binom{40}{j}\frac{3^{40-j}}{4^{40}},
        &j\ \text{odd}.
    \end{cases}
    \label{eq:FGlevels}
\end{equation}
At $j=40$,
\begin{equation}
    \widehat f([40])=\widehat g([40])=-\frac13\left(1+2^{-39}\right),
    \label{eq:toplevel}
\end{equation}
so the common tail is
\begin{align}
    \eps={}&\frac{4}{9\,4^{40}}\sum_{\substack{14\leq j\leq38\\
    j\ \mathrm{even}}}\binom{40}{j}3^{40-j}+\frac{4}{3\,4^{40}}\sum_{\substack{15\leq j\leq39\\
    j\ \mathrm{odd}}}\binom{40}{j}3^{40-j}\notag\\
    &\quad+\frac19\left(1+2^{-39}\right)^2.
    \label{eq:epssum}
\end{align}
Exact integer evaluation of the two finite sums yields
\begin{equation}
    \eps=\frac{14304966732952173233665}{2^{76}}.\label{eq:epsexact}
\end{equation}

\subsubsection{The spectrum of the product}

Let $h=f\cdot g$.  The residue table shows that
\[
    h(x)=
    \begin{cases}
        +1,&\abs{x}\equiv2\pmod3,\\
        -1,&\abs{x}\equiv0,1\pmod3,
    \end{cases}
\]
or equivalently
\begin{equation}
    h(x)=-\frac13+\frac43\cos\!\left(\frac{2\pi\abs{x}}{3}+\frac{2\pi}{3}\right).
    \label{eq:Htrig}
\end{equation}
For $\theta=2\pi/3$,
\[
    A_{2\pi/3}=\frac12e^{i\pi/3},\qquad B_{2\pi/3}=\frac{\sqrt3}{2}e^{-i\pi/6}.
\]
Applying Lemma~\ref{lem:exponential} to \eqref{eq:Htrig}, for every nonempty $S$ with $\lvert S\rvert=j$ we obtain
\begin{align}
    \widehat h(S)&=\frac43\,2^{-40}3^{j/2}\cos\!\left(\frac{\pi(2\cdot40-3j+4)}6\right)\notag\\
    &=\frac43\,2^{-40}3^{j/2}\cos\!\left(14\pi-\frac{j\pi}{2}\right)\notag\\
    &=\frac43\,2^{-40}3^{j/2}\cos\frac{j\pi}{2}.
    \label{eq:Hhat}
\end{align}
It follows that all odd Fourier levels of $h$ vanish, while for every even $j>0$,
\begin{equation}
    W_j(h)=\frac{16}{9}\binom{40}{j}\frac{3^j}{4^{40}}.
    \label{eq:Hlevels}
\end{equation}
Therefore
\begin{align}
    W^{>26}[h]&=\frac{16}{9\,4^{40}}\sum_{\substack{28\leq j\leq40\\
    j\ \mathrm{even}}}\binom{40}{j}3^j\notag\\
    &=\frac{14545481628327743687463}{2^{74}}.
    \label{eq:Hexact}
\end{align}
Combining \eqref{eq:epsexact} and \eqref{eq:Hexact},
\begin{align*}
    W^{>26}[f\cdot g]-4\eps&=\frac{14545481628327743687463-14304966732952173233665}{2^{74}}\\
    &=\frac{240514895375570453798}{2^{74}}>0.
\end{align*}
This proves Theorem~\ref{thm:main}.

\begin{remark}
The exact ratio in this example is
\[
    \frac{W^{>26}[f\cdot g]}{W^{>13}[f]}=\frac{4\cdot14545481628327743687463}{14304966732952173233665}\approx4.06725353504571.
\]
No claim of global minimality of the dimension $40$ is made here.
\end{remark}

\subsection{Asymptotic context: the same construction approaches eight}

The dimension-forty example is a member of a simple infinite family.  For $m\geq1$, take
\[
    n=6m+4,\qquad C=2m+1,
\]
and define $f$ and $g$ by the same residue tables modulo~$6$. The preceding Fourier formulas remain valid.  For the non-top levels of $f$, equation \eqref{eq:FGlevels} is bounded by a constant multiple of the probability mass function of
\[
    X\sim\operatorname{Bin}(n,1/4).
\]
Since $C+1=2m+2$ lies linearly above $\E X=(6m+4)/4$, a standard binomial concentration bound implies that the contribution from levels $C<j<n$ tends to zero exponentially.  The top Fourier coefficient tends to $-1/3$, and hence
\[
    W^{>C}[f]=W^{>C}[g]\longrightarrow\frac19.
\]

For $h=f\cdot g$, equation \eqref{eq:Hlevels} identifies the positive even-level weights with $16/9$ times the corresponding probabilities for
\[
    Y\sim\operatorname{Bin}(n,3/4).
\]
Moreover,
\[
    \Pr(Y\text{ is even})=\frac{1+(1-2\cdot\frac34)^n}{2}=\frac{1+2^{-n}}{2},
\]
because $n$ is even.  The threshold $2C=4m+2$ lies linearly below $\E Y=3(6m+4)/4$, so concentration gives
\[
    W^{>2C}[f\cdot g]\longrightarrow\frac{16}{9}\cdot\frac12=\frac89.
\]
Consequently,
\[
    \frac{W^{>2C}[f\cdot g]}{W^{>C}[f]}\longrightarrow8.
\]
Thus the failure at $n=40$ is not an isolated numerical accident; it belongs to a family whose tail-amplification ratio approaches~$8$.

\end{document}